\documentclass[a4paper,11pt]{article}
\pdfoutput=1 % if your are submitting a pdflatex (i.e. if you have
\usepackage{amsmath} 
\allowdisplaybreaks
\usepackage{jheppub} % for details on the use of the package, please
\usepackage{empheq}
\usepackage{bm}
\usepackage{slashed}
\usepackage[T1]{fontenc} % if needed
\usepackage[all]{xy}
\usepackage{rotating}
\usepackage{mathtools}
\usepackage{bbm}% mathbb
\usepackage{dcolumn}
\usepackage{multirow}
\usepackage{mathrsfs}% mathscr
\usepackage{arydshln}% dashed line
\usepackage{verbatim}% comment
\usepackage{graphics}
\usepackage{longtable}
\usepackage{tcolorbox}

\def\cT{\mathcal{T}}

\def\diag{\mathop{\rm diag}\nolimits}

\usepackage{slashed}

\def\a{\alpha}
\def\b{\beta}

\def\CD{{\cal D}}

\def\CH{{\cal H}}
\def\CI{{\cal I}}
\def\CJ{{\cal J}}
\def\CK{{\cal K}}
\def\CL{{\cal L}}

\def\CN{{\cal N}}
\def\CO{{\cal O}}

\def\CT{{\cal T}}

\def\CW{{\cal W}}

\def\BZ{{\mathbb Z}}

\def\beq#1\eeq{\begin{align}#1\end{align}}

\let\g=\gamma \let\d=\delta

\newcommand{\be}{\begin{equation}}
	\newcommand{\ee}{\end{equation}}
\newcommand{\ba}{\begin{align}}
	\newcommand{\ea}{\end{align}}%Verbatim cannot recognize this command. 

\newcommand{\bi}{\begin{itemize}}
	\newcommand{\ei}{\end{itemize}}

\makeatletter
\newcommand*{\rom}[1]{\expandafter\romannumeral #1}
\makeatother%For Roman Numeral

\usepackage{amsthm}
\theoremstyle{plain} % Italic body text
\newtheorem{theorem}{Theorem}[section]
\newtheorem{lemma}[theorem]{Lemma} % Lemmas share numbering with Theorems
\usepackage{rotating}
{\preprint{}}

\title{\boldmath Refined 3D index}

	\abstract{We introduce a refined version of the 3D index for 3-manifolds, building on the construction of the 3D $\mathcal{N}=2$ gauge theory $T[M]$ by Dimofte--Gaiotto--Gukov and Gang--Yonekura. The refined index is a superconformal index of $T[M]$ equipped with additional gradings that capture enhanced flavor symmetries of the effective theory. Our construction is based on a Dehn surgery presentation of $M$ in terms of an ideally triangulated link complement $N$. We derive an explicit infinite-sum formula for the refined index and provide nontrivial checks in representative examples, supporting its invariance under changes of triangulation, Dehn surgery presentation, and other auxiliary data. As a strictly stronger invariant, the refined index enables finer distinctions among 3-manifolds and among distinct IR phases of the associated gauge theories. We also introduce a computational tool, \textsc{Refined Index Calculator}, for its explicit evaluation.
    }
	
	\author[a]{Dongmin Gang,}
    \author[a]{Kibok Jeong,}
	\author[a]{Taeyoon Kim,}
    \author[a]{Soochang Lee}
	\affiliation[a]{
		Department of Physics and Astronomy $\&$ Center for Theoretical Physics,
		\\
		Seoul National University, 1 Gwanak-ro, Seoul 08826, Korea}
	\emailAdd{arima275@snu.ac.kr}
 	\emailAdd{boki0322@snu.ac.kr}
    \emailAdd{donpotleon@snu.ac.kr}
    \emailAdd{physicsmp1217@snu.ac.kr}
	
\begin{document} 
		\maketitle
		\flushbottom	
		
	\section{Introduction}
Quantum field theory and string/M-theory provide physical frameworks in which many mathematical objects are naturally realized. Building on this physical realization, a broad and systematic interplay between modern mathematics and quantum field theory/string theory has developed. The Jones polynomial is one such example, first discovered by mathematicians and later realized physically in $SU(2)$ Chern–Simons theory \cite{jones1985polynomial,kauffman1987state,zbMATH04092352}. Using this realization, the Jones polynomial was generalized to the colored Jones polynomial and the Witten–Reshetikhin–Turaev (WRT) invariants \cite{reshetikhin1991invariants}. The connection between the Jones polynomial and Chern–Simons theory gave rise to the subfield of quantum topology, which remains an active area of research.

The 3D index \cite{Dimofte:2011py,Gang:2018wek,Celoria:2025lqm} is another example of the relatively recent interplay between quantum field theory, string/M-theory, and low-dimensional topology. It was originally introduced as the superconformal index \cite{Kinney:2005ej,Kim:2009wb,Imamura:2011su} of the 3D $\CN=2$ gauge theory $T[M]$ associated with a 3-manifold $M$ \cite{Terashima:2011qi,Dimofte:2011ju}. Equivalently, through the 3D--3D correspondence, it can be identified with the partition function of $SL(2,\mathbb{C})$ Chern--Simons theory on $M$ at quantized Chern--Simons level $k=0$. The gauge theory $T[M]$ is expected to arise from the physics of two coincident M5-branes wrapped on the 3-manifold.  Its IR behavior depends only on the choice of the 3-manifold, although its UV gauge theory description depends on additional auxiliary data, such as a Dehn surgery presentation and choices of ideal triangulations. This provides a systematic way of understanding 3D $\CN=2$ dualities from topological operations on 3-manifolds, such as Pachner 2--3 moves and Kirby moves. Since the superconformal index is an RG-invariant quantity and depends only on the IR physics of the gauge theory, the 3D index is expected to be a 3-manifold invariant. Since its introduction, the 3D index has also been found to arise in a variety of other contexts, such as normal surface counting \cite{GHHR16} and the resurgence of complex Chern--Simons theory \cite{Garoufalidis:2020nut}.

In this work, we propose a refined version of the 3D index. As emphasized in \cite{Gang:2018wek}, the 3D gauge theories \(T[M]\) can exhibit additional accidental symmetries that are not directly visible in the M5-brane construction. The original 3D index is the superconformal index of \(T[M]\), computed using only the symmetries manifest in the M5-brane description. Our refined 3D index incorporates additional quantum numbers associated with these accidental symmetries. There are two main sources of such additional symmetries. The first arises from an insufficient number of admissible superpotential deformations in the effective 3D theory associated with a link complement \(N\), which is introduced in the Dehn surgery representation of \(M\) in \eqref{Dehn filling rep of M}. For a given gauge group and Chern--Simons levels of the effective theory, the possible superpotential deformations compatible with the \(U(1)_R\) symmetry in \eqref{R_geo}, inherited from the wrapped M5-brane system, can be analyzed using field-theoretic methods. Even after turning on all such allowed deformations, the resulting theory may still possess an additional symmetry. The second source originates from a refinement of the Dehn filling operation with non-integer slope. In this construction, the Dehn filling corresponds to coupling the \(\CN=2\) gauge theory associated with a link complement to  \(\CN=4\) SCFTs determined by the Dehn filling slopes. The \(\CN=4\) supersymmetry is broken to \(\CN=2\) by the coupling; however, the Cartan subalgebra of the \(SO(4)\) R-symmetry is preserved. One linear combination is identified with the \(U(1)\) R-symmetry, while the other, denoted by \(U(1)_A\), gives rise to an additional refinement. These two types of refinements are intertwined under changes of the Dehn surgery representation.
\paragraph{Why refinements?}
\begin{enumerate}
  \item \textbf{Stronger invariants.}  
The refined index can distinguish 3-manifolds that are indistinguishable by the unrefined index, and thus provides a strictly more informative invariant. For example, the 3D index for some non-hyperbolic 3-manifolds becomes extremely simple. In particular, for Seifert-fibered manifolds with three exceptional fibers, the 3D index reduces to the trivial values $1$ or $2$ and therefore cannot distinguish among them. The refined 3D index, however, is not necessarily trivial for such manifolds and can distinguish some of them, although it still does not provide a complete invariant for non-hyperbolic 3-manifolds.

 \item \textbf{Regularization.}  
For some non-hyperbolic 3-manifolds, the unrefined 3D index does not converge as a formal power series in $q^{1/2}$. In some cases, the refined index resolves this non-convergence by incorporating additional quantum numbers into the counting.

\item \textbf{Better probing of the IR physics of $T[M]$.} For some non-hyperbolic 3-manifolds, such as the Seifert-fibered manifolds mentioned above, the 3D index becomes trivial. There are two possibilities. One is that the $T[M]$ theory has a mass gap and flows to a unitary TQFT in the IR. The other is that the $T[M]$ theory experiences supersymmetry enhancement to a $\CN=4$ superconformal field theory, specifically a rank-0 SCFT \cite{Gang:2018huc,Gang:2021hrd}. In this case, the 3D index computes the superconformal index at the A- or B-twisting point, which counts Coulomb- or Higgs-branch operators. The refined 3D index distinguishes between these two distinct IR phases.
\end{enumerate}
One notable feature of the refined 3D index is that the number of refinements depends on the 3-manifold. In the examples studied in this work, we find that the number of independent refinements ranges from $0$ to $3$. It is natural to expect that there exist 3-manifolds admitting a larger number of refinements.
Furthermore, our construction depends on several auxiliary choices, such as the Dehn surgery representation and the choice of ideal triangulation. In the construction of the 3D gauge theory, each such choice corresponds to a distinct IR duality frame. The number of refinements can vary depending on these choices. This reflects the fact that the manifest UV flavor symmetry can differ between duality frames. Nevertheless, among all such choices, there typically exists a distinguished one that maximizes the number of available refinements. This choice corresponds to a duality frame in which the full IR flavor symmetry, at least its Cartan subalgebra, is manifest in the UV.  For other choices, the corresponding refined index can be obtained from that of the distinguished one by turning off a subset of the refinement parameters. 
A more precise formulation of these statements is given in the main conjecture \eqref{main conjecture}.

The rest of the paper is organized as follows. In section~\ref{sec : Gauge theory TM-main}, we review—along with some clarifications and improvements—the construction of the 3D gauge theory associated with a 3-manifold, with particular emphasis on the structure of its flavor symmetry.
In section~\ref{sec : refined 3D index-main}, which constitutes the main part of the paper, we present a mathematical expression for the refined 3D index and formulate our main conjecture \eqref{main conjecture}. We also reformulate this expression as a refined counting of normal surfaces, and prove that the refined 3D index is invariant under Pachner $2\text{--}3$ moves of of the ideal triangulation. 
In section~\ref{sec : examples}, we provide explicit examples of refined 3D index computations and verify the conjecture in these cases.
In the appendices, we collect various technical details and include a brief introduction to the application \textsc{Refined Index Calculator} (with a link for download), which computes the refined 3D index.

    \section{3-Manifolds and 3D Gauge Theories} \label{sec : Gauge theory TM-main}
    In this section, we review the construction of 3D gauge theories labeled by 3-manifolds in \cite{Dimofte:2011ju,Gang:2018wek}, with emphasis on their flavor symmetries. In the next section, we construct a refined 3D index that computes the superconformal index of the theory in the presence of additional flavor fugacities.
    
    \subsection{3-Manifolds from Dehn Filling}
Let $M$ be a $3$-manifold obtained from a  $3$-manifold $N$ whose boundary consists of $(n+d)$ torus components,
\begin{align}
\partial N
&=
T^2_1 \sqcup \cdots \sqcup T^2_n
\sqcup T^2_{n+1} \sqcup \cdots \sqcup T^2_{n+d}\;,
\end{align}
by performing Dehn filling along the last $d$ boundary components:
\begin{align}
\mathcal{D}\;:\;
M
=
N_{[P_1\gamma_1 + Q_1\delta_1],\,\ldots,\,[P_d\gamma_d + Q_d\delta_d]}
\;.
\label{Dehn filling rep of M}
\end{align}
We denote a basis of
\(
H_1(\partial N;\mathbb{Z})
=
\bigoplus_{i=1}^{n+d} H_1(T_i^2;\mathbb{Z})
\)
as follows:
\begin{align}
\begin{split}
&(\alpha_i,\beta_i)_{i=1,\ldots,n}
\;:\;
\text{a primitive basis of }
H_1(T_i^2;\mathbb{Z}),
\\
&(\gamma_I,\delta_I)_{I=1,\ldots,d}
\;:\;
\text{a primitive basis of }
H_1(T_{n+I}^2;\mathbb{Z}).
\end{split}
\end{align}
For each $I=1,\ldots,d$, the primitive class
\begin{align}
[P_I\gamma_I + Q_I\delta_I]
\in H_1(T^2_{n+I};\mathbb{Z})/\mathbb{Z}_2
\end{align}
specifies the slope along which the Dehn filling is performed.

\subsection{3D $T[M;\vec{\alpha} ]$ theory}
  We present an explicit field-theoretic construction of 3D supersymmetric gauge theories labeled by a 3-manifold \(M\) and a collection of primitive boundary 1-cycles \(\vec{\alpha} = (\alpha_1, \ldots, \alpha_n)\), following \cite{Dimofte:2011ju,Gang:2018wek,Gang:2025ykf}.

  \subsubsection{$T[M;\vec{\a}]$ from 6D  theory}
  Before presenting a bottom-up 3D field-theoretic description of the theory 
 $T[M;\vec{\alpha}]$, we first give a top-down definition in terms of the compactification of a 6D theory. This 6D compactification picture provides useful insight into the IR physics of the 3D theory.
  
For a closed three-manifold $\widetilde{M}$ and knots
$K_1,\ldots,K_n \subset \widetilde{M}$, we define $\mathbf{T}[\widetilde{M}; K_1,\ldots,K_n]$ as
\begin{align}
\begin{split}
&\Bigl(
\text{6D $A_1$ $(2,0)$ theory on $\mathbb{R}^{1,2}\times \widetilde{M}$ with regular BPS}
\\
&\quad
\text{codimension-two defects along }
\mathbb{R}^{1,2}\times K_1,\,\ldots,
\mathbb{R}^{1,2}\times K_{n}
\Bigr)
\\
&\xrightarrow{\;\; \mathrm{size}(\widetilde{M})\to 0\;\;}
\text{3D theory }
\mathbf{T}[\widetilde{M}; K_1,\ldots,K_n]
\text{ on }\mathbb{R}^{1,2}\; . \label{T[M] from 6D}
\end{split}
\end{align}
The 6D theory describes the low-energy world-volume theory of two coincident M5-branes in M-theory, for which no explicit Lagrangian description is known. In the dimensional reduction, we perform a partial topological twist along the
internal manifold $\widetilde{M}$ using an $SO(3)$ subgroup of the $SO(5)$
R-symmetry of the six-dimensional theory. This twisting preserves one quarter of the supersymmetry, and the resulting
three-dimensional theory generically has $\CN=2$ supersymmetry with
$U(1)_R \cong SO(2)$ R-symmetry.
The R-symmetry originates from the $SO(2)$ subgroup of the $SO(5)$
R-symmetry of the six-dimensional theory that survives the topological
twisting. In the 3D effective field theory, this R-symmetry can mix
with other abelian flavor symmetries. We denote the $SO(2)$ R-symmetry
descending from the six-dimensional theory by $SO(2)_{\rm geo}$ and denote its generator (or charge) by
\begin{align}
R_{\rm geo}
:\; \text{generator of the } SO(2)_{\rm geo}
\subset SO(2)\times SO(3) \subset SO(5)
\label{R_geo}
\end{align}
The resulting 3D \(\CN=2\) theory also exhibits an \(SU(2)^n\) flavor symmetry associated with the regular codimension-two defects, with corresponding moment maps \(\boldsymbol{\mu}_{i=1,\ldots,n}\) transforming in the adjoint representation. As in the case of the 6D theory itself, whose intrinsic description is rather subtle but becomes more tractable upon reduction on \(S^1\) to 5D maximally supersymmetric Yang--Mills theory, the regular BPS codimension-two defects can likewise be more clearly understood after such an \(S^1\) reduction. Under this compactification, each codimension-two defect is realized by coupling a 3D \(\mathcal{N}=4\) \(T[SU(2)]\) theory \cite{Gaiotto:2008ak} to the 5D bulk. The \(T[SU(2)]\) theory has \(SU(2)_H \times SU(2)_C\) flavor symmetry: the \(SU(2)_H\) factor is gauged by the bulk fields, while the remaining \(SU(2)_C\) provides the \(SU(2)\) flavor symmetry associated with the defect. The moment maps \(\boldsymbol{\mu}\) originate from those of the \(SU(2)_C\) symmetry of the \(T[SU(2)]\) theory.

Since the resulting 3D theory has only $\CN=2$ supersymmetry, the existence of a
moment map for each $SU(2)$ flavor symmetry is not guaranteed; we will give a
criterion for their existence below. When absent, the corresponding moment map
is taken to vanish. We then define
\begin{align}
\begin{split}
T[\widetilde{M}; K_1,\ldots,K_{n}]
\;:=\;
\mathbf{T}[\widetilde{M}; K_1,\ldots,K_{n}]
\;\text{deformed by the superpotential}\;
\CW = \sum_{i=1}^{n} \mu_i^{3}\;, \label{T[M] from 6D-2}
\end{split}
\end{align}
where $\mu_i^{3}$ denotes the Cartan component of the moment map
$\boldsymbol{\mu}_i = \sum_{a=1}^{3} \mu_i^{a} t^{a}$, with $\{t^{a}\}$ the
generators of $\mathfrak{su}(2)$.
Upon this deformation, the $SU(2)_i$ flavor symmetry is broken to $U(1)_i$
whenever the moment map $\boldsymbol{\mu}_i$ exists. In \cite{Gang:2018wek}, it is argued that the theory $T[M; \vec{\alpha}]$  corresponds to
\begin{align}
\begin{split}
&T[M; \alpha_1, \ldots, \alpha_n]
= T[\widetilde{M}; K_1,\ldots,K_{n}]\;,
\\
&\text{where the two sets of topological data, $(M,\vec{\a})$ and $(\widetilde{M};\vec{K})$, are related by}
\\
&\widetilde{M}
= M_{[\a_1],\ldots,[\a_{n}]}
\quad\text{and}\quad
\widetilde{M}\setminus
\Bigl(\bigcup_{i=1}^{n} K_i\Bigr)
= M\;. \label{T[M] from 6D-3}
\end{split}
\end{align}
This relation states that $\widetilde{M}$ is obtained from $M$ by Dehn filling
along the slopes $\{[\a_i]\}$, and that removing the link
$L=\bigcup_{i=1}^{n}K_i$ from $\widetilde{M}$ recovers the original manifold
$M$.

\paragraph{Non-closable cycles and $SU(2)$ symmetries}

In \cite{Gang:2018wek}, it was argued that the $SU(2)_i$ moment map does not
exist—and hence the $SU(2)_i$ symmetry is not broken to $U(1)_i$—when the
corresponding boundary one-cycle $[\alpha_i]$ in $M$ is {\it non-closable}:
\begin{align}
\textrm{$\alpha_i$ is non-closable} \;\Rightarrow\; \textrm{$\boldsymbol{\mu}_i$ is absent} \;\Rightarrow\; \textrm{$SU(2)_i$ is not broken to $U(1)_i$}\;.
\label{non-closability and SU(2)}
\end{align}
A primitive boundary one-cycle $\alpha$ is called \emph{non-closable} if the
3D index of the three-manifold obtained by Dehn filling along $\alpha$ vanishes:
\begin{align}
\CI_{M_{[\alpha]}} = 0 \, . \label{def:non-closable}
\end{align}
Here $M_{[\alpha]}$ denotes the three-manifold obtained from $M$ by Dehn
filling along the slope $[\alpha]$. The 3D index $\CI$ will be introduced in
Section~\ref{sec : unrefined 3D index}, where it computes the superconformal index of the theory $T[M;\vec{\alpha}]$.
This condition is expected to be equivalent to
\begin{align}
\CI_{M_{[\alpha]}} = 0
\;\Longleftrightarrow\;
\chi_{\rm irred}[M_{[\alpha]}] = \emptyset \; .
\end{align}
In the 3D--3D correspondence, the set $\chi_{\rm irred}[M_{[\alpha]}]$ is related to the set
of Bethe vacua of the theory $T[M_{[\alpha]}]$.
The emptiness of this set implies spontaneous supersymmetry breaking, which in
turn leads to the vanishing of the superconformal index.
Thus, a non-closable cycle is one that cannot be closed without breaking
supersymmetry. Here $\chi_{\rm irred}[M]$ denotes the moduli space of irreducible
$SL(2,\mathbb{C})$ flat connections on $M$,
\begin{align}
\chi_{\rm irred}[M]
:=
\Bigl\{
R \in \mathrm{Hom}\!\bigl(\pi_1(M), SL(2,\mathbb{C})\bigr)
\ \Big|\ 
R \text{ is adjoint-irreducible}
\Bigr\}
\big/ \sim \; ,
\end{align}
where $\sim$ denotes conjugation by $SL(2,\mathbb{C})$.
A representation $R$ is called \emph{adjoint-irreducible} if the centralizer
of $\mathrm{Im}(R)$ in $SL(2,\mathbb{C})$ is equal to the center,
\begin{align}
\mathrm{Cent}\bigl(\mathrm{Im}(R)\bigr)
=
\{\pm I\} \subset SL(2,\mathbb{C})\; .
\end{align}
The simplest examples with empty $\chi_{\rm irred}$ are the lens spaces
$L(P,Q)$, whose fundamental group is $\mathbb{Z}_P$ or $\mathbb{Z}$.
Further examples include the unknot complement and the Hopf link complement.

Accordingly, from the 6D compactification perspective, the theory $T[M;\vec{\alpha}]$ is expected to have the following flavor symmetry:
\begin{align}
\begin{split}
&\prod_{i=1}^n  \left( 
\begin{cases}
SU(2)_i, & \text{if } \alpha_i \text{ is non-closable}, \\
U(1)_i, & \text{otherwise}
\end{cases} \right)
\\
&=
\left( \prod_{\alpha_i \,\text{: non-closable}} SU(2)_i \right)
\times
\left( \prod_{\alpha_i \,\text{: closable}} U(1)_i \right)\;. \label{flavor symmetry of T[M] from 6D}
\end{split}
\end{align}
However, the actual IR symmetry of the 3D theory may differ, as some symmetries decouple while others emerge accidentally in the IR.

    \subsubsection{UV gauge theory  for $T[N;\vec{\a}, \vec{\g}]$}
    As a first step in constructing the field theory $T[M;\vec{\alpha}]$ using the Dehn surgery representation in \eqref{Dehn filling rep of M}, we construct $T[N;\vec{\alpha},\vec{\beta}]$ based on an ideal triangulation of $N$. The theory $T[N;\vec{\alpha},\vec{\beta}]$ admits the same 6D compactification picture as $T[M;\vec{\alpha}]$, as reviewed in the previous subsection, with $M$ replaced by $N$, which has $(n+d)$ torus boundary components instead of $n$, and with the boundary primitive 1-cycles $\vec{\alpha}$ extended to $(\vec{\alpha},\vec{\gamma})$.
  
\paragraph{An ideal triangulation $\mathcal{T}$ and gluing equations for $N$}
Consider an ideal triangulation of $N$ consisting of $r$ ideal tetrahedra \cite{thurston2022geometry},
\begin{align}
\mathcal{T}\;:\;
N
=
\left(\bigcup_{a=1}^{r} \Delta_a\right)\big/\sim \; .
\label{Ideal triangulation of N}
\end{align}
For each ideal tetrahedron $\Delta_a$, we associate a boundary phase space
\begin{align}
P(\partial \Delta_a)
=
\{ Z_a,\, Z_a',\, Z_a'' \;\mid\; Z_a + Z_a' + Z_a'' = \mathbf{i}\pi \},
\end{align}
parametrized by the logarithmic edge parameters $\{Z_a, Z_a', Z_a''\}$. Here we define
\begin{align}
\mathbf{i} = \sqrt{-1}\;.
\end{align}
This phase space is endowed with the symplectic structure
\begin{align}
\{ Z_a, Z_a'' \}_{\rm P.B.}
=
\{ Z_a', Z_a \}_{\rm P.B.}
=
\{ Z_a'', Z_a' \}_{\rm P.B.}
=
1 \; .
\end{align}
The gluing equations consist of the quantities
$C_a$, $hol(\alpha_i)$, $hol(\beta_i)$, $hol(\gamma_I)$, and
$hol(\delta_I)$, each of which is given by a linear combination of the edge
parameters with integer coefficients.
Here $C_a$ denotes the internal edge variable  associated with the
$a$-th internal edge, given by the sum of the edge parameters meeting at that
edge. The condition $C_a = 2\pi  \mathbf{i}$ corresponds to the requirement that the
total dihedral angle around each internal edge is $2\pi$.
The quantities $hol(\alpha_i)$ and $hol(\beta_i)$ denote the logarithmic
holonomies associated with the boundary $1$-cycles
$\alpha_i,\beta_i \in H_1(T^2_i;\mathbb{Z})$ on the $i$-th boundary torus.
Similarly, $hol(\gamma_I)$ and $hol(\delta_I)$ denote the logarithmic
holonomies associated with the basis cycles
$\gamma_I,\delta_I \in H_1(T^2_{n+I};\mathbb{Z})$ on the boundary tori that are
eventually Dehn filled.
All these quantities are given by integer linear combinations of the edge
parameters $\{Z_a,Z_a',Z_a''\}$ of the ideal tetrahedra.

These gluing equations are known to enjoy the following symplectic structure \cite{neumann1985volumes}:
\begin{align}
\begin{split}
&\{ hol(\alpha_i), hol(\beta_j) \}_{\rm P.B.}
= 2\,\delta_{ij}\,\langle \alpha_i, \beta_j \rangle \;,
\\
&\{ hol(\gamma_I), hol(\delta_J) \}_{\rm P.B.}
= 2\,\delta_{IJ}\,\langle \gamma_I, \delta_J \rangle \;,
\\
&\text{all other Poisson brackets} = 0 \; .
\end{split}
\end{align}
Here $\langle\,\cdot\,,\,\cdot\,\rangle$ denotes the oriented intersection number of
$1$-cycles on the torus.

    \paragraph{Hard and easy internal edges}  Let $D$ be a linear combination of edge parameters of the form
\begin{align}
D = \sum_{a=1}^r \bigl( f_a Z_a + g_a Z_a' + h_a Z_a'' \bigr),
\qquad
f_a, g_a, h_a \in \mathbb{Z}_{\ge 0}\;.
\end{align}
We call $D$ an \emph{internal edge} if
\begin{align}
\begin{split}
&D =
2\pi \mathbf{i}
\quad
\Bigl(
\mathrm{mod}\;
\bigl\langle
C_a - 2\pi \mathbf{i},\;
Z_i + Z_i' + Z_i'' - \mathbf{i}\pi
\;\big|\;
a,i=1,\ldots,r
\bigr\rangle
\Bigr),
\\
&\{D, C_a\}_{\rm P.B.} =\{D, hol(\alpha_i)\}_{\rm P.B.}
=
\{D, hol(\beta_i)\}_{\rm P.B.} 
\\
&= \{D, hol(\gamma_I)\}_{\rm P.B.}
=
\{D, hol(\delta_I)\}_{\rm P.B.} 
=0.
\end{split}
\end{align}
An internal edge $D$ is called \emph{easy} if it satisfies
\begin{align}
\text{for each } a=1,\ldots,r,\quad
\text{at most one of } (f_a, g_a, h_a) \text{ is nonzero}.
\end{align}
Let $\{E_I\}_{I=1}^{\sharp_E}$ denote a maximal set of easy internal edges, and
let $\{H_I\}_{I=1}^{\sharp_H := r-n-d-\sharp_E}$ denote a set of hard internal edges
that are linearly independent of the easy internal edges.

    \paragraph{Refined Neumann-Zagier matrices $g_{\rm NZ} $} The matrix is defined using the following relation:
   \begin{align}
\left(
\begin{array}{c}
\mathbf{X} \\
\mathbf{H} -2\pi \mathbf{i} \\
\mathbf{E} -2\pi \mathbf{i} \\
\hline
\mathbf{P} \\
\mathbf{\Gamma}_{H} 
\\ 
\mathbf{\Gamma}_{E} 
\end{array}
\right) = g_{\rm NZ} \cdot \left(
\begin{array}{c}
Z_1 \\
\vdots \\
Z_{r} \\ \hline
Z''_1 \\
\vdots \\
Z''_r \\
\end{array}
\right) + \mathbf{i} \pi \left(
\begin{array}{c}
\boldsymbol{\nu}_x \\ \hline
\boldsymbol{\nu}_p
\end{array}
\right) \label{NZ matrices}
\end{align}
Here $\mathbf{X} = (X_1,\ldots, X_{n+d})^T,  \mathbf{P} = (P_1,\ldots, P_{n+d})^T$ with 
\begin{align}
X_i = hol(\alpha_i) , \;\;  X_{n+I} =hol(\g_I), \;\;P_i = \frac{hol(\b_i)}{2} , \;\;  P_{n+I} = \frac{hol(\d_I)}{2} \;.
\end{align}
The basis are  chosen such that
\begin{align}
\begin{split}
&\;\langle \alpha_i,\beta_i\rangle \;=\; \langle \gamma_I,\delta_I\rangle \;=\; 1\;. \label{basis (a,b,g,d)}
\end{split}
\end{align}
The vectors
$\mathbf{H}=(H_1,\ldots,H_{\sharp_H})^T$ and
$\mathbf{E}=(E_1,\ldots,E_{\sharp_E})^T$
denote the hard and easy internal edges, respectively.
Their conjugate variables are
$\boldsymbol{\Gamma}_H=(\Gamma_{H_1},\ldots,\Gamma_{H_{\sharp_H}})^T$ and
$\boldsymbol{\Gamma}_E=(\Gamma_{E_1},\ldots,\Gamma_{E_{\sharp_E}})^T$, satisfying
the Poisson bracket relations
\begin{align}
\begin{split}
&\{ H_I, \Gamma_{H_J} \}_{\rm P.B.} = \delta_{IJ}, \qquad
\{ E_I, \Gamma_{E_J} \}_{\rm P.B.} = \delta_{IJ},
\\
&\text{all other Poisson brackets} = 0 \; .
\end{split}
\end{align}
The matrix $g_{\rm NZ}$ is symplectic, $g_{\rm NZ} \in Sp(2r,\mathbb{Q})$, with half-integer entries. The conventional Neumann--Zagier matrices can be constructed in the same way, except that $(\mathbf{H}, \mathbf{E})$ are chosen to be linearly independent internal edges without regard to the easy/hard classification. Our refined NZ matrix will be used to define the refined 3D index, while the conventional one defines the (unrefined) 3D index.

   \paragraph{Gauge theory $T[N;\vec{\alpha},\vec{\gamma}]$}
Using the Neumann--Zagier (NZ) matrix, the theory is defined as \cite{Dimofte:2011ju}
\begin{align}
T[N;\vec{\alpha},\vec{\gamma}]
:=
\left(
g_{\rm NZ} \cdot \bigl[(T_\Delta)^{\otimes r}\bigr]
\;\text{deformed by the superpotential }\;
\CW = \sum_{I=1}^{\sharp_E} \CO_{E_I}
\right)\; . \label{T[N] theory}
\end{align}
Here $T_\Delta$ denotes the free chiral theory with background Chern--Simons
level $-\frac{1}{2}$ for the $U(1)$ flavor symmetry. In terms of superfields,
its Lagrangian density is given by
\begin{align}
\CL_{T_\Delta}(V)
= \int d^4\theta \left(
- \frac{1}{8\pi}\Sigma V + \Phi^\dagger e^{V}\Phi
\right)\,,
\end{align}
where $\Sigma$ is the field strength of the vector multiplet $V$
associated with the $U(1)$ flavor symmetry. $(T_\Delta)^{\otimes r}$ denotes $r$ copies of the $T_\Delta$ theory, whose
Lagrangian density is given by
\begin{align}
\CL_{(T_\Delta)^{\otimes r}}(V_1,\ldots,V_r)
= \sum_{i=1}^{r} \CL_{T_\Delta}(V_i)\,,
\end{align}
where $\{V_i\}_{i=1}^r$ denote the vector multiplets associated with the
$U(1)^r$ flavor symmetry. Using the $U(1)^r$ flavor symmetry, one can consider the action of
$Sp(2r,\mathbb{Q})$ on the theory $(T_\Delta)^{\otimes r}$.
We denote by $g_{\rm NZ}\cdot (T_\Delta)^{\otimes r}$ the theory obtained
by acting with $g_{\rm NZ}\in Sp(2r,\mathbb{Q})$. To be more explicit, one decomposes $g_{\rm NZ}$ into a product of ``T-type'' ($g_K^{t}$),
``S-type'' ($g_J^{s}$), and ``GL-type'' ($g_U^{gl}$) matrices:
\begin{align}
g_K^{t} &:= \begin{pmatrix} I & 0 \\ K & I \end{pmatrix}, \qquad
g_J^{s} := \begin{pmatrix} I - J & -J \\ J & I - J \end{pmatrix}, \qquad
g_U^{gl} := \begin{pmatrix} U & 0 \\ 0 & (U^{-1})^{t} \end{pmatrix}.
\end{align}
Here $K$ is a symmetric $r\times r$ matrix with integer entries,
and $J$ is a diagonal matrix with entries $0$ or $1$.
The matrix $U$ is a general element of $GL(r,\mathbb{Q})$.
Note that $g_K^t$ and $g_J^s$ are elements of $Sp(2r,\mathbb{Z})$,
while $g_U^{ gl}$ is, in general, an element of $Sp(2r,\mathbb{Q})$.
The field-theoretic action of these basic types is given by
\begin{align}
\begin{split}
&\CL_{g^t_K \cdot T}  (\vec{V}):=
\CL_T(\vec{V})
+ \frac{1}{4\pi}\int d^4\theta \sum_{i,j=1}^r \Sigma_i K_{ij} V_j \,,
\\
&\CL_{g^s_J \cdot T} (\vec{V}) :=
\CL_T\!\left((I-J)\cdot \vec{V}+J\cdot \vec{W}\right)
+ \frac{1}{2\pi}\int d^4\theta \sum_{i=1}^r \Sigma_i J_{ii} W_i \,,
\\
&\CL_{g^{gl}_U \cdot T}(\vec{V}) :=
\CL_T(U^{-1}\cdot \vec{V}) \,,
\end{split}
\end{align}
where $\vec{W}=(W_1,\ldots,W_r)$ has non-zero components only for indices
with $J_{ii}=1$. These vector multiplets are dynamical fields and are
integrated out in the path integral. Thus, the $g^s_J$ action amounts to
gauging the $i$-th $U(1)$ symmetry whenever $J_{ii}=1$, and the vector
multiplet associated with the resulting topological $U(1)$ symmetry is
denoted by $V_i$. The theory $g_{\rm NZ}\cdot\bigl[(T_\Delta)^{\otimes r}\bigr]$ admits (gauge-invaraint) chiral
primary operators $\CO_{E_I}$ associated with each easy internal edge \cite{Dimofte:2011ju,Gang:2025ykf}.
The theory $T[N;\vec{\alpha},\vec{\gamma}]$ is obtained from
$g_{\rm NZ}\cdot\bigl[(T_\Delta)^{\otimes r}\bigr]$ by a superpotential
deformation using these operators.

The gauge theory $T[N;\vec{\alpha},\vec{\gamma}]$ depends on the choice of an ideal triangulation
$\mathcal{T}$ and on the choice of ``quad'' structures (i.e.\ the assignment
of edge parameters to each ideal tetrahedron in $\mathcal{T}$).
For simplicity, we suppress these dependences and do not explicitly indicate
them in the notation $T[N;\vec{\alpha},\vec{\gamma}]$. Different choices generally lead to different UV gauge theories, but they are
known to flow to the same IR fixed point \cite{Dimofte:2011ju}. 

\paragraph{Flavor symmetry of $T[N;\vec{\alpha},\vec{\gamma}]$}
The theory $g_{\rm NZ} \cdot \bigl[(T_\Delta)^{\otimes r}\bigr]$ has a
$U(1)^r$ flavor symmetry, since the $Sp(2r,\mathbb{Q})$ action preserves the number of $U(1)$ flavor symmetries. The Cartan generators can be labeled by the position variables ($\mathbf{X}$, $\mathbf{H}$, and $\mathbf{E}$) appearing in the symplectic transformation in \eqref{NZ matrices}; we denote them by
\begin{align}
\text{Cartan generators of } U(1)^r \;:\;
T_{x_1},\ldots,T_{x_{n+d}},
\;
T_{h_1},\ldots,T_{h_{\sharp_H}},
\;
T_{e_1},\ldots,T_{e_{\sharp_E}} \; .
\end{align}
The superpotential term $\sum \CO_{E_I}$ carries nontrivial charges only under
$T_{e_I}$ $(I=1,\ldots,\sharp_E)$, and consequently the corresponding
$U(1)^{\sharp_E} \subset U(1)^r$ flavor symmetry is broken.
As a result, the theory $T[N;\vec{\alpha},\vec{\gamma}]$ has a residual
UV flavor symmetry
\begin{align}
\begin{split}
U(1)^{n+d+\sharp_H}
&= U(1)^{r-\sharp_E} 
\\
&= U(1)_{x_1} \times \cdots \times U(1)_{x_{n+d}}
\times U(1)_{h_1} \times \cdots \times U(1)_{h_{\sharp_H}}\;,
\end{split}
\label{flavor of T[N]}
\end{align}
whose Cartan subalgebra $\mathfrak{F}[N]$ decomposes as
\begin{align}
\begin{split}
&\mathfrak{F}[N]  = \mathfrak{F}^{\rm tori}[N] \oplus \mathfrak{F}^{\rm hard}[N]\;,
\\
&\mathfrak{F}^{\rm tori}[N]
=\mathrm{Span}\{T_{x_1},\ldots,T_{x_{n+d}}\}\;,\\
&\mathfrak{F}^{\rm hard}[N]
=\mathrm{Span}\{T_{h_1},\ldots,T_{h_{\sharp_H}}\}\;.
\end{split}
\label{Cartans of T[N]}
\end{align}
In the 6D compactification picture, the Cartan algebra
$\mathfrak{F}^{\rm tori}[N]$ in \eqref{Cartans of T[N]} can be identified with
the Cartan subalgebra of the flavor symmetry in
\eqref{flavor symmetry of T[M] from 6D}. On the other hand, from the
six-dimensional reduction viewpoint, $\mathfrak{F}^{\rm hard}[N]$ in
\eqref{Cartans of T[N]} can be regarded as the Cartan subalgebra of an
accidental symmetry that emerges when the size of the internal manifold
shrinks to zero. These accidental symmetries provide one source of the refinement.

    \subsubsection{UV gauge theory for $T[M; \vec{\a}] $ } \label{sec : Gauge theory TM}
   We now give a bottom-up field-theoretic description of $T[M;\vec{\alpha}]$ by performing the field-theoretic counterpart of the Dehn filling operation on $T[N;\vec{\alpha},\vec{\gamma}]$. For this construction to work, we require that
\begin{align}
\gamma_I \ \text{is non-closable}\;, \quad I=1,\ldots, d\;,
\label{basis (a,b,g,d)-2}
\end{align}
in addition to the condition in \eqref{basis (a,b,g,d)} on the basis
$(\gamma_I,\delta_I)$.
In this case, for each such $I$, the theory exhibits an $SU(2)$ flavor symmetry, as expected from \eqref{non-closability and SU(2)}, unless it decouples in the IR. The Dehn filling operation is then implemented through operations involving these $SU(2)$ symmetries. The $SU(2)$ symmetries may be manifest in the bottom-up field theory, or only its $U(1)$ Cartan subgroup may be manifest and enhanced to $SU(2)$ in the IR. For some choices of $N$, there may be no non-closable boundary cycles. In such cases, one must instead use a different Dehn surgery representation \eqref{Dehn filling rep of M} of $M$, for which $N$ admits non-closable cycles.

For the field-theoretic construction, it is useful to introduce the notion of a {\it marginal} non-closable cycle. We define it as follows:
\begin{align}
\begin{split}
&\text{A non-closable cycle } \gamma_I \in H_1(\partial N;\mathbb{Z}) \text{ is called {\it marginal} if}
\\
&\text{a) }\;
\operatorname{coeff}_{q^1}\!\left[
\CI_{N}^{(\vec{\alpha},\vec{\gamma}; \vec{\beta}, \vec{\delta});{\rm fug}}(\vec{u})
\right]\Big|_{(\mathrm{adj}\; SU(2)_I)}
\ge 0 \; ,
\\
&\text{b) }\;
\CI_{N}^{(\vec{\alpha},\vec{\gamma}; \vec{\beta}, \vec{\delta});{\rm fug}}(\vec{u})
\text{ depends non-trivially on } u_{n+I}\;,
\\
&\quad\;\; \text{i.e. } \CI_{N}^{(\vec{\a},\vec{\g};\vec{\beta},\vec{\d})}  ( \vec{m}= \vec{0}, \vec{e}) \textrm{ is non-zero for some $\vec{e}$ with $e_{n+I} \neq 0$.}\label{Def : marginal NC}
\end{split}
\end{align}
Here
$\CI_{N}^{(\vec{\alpha},\vec{\gamma}; \vec{\beta}, \vec{\delta});{\rm fug}}$
denotes the 3D index, to be introduced in Section~\ref{sec : unrefined 3D index}, in the fugacity basis at $\vec{m}=\mathbf{0}$:
\begin{align}
\CI_{N}^{(\vec{\alpha},\vec{\gamma}; \vec{\beta}, \vec{\delta});{\rm fug}}(\vec{u})
:=
\sum_{e_1,\ldots,e_{n+d}\in \mathbb{Z}/2}
\CI^{(\vec{\alpha},\vec{\gamma};\vec{\beta}, \vec{\delta})}_{N}
(\vec{m}=\mathbf{0},\vec{e})
\prod_{i=1}^{n+d}u_i^{e_i}\;,
\end{align}
where the corresponding index in the charge basis is given in \eqref{unrefined 3D index 2}. This quantity computes the generalized superconformal index of the theory $T[N;\vec{\alpha},\vec{\gamma}]$, in the same way that \eqref{urefined index as SCI} computes the index of $T[M;\vec{\alpha}]$.

The notation $|_{(\mathrm{adj}\; SU(2)_I)}$ denotes the projection onto the adjoint representation of $SU(2)_I$---the symmetry associated with the non-closable cycle $\gamma_I$---while projecting onto the trivial representation for all remaining flavor symmetries. Explicitly,
\begin{align}
\begin{split}
f(u_1,\ldots,u_{n+d})\Big|_{(\mathrm{adj}\; SU(2)_I)}
&:=
\left(\prod_{i=1}^{n}\oint\frac{du_i}{2\pi \mathbf{i}\,u_i}\right)
\left(\frac{1}{2^{d}}
\prod_{J=1}^{d}
\oint
\frac{du_{n+J}}{2\pi \mathbf{i}\,u_{n+J}}
(2-u_{n+J}^{2}-u_{n+J}^{-2})\right)
\\
&\qquad\times
(1+u_{n+I}^{2}+u_{n+I}^{-2})\,
f\!\left(u_1^2,\ldots,u_{n+d}^2\right).
\end{split}
\end{align}
When $\gamma_I$ is marginal, we expect that there exists a marginal ($R=2$) chiral primary operator $\mathfrak{M}_I$ in $T[N;\vec{\alpha},\vec{\gamma}]$, transforming in the adjoint representation of the associated $SU(2)_I$:
\begin{align}
\gamma_I \text{ marginal non-closable}
\;\Longrightarrow\;
\mathfrak{M}_I
\;\;(\text{marginal chiral primary in the adjoint of } SU(2)_I)\;.
\end{align}
The $SU(2)_I$ flavor current multiplet contributes to the index at order $q$ as
$-q\,\chi_{\rm adj}(u_{n+I})
=
-q(1+u_{n+I}+u_{n+I}^{-1})$.
Condition a) implies that this  contribution is not present in the index. There are then two possible ways to reconcile this. One possibility is that the $SU(2)_I$ flavor symmetry decouples in the IR and therefore does not survive as a faithful symmetry. In that case, the superconformal index would be independent of the corresponding fugacity $u_{n+I}$. Condition b) rules out this possibility. The only remaining possibility is that there exist marginal operators in the adjoint representation whose contribution cancels that of the flavor current multiplet. We identify these operators with $\mathfrak{M}_I$. 

In practice, marginal non-closable cycles are very rare, and most $3$-manifolds $N$ do not admit such cycles. We will present an example in Section~\ref{sec : example - other hyperbolic Ms}.

    \paragraph{Dehn filling operation on $T[N;\vec{\a},
\vec{\gamma}]$}
We define the Dehn-filled theory \(T[M;\vec{\alpha}]\) by incorporating the Dehn filling operation \cite{Pei:2015jsa,Alday:2017yxk,Gukov:2017kmk,Gang:2018wek,Assel:2022row} on the theory \(T[N]\):
\begin{align}
\begin{split}
&T[M;\vec{\alpha}]
\sim
\frac{
T[N;\vec{\alpha},\vec{\gamma}]
\;\otimes\;
\prod_{I=1}^{d}\mathbb{D}(k^{(I)}_2,\ldots,k^{(I)}_{\ell_I})
}{
\prod_{I=1}^{d} SU(2)^{(I)}_{k^{(I)}_1}
} + \left( \CW = \sum_{\gamma_I : {\rm marginal}} \textrm{Tr}(\mathfrak{M}_I \cdot \boldsymbol{\mu}^{(I)}_H) \right),
\qquad 
\\
&\text{where}
\\[4pt]
&\mathbb{D}(k_2,\ldots,k_\ell)
\;\sim\;
\frac{
T[SU(2)]^{\otimes(\ell-1)}
}{
\prod_{i=2}^{\ell} SU(2)^{(i)}_{k_i}
}\; . \label{T[M] theory}
\end{split}
\end{align}
In the first line, $SU(2)^{(I)}$ $(I=1,\ldots,d)$ denotes the diagonal $SU(2)$
subgroup obtained from the $I$-th $SU(2)$ flavor symmetry, whose Cartan is $T_{n+I}$, of
$T[N;\vec{\alpha},\vec{\gamma}]$ and the $SU(2)$ flavor symmetry of
$\mathbb{D}(k^{(I)}_2,\ldots,k^{(I)}_{\ell_I})$.
In the second line, $SU(2)^{(i)}$ $(i=2,\ldots,\ell-1)$ denotes the diagonal
$SU(2)$ subgroup of the $SU(2)_C$ symmetry of the $(i-1)$-th $T[SU(2)]$ factor
and the $SU(2)_H$ symmetry of the $i$-th $T[SU(2)]$ factor, while
$SU(2)^{(\ell)}$ corresponds to the $SU(2)_C$ symmetry of the
$(\ell-1)$-th $T[SU(2)]$ factor.
The resulting theory $\mathbb{D}(k_2,\ldots,k_\ell)$ is a $3$D $\mathcal{N}=4$
theory with an $SU(2)$ flavor symmetry inherited from the $SU(2)_H$ symmetry
of the first $T[SU(2)]$ factor. Its moment map is denoted by $\boldsymbol{\mu}^{(I)}_H$.
\begin{figure}[h] 
    \centering
    \includegraphics[width=0.85\linewidth]{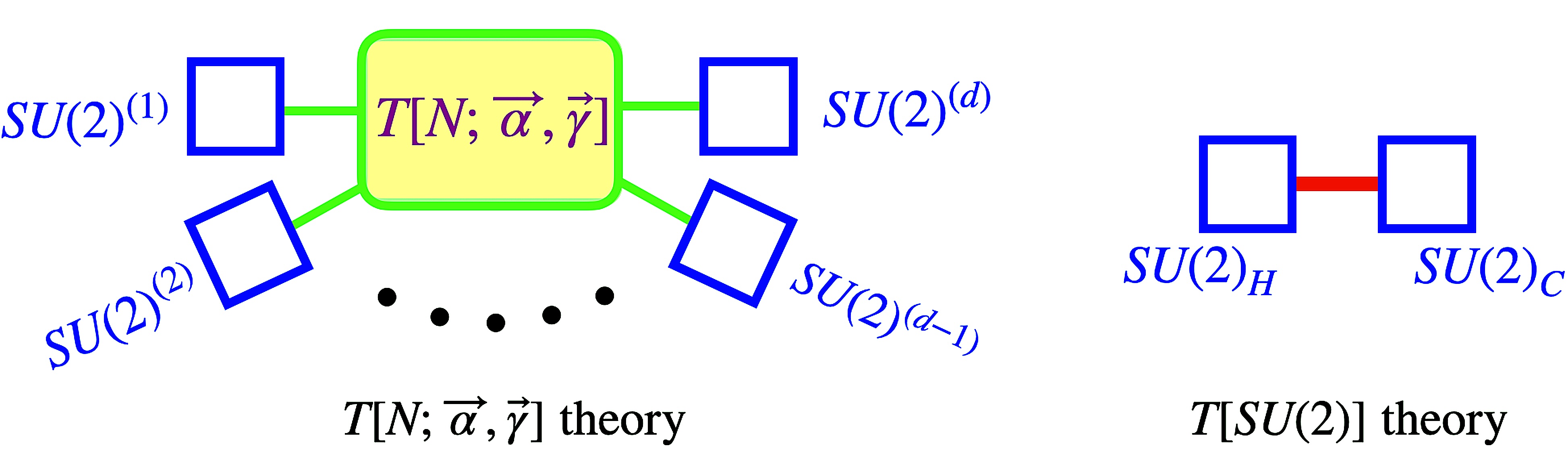}
    \\
    \includegraphics[width=0.78\linewidth]{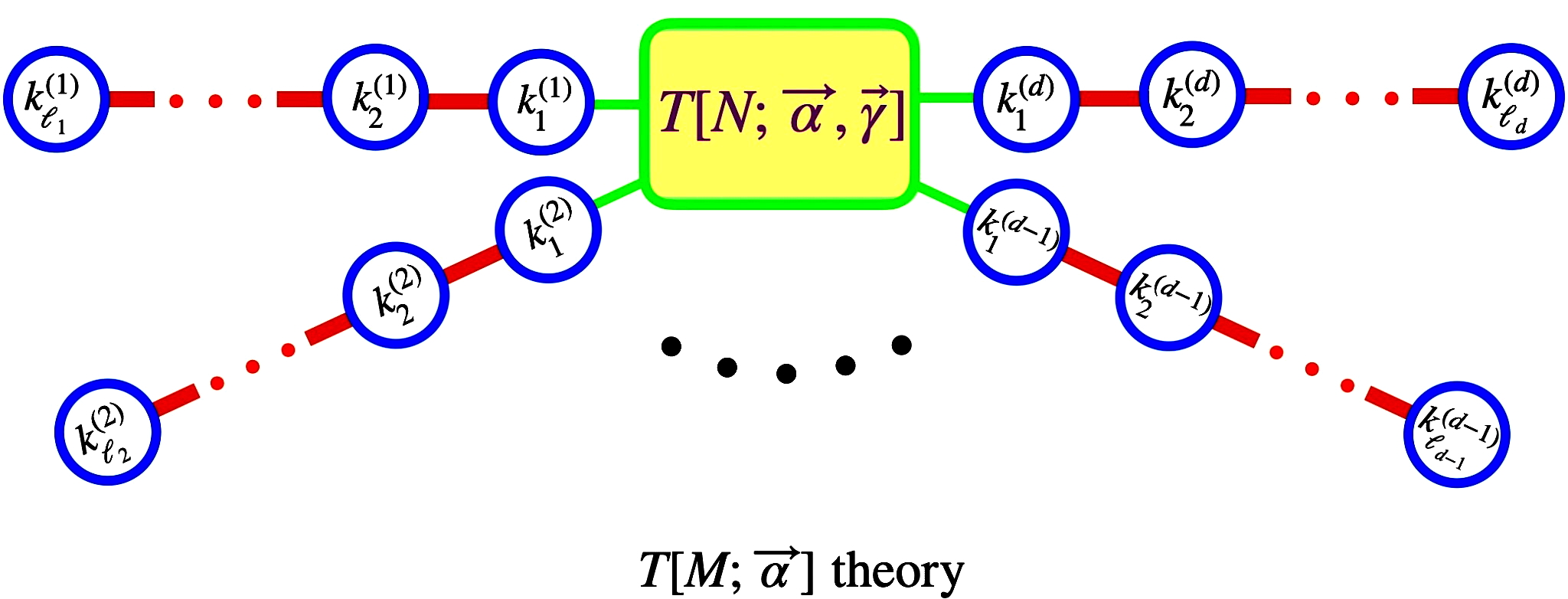}
\caption{
Generalized quiver diagram for the $T[M;\vec{\alpha}]$ theory.
The theory $T[N;\vec{\alpha},\vec{\gamma}]$, represented by a green rounded rectangle, 
has $d$ copies of $SU(2)$ flavor symmetry, shown as blue squares.
The $T[SU(2)]$ theory, represented by red link, has 
$SU(2)_H \times SU(2)_C$ flavor symmetry.
Circles labeled by an integer $k$ denote the gauged $SU(2)$ symmetries, each with Chern-Simons level $k$.
}
\label{fig : T[M]}
\end{figure}
The symbol ``$\sim$'' denotes
\begin{align}
\sim \;:\;&\text{an IR equivalence up to tensoring with a decoupled TQFT} \nonumber\\
&\text{and/or gauging a discrete $1$-form symmetry.}
\end{align}
The decoupled TQFT is irrelevant for our purposes, as it does not contribute to the superconformal index. To obtain $T[M;\vec{\alpha}]$, we must gauge all non-anomalous $\mathbb{Z}_2$ $1$-form symmetries arising from the centers of the $SU(2)$ gauge groups.

The Chern--Simons levels
$\vec{k}^{(I)}=(k^{(I)}_1,\ldots,k^{(I)}_{\ell_I})$
are determined by the continued fraction expansion as follows:
\begin{align}
\begin{split}
&\text{when $Q_I \neq 0$:}\qquad
\frac{P_I}{Q_I}
=
k^{(I)}_1
-\dfrac{1}{
k^{(I)}_2
-\dfrac{1}{
k^{(I)}_3
-\cdots
-\dfrac{1}{k^{(I)}_{\ell_I}}
}}\;,
\\[6pt]
&\text{when $P_I=\pm1$ and $Q_I=0$:}\quad
\ell_I=2,
\qquad
k^{(I)}_1=k^{(I)}_2=0\;. \label{vec-k from P/Q}
\end{split}
\end{align}
When $Q_I=\pm 1$, i.e.\ when the slope $P_I/Q_I$ is an integer, one can choose
$k^{(I)}_1 = P_I/Q_I$ and $\ell_I = 1$.
In this case, the theory $\mathbb{D}(k^{(I)}_2,\ldots,k^{(I)}_{\ell_I})$ is trivial, 
and the Dehn filling operation simply corresponds to gauging $SU(2)^{(I)}$ 
with an additional Chern-Simons term at level $P_I/Q_I$. 
When $Q_I \neq \pm 1$, on the other hand, the theory 
$\mathbb{D}(k^{(I)}_2,\ldots,k^{(I)}_{\ell_I})$ is a nontrivial 
three-dimensional $\mathcal{N}=4$ SCFT with $SU(2)$ flavor symmetry and $SO(4)_R$ R-symmetry.
\begin{align}
\mathbb{D}(k_2, \ldots, k_\ell) \sim 
\begin{cases}
\text{trivial theory}, & P/Q \in \mathbb{Z}, \\
\mathcal{N}=4\ \text{SCFT}, & \text{otherwise}.
\end{cases}
\end{align}
In the theory $T[M;\vec{\alpha}]$, this $\mathcal{N}=4$ supersymmetry is generically 
broken to $\mathcal{N}=2$, since $T[N;\vec{\alpha},\vec{\gamma}]$ itself typically 
preserves only $\mathcal{N}=2$ supersymmetry. Under this partial supersymmetry breaking, 
the $SO(4)_R$ symmetry is reduced to $U(1)_R \times U(1)_A$, where $U(1)_R$ is identified 
with the R-symmetry of the $\mathcal{N}=2$ theory $T[M;\vec{\alpha}]$, while $U(1)_A$ is 
interpreted as a flavor symmetry from the $\mathcal{N}=2$ perspective.

This $U(1)_A$ symmetry remains unbroken when coupled to the theory $T[N;\vec{\alpha},\vec{\gamma}]$, 
unless $\gamma_I$ is marginal, in which case it is broken by a superpotential deformation. 
When preserved, it provides an additional refined flavor symmetry of $T[M]$, 
leading to a refinement of the 3D index.

\paragraph{Auxiliary Choices and IR Dualities}
The UV theory $T[M;\vec{\alpha}]$  depends on several auxiliary choices:
\begin{align}
\begin{split}
&\bullet\ \mathcal{D}\;:\; \text{the Dehn surgery representation in \eqref{Dehn filling rep of M}},\\
&\bullet\ \mathcal{T}\;:\; \text{an ideal triangulation of $N$ in \eqref{Ideal triangulation of N}},\\
&\bullet\ \text{quad structure}\;:\; \text{cyclic labeling of $Z, Z', Z''$ for each tetrahedron},\\
&\bullet\ \vec{\gamma}\;:\; \text{a choice of non-closable cycles in \eqref{basis (a,b,g,d)-2}}.
\end{split} \label{auxiliary choices}
\end{align}
The IR physics is expected to be independent of these choices, and we therefore suppress them in the notation. Different choices of auxiliary data correspond to different IR duality frames of \(T[M;\vec{\alpha}]\) \cite{Dimofte:2011ju,Gang:2018wek}:
\begin{align}
\begin{split}
&\text{Different choices of } \CD,\ \CT,\ \vec{\gamma},\ \text{and quad structure}
\\
&\Longrightarrow \text{different IR duality frames of } T[M;\vec{\alpha}] \; .
\end{split}
\end{align}
For example, the Pachner \(2\text{--}3\) move corresponds to the IR duality between $N_f=1$ SQED and the XYZ model, while a change of quad structure corresponds to the duality between a free chiral theory and a \(U(1)_{-1/2}\) gauge theory coupled to a chiral multiplet \cite{Dimofte:2011ju}. The IR dualities associated with changes of $\CD$ and $\vec{\gamma}$ remain to be clarified.
Taken together, these dual descriptions provide complementary probes of the infrared dynamics.

\paragraph{UV flavor symmetry of $T[M;\vec{\alpha}]$ and its Cartan subalgebra $\mathfrak{F}[M]= \mathfrak{F}^{\rm tori}[M] \oplus \mathfrak{F}^{\rm ref}[M]$}
The subalgebra $\mathfrak{F}^{\rm tori}[M]$ is given by
\begin{align}
\mathfrak{F}^{\rm tori}[M] = \mathrm{Span}\{T_{x_1}, \ldots, T_{x_n}\}\,,
\label{F-tori[M]}
\end{align}
and can be identified, in the 6D construction \eqref{T[M] from 6D}--\eqref{T[M] from 6D-3}, with the Cartan subalgebra of the flavor symmetry in \eqref{flavor symmetry of T[M] from 6D}.  In addition to the symmetries arising from the geometric construction, the 3D effective theory may exhibit extra flavor symmetries. We denote their Cartan subalgebra by $\mathfrak{F}^{\rm ref}[M]$:
\begin{align}
\mathfrak{F}^{\rm ref}[M] \;:\; \text{Cartan subalgebra of additional flavor symmetries}\,.
\end{align}
These additional symmetries give rise to refinements of the 3D index.
There are two possible sources for such symmetries in the theory 
$T[M;\vec{\alpha}]$ defined in \eqref{T[M] theory} and \eqref{T[N] theory}. The first originates from the $U(1)_A$ symmetry of 
$\mathbb{D}(k^{(I)}_2,\ldots,k^{(I)}_{\ell_I})$ in the case $Q_I \neq \pm 1$ 
and $\gamma_I$ is not marginal. We denote its Cartan generator by $T_{a_I}$:
\begin{align}
   T_{a_I} \;:\; \text{Cartan generator of } U(1)_A^{(I)} \, .
\end{align}
The second arises from a subgroup $F^{\rm hard}_{\rm d}$ of $F^{\rm hard}$ in $T[N; \vec{\alpha},\vec{\gamma}]$, whose Cartan subalgebra $\mathfrak{F}^{\rm hard}[N]$ is defined in \eqref{Cartans of T[N]}. This subgroup is chosen to commute with the gauged $SU(2)^d$ symmetry and therefore survives the gauging:
\begin{align}
F^{\rm hard}_{\rm d} \subset F^{\rm hard} := \prod_{i=1}^{\sharp_H} U(1)_{h_i},
\qquad
[F^{\rm hard}_{\rm d}, SU(2)^d]=0\;. \label{Fd-hard and F-hard}
\end{align}
As a result, the Cartan subalgebra of the additional flavor symmetry,
$\mathfrak{F}^{\rm ref}[M]$, takes the form
\begin{align}
\begin{split}
&\mathfrak{F}^{\rm ref}[M]
=
\mathfrak{F}^{\rm dehn}
\oplus
\mathfrak{F}^{\rm hard}_{\rm d}\;, \\[4pt]
&\mathfrak{F}^{\rm dehn}[M]
=
\mathrm{Span}\!\left\{
T_{a_I}\;\middle|\; I=1,\ldots,d,\; (Q_I \neq \pm 1) \textrm{ and }(\textrm{$\gamma_I$ is not marginal})\right\}, \\[4pt]
&\mathfrak{F}^{\rm hard}_{\rm d}[M]
=
\left\{
T \in \mathfrak{F}^{\rm hard}[N] \;\middle|\; [T, SU(2)^d]=0
\right\}.
\label{F-ref[M]}
\end{split}
\end{align}
Here the vector space $\mathfrak{F}^{\rm hard}[N]$ is defined in \eqref{Cartans of T[N]}. Sometimes, some of the $SU(2)^d$ symmetry in $T[N;\vec{\alpha}, \vec{\gamma}]$ emerges only in the IR, so it is nontrivial to identify $\mathfrak{F}^{\rm hard}_{\rm d}$.

\paragraph{IR Flavor Symmetry of \( T[M;\vec{\alpha}] \)}
The Cartan subalgebra $\mathfrak{F}[M]$, in particular $\mathfrak{F}^{\rm ref}[M]$, of the UV flavor symmetry of $T[M;\vec{\alpha}]$ depends on the auxiliary choices,  $\CD,\CT$  and $\vec{\gamma}$ in \eqref{auxiliary choices}. Nevertheless, the faithful IR flavor symmetry is expected to be independent of these choices. How are these seemingly different UV flavor symmetries reconciled in the IR?  Two mechanisms account for this reconciliation. First, some of the UV flavor symmetries do not act faithfully in the IR and hence decouple along the RG flow. Let $\mathfrak{F}^{\rm ref}_{\rm dec}[M] \subset \mathfrak{F}^{\rm ref}[M]$ be the Cartan subalgebra corresponding to the decoupled flavor symmetry. Then the Cartan subalgebra of the remaining (faithful) IR flavor symmetry is given by
\begin{align}
\mathfrak{F}^{\rm ref}_{\rm IR}[M] = \mathfrak{F}^{\rm ref}[M]/\mathfrak{F}^{\rm ref}_{\rm dec}[M] \, . \label{F-ref[M]-IR}
\end{align}
Second, additional symmetries, not manifest in the UV description, can emerge in the IR, leading to symmetry enhancement. In most cases, such enhancement does not change the Cartan subalgebra—for example, a $U(1)$ symmetry may be enhanced to $SU(2)$ without increasing the rank. However, in more rare situations, the rank of the flavor symmetry itself can increase in the IR. In general, such cases are difficult to analyze directly. In practice, we instead consider different IR duality frames and compute $\mathfrak{F}^{\rm ref}_{\rm IR}[M]$ in each frame, taking the largest resulting vector space as the Cartan subalgebra of the IR flavor symmetry. Let $(\CD_0, \CT_0,  \vec{\gamma}_0)$ be a choice that yields the largest resulting vector space. Then, for a general choice of $(\CD, \CT, \vec{\gamma})$, there exists an injective map
\begin{align}
\mathfrak{F}^{\rm ref}_{\rm IR}[M; (\CD, \CT, \vec{\gamma})] \hookrightarrow 
\mathfrak{F}^{\rm ref}_{\rm IR}[M; (\CD_0, \CT_0, \vec{\gamma}_0)] \, , \label{injective map F}
\end{align}
which indicates how the flavor symmetries in different duality frames are related to each other. Since rank enhancement occurs only exceptionally, the injective map is expected to be bijective in generic cases.

\paragraph{Example $M= (S^3\backslash \mathbf{4}_1)_{[P \mu + \lambda]}$} 
To find a UV gauge theory description of $T[M]$, we choose
\begin{align}
N = S^3\backslash \mathbf{4}_1, \quad \{ \gamma, \delta \} = \{\mu, \lambda\} \, .
\end{align}
Here $\mathbf{4}_1$ denotes the figure-eight knot, the simplest knot with 4 crossings. $\mu$ and $\lambda$ denote the boundary 1-cycles, namely the meridian and longitude, respectively. The knot complement $N = S^3\backslash \mathbf{4}_1$ admits a triangulation by two ideal tetrahedra, and the corresponding gluing equations, in an appropriate choice of quad structure, are
\begin{align}
\begin{split}
&C_1 = 2Z_1'+Z_1''+2Z_2+Z_2'\;, \quad C_2 = 2Z_1+Z_1''+Z_2'+2Z_2''\;,
\\
&hol(\mu) = Z_1 - Z_2\;, \quad hol(\lambda) = 2(Z_1-Z_1')\;. \label{gluing eqs for 4_1}
\end{split}
\end{align}
There is no easy internal edge, and there is one hard internal edge, namely
\begin{align}
H=C_1 \;.
\end{align}
Then the Cartan subalgebras in \eqref{Cartans of T[N]} are
\begin{align}
\begin{split}
&\mathfrak{F}^{\rm tori}[N] = \mathrm{Span}\{T_x\}\;, \\
&\mathfrak{F}^{\rm hard}[N] = \mathrm{Span}\{T_h\}\;.
\end{split}
\end{align}
The corresponding NZ matrices are
\begin{align}
g_{\rm NZ} = \begin{pmatrix} 1 & -1 & 0 & 0 \\ -2 & 1 & -1 & -1 \\ 2 & 0 & 1 & 0 \\ 0 & 1 & 0 & 0\end{pmatrix}, \qquad
\boldsymbol{\nu}_x =  \begin{pmatrix} 0 \\ 1 \end{pmatrix}, \qquad
\boldsymbol{\nu}_p =\begin{pmatrix} -1 \\ 0 \end{pmatrix}. \label{NZ matrices of 4_1}
\end{align}
Using the decomposition $g_{\rm NZ} = g^s_J g^t_K g^{\rm gl}_U$, with
\begin{align}
U = \begin{pmatrix}  1 & -1 \\ 0 & 1\end{pmatrix}, \qquad
K = \begin{pmatrix}  2 & 2 \\ 2 & 1\end{pmatrix}, \qquad
J = \begin{pmatrix}  0 & 0 \\ 0 & 1\end{pmatrix},
\end{align}
the UV gauge theory $T[S^3\backslash \mathbf{4}_1; \mu ]= [g_{\rm NZ}\cdot (T_\Delta)^{\otimes 2}]$ is (after shifting the dynamical vector field as $V \rightarrow V-V_1/2$)
\begin{align}
\begin{split}
&\mathcal{L}_{T[S^3\backslash \mathbf{4}_1; \mu]} (V_1, V_2)
\\
&= \frac{1}{4\pi }\int d^4 \theta \left(-\frac{1}{2} \Sigma_2 V_1 + \Sigma (2V_2 +3 V_1)\right)
+ \int d^4 \theta \left(\Phi_1^\dagger e^{V+\frac{V_1}{2}}\Phi_1+\Phi_2^\dagger e^{V-\frac{V_1}{2}}\Phi_2 \right).
\end{split}
\end{align}
Here $V$ (with field strength $\Sigma$) is the dynamical vector multiplet, $V_1$ is the background vector multiplet coupled to $U(1)_x$ and $V_2$ is the background vector multiplet coupled to $U(1)_h$.
The theory can therefore be summarized as
\begin{align}
\begin{split}
&T[S^3\backslash \mathbf{4}_1; \mu]
\\
&= U(1)_0 \ \text{coupled to two chiral multiplets, $\Phi_1$ and $\Phi_2$, each of charge $+1$.} \label{T[4-1]}
\end{split}
\end{align}
This theory has a manifest $U(1)\times SU(2)$ flavor symmetry.  The $U(1)$ factor is the topological symmetry associated with the $U(1)$ gauge symmetry, while the $SU(2)$ rotates the two chiral multiplets. The two Cartan generators, $T_x$ and $T_h$, are embedded into the $U(1)\times SU(2)$ as follows
\begin{align}
T_h  = T_{\rm top}, \quad T_x = \frac{3}2 T_{\rm top} + T_3,
\end{align}
where $T_{\rm top}$ and $T_3$ are Cartan generators of $U(1)$ and $SU(2)$ respectively.  Since $\mu$ is a non-closable cycle, we expect the $U(1)_x$ symmetry to be enhanced to $SU(2)$ in the IR, as argued in \eqref{non-closability and SU(2)}. Indeed, the gauge theory exhibits an enhanced $SU(3)$ symmetry in the IR \cite{Gang:2018wek,Benini:2018bhk,Gaiotto:2018yjh}, in which $U(1)_h \times U(1)_x$ is embedded as
\begin{align}
T_x = \textrm{diag}(0,-1,1) \in \mathfrak{su}(3), \qquad
T_h = \frac{1}3 \textrm{diag}(-1,-1,2) \in \mathfrak{su}(3).
\end{align}
The $SU(2)_x$ symmetry is embedded into the $SU(3)$ flavor symmetry as follows
\begin{align}
(\textrm{adjoint of }SU(2)_x) = (\mathbf{3} \textrm{ of }SU(3)) \;. \label{SU(2) in SU(3)}
\end{align}
Integer Dehn filling corresponds to gauging this $SU(2)_x$ flavor symmetry:
\begin{align} 
T[(S^3 \backslash \mathbf{4}_1)_{[P \mu + \lambda]}] \sim \frac{T[S^3\backslash \mathbf{4}_1;\mu]}{SU(2)_P}\;. \label{T[(4-1)P]}
\end{align}
Here $/G_k$ denotes gauging the flavor symmetry $G$ with additional Chern--Simons level $k$. The gauged $SU(2)_x$ symmetry does not commute with $U(1)_h$, and thus we have
\begin{align}
\mathfrak{F}^{\rm hard}_{\rm d}[M] = \{ \mathbf{0} \}\;.
\end{align}
Since we are considering integer Dehn filling, it follows from \eqref{F-ref[M]} that
\begin{align}
\mathfrak{F}^{\rm dehn}[M] = \{ \mathbf{0} \}
\quad \textrm{and} \quad
\mathfrak{F}^{\rm ref}[M] = \{ \mathbf{0} \}\;.
\end{align}
So generically, we expect that $T[(S^3\backslash \mathbf{4}_1)_{[P\mu+\lambda]}]$ does not have any continuous flavor symmetry. In particular, when $|P|$ is sufficiently large, quantum effects are suppressed by $1/P$, and the naive counting is expected to be reliable. However, for small values of $P$, strong quantum effects can invalidate this naive counting. In this case, monopole operators associated with the $SU(2)$ gauge symmetry are light (i.e.\ have small conformal dimensions) and can give rise to additional conserved currents, whose conformal dimension is fixed to be $2$ in $D=3$.

For $|P|=5$, we claim that the IR theory has an emergent $U(1)$ flavor symmetry, which is manifest in a suitable duality frame. This frame can be reached using the topological identity
\begin{align}
(S^3 \backslash \mathbf{4}_1)_{[-5\mu+\lambda]} = (S^3\backslash \mathbf{5}_2)_{[5 \mu+\lambda]} \, .
\end{align}
Here $\mathbf{5}_2$ is the three-twist knot, the second simplest knot with five crossings.
\paragraph{Example $M= (S^3\backslash \mathbf{5}_2)_{[P \mu + \lambda]}$} We choose
\begin{align}
N = S^3\backslash \mathbf{5}_2, \quad \{ \gamma, \delta \} = \{\mu, \lambda\} \, .
\end{align}
The knot complement $N$ can be ideally triangulated by 3 tetrahedra and the gluing equations are
\begin{align}
\begin{split}
&C_1= Z_1'+Z_1''+2Z_2+Z_3'+Z_3'', \quad C_2 = Z_1+Z_1'+Z_2'+Z_3+Z_3',  
\\
&C_3 = Z_1+Z_1''+Z_2'+2Z_2''+Z_3+Z_3'',
\\
&hol(\mu) = -Z_1+Z_2, \quad hol(\lambda)= 2 (Z_1-Z_1''-Z_2+Z_3'') \label{gluing eqs for 5_2}
\end{split}
\end{align}
There is no easy internal edge and two hard internal edges, $H_1 = C_1$ and $H_2 = C_2$. The NZ matrices are
\begin{align} \label{NZ matrices 5_2}
g_{\rm NZ} = \begin{pmatrix} -1 & 1 & 0 & 0 & 0 & 0 
\\ -1 & 2 & -1 & 0 & 0 & 0 
\\ 0 & -1 & 0 & -1 & -1 & -1
\\ 1 & -1 & 0 & -1 & 0 & 1
\\ 0 & 0 & 0 & 0 & 0 & -1
\\ 0 & 1 & 0 & 0 & 0 & 0
\end{pmatrix}, \; \boldsymbol{\nu}_x = \begin{pmatrix} 0 \\ 0 \\ 1 \end{pmatrix}, \; \boldsymbol{\nu}_p = \begin{pmatrix} 0 \\ 0 \\ 0 \end{pmatrix}
\end{align}
Using the decomposition $g_{\rm NZ} = g^s_J g^t_K g^{\rm gl}_U$, with
\begin{align}
U = \begin{pmatrix}  -1 & 1 & 0 \\ -1 & 2 & -1 \\ 0 & 1 & 0 \end{pmatrix}, \qquad
K = \begin{pmatrix}  -1 & 0 &  0 \\ 0 & 0 & 0 \\ 0 & 0 & 1\end{pmatrix}, \qquad
J = \begin{pmatrix}  0 & 0 &  0 \\ 0 & 0 & 0 \\ 0 & 0 & 1\end{pmatrix},
\end{align}
the UV gauge theory $T[S^3\backslash \mathbf{5}_2; \mu ]= [g_{\rm NZ}\cdot (T_\Delta)^{\otimes 3}]$ is 
\begin{align}
\begin{split}
&\mathcal{L}_{T[S^3\backslash \mathbf{5}_2; \mu]} (V_1, V_2,V_3)
\\
&= \frac{1}{4\pi }\int d^4 \theta \left(-\frac{1}{2} \Sigma V + \Sigma (V_2+2V_3)-2 \Sigma_1 V_1+ \Sigma_1 \Sigma_2 - \frac{1}2 \Sigma_2 V_2\right)
\\
&+ \int d^4 \theta \left(\Phi_1^\dagger e^{V-V_1}\Phi_1+\Phi_2^\dagger e^{V}\Phi_2  +\Phi_3^\dagger e^{V+V_1-V_2}\Phi_3\right). \label{L for 5_2}
\end{split}
\end{align}
Here $V$ (with field strength $\Sigma$) is the dynamical vector multiplet, $V_1$ is the background vector multiplet coupled to $U(1)_x$ and $V_2$ and $V_3$ are the background vector multiplet coupled to $U(1)_{h_1}$ and $U(1)_{h_2}$ respectively. The theory can therefore be summarized as
\begin{align}
\begin{split}
&T[S^3\backslash \mathbf{5}_2; \mu]
\\
&= U(1)_{-\frac{1}2} \text{ coupled to three chiral multiplets, $\Phi_1,\Phi_2$ and $\Phi_3$, each of charge $+1$.}
\end{split}
\end{align}
This theory has a manifest $U(1)\times SU(3)$ flavor symmetry.  The $U(1)$ factor is the topological symmetry associated with the $U(1)$ gauge symmetry, while the $SU(3)$ rotates the three chiral multiplets. The Cartan generators, $T_x,T_{h_1}$ and $T_{h_2}$, are embedded into the $U(1)\times SU(3)$ as follows
\begin{align}
T_x = \diag (-1,0,1) \in \mathfrak{su}(3), \quad T_{h_1} = \left(\diag (0,0,-1) \in \mathfrak{su}(3)\right)+\frac{1}2 T_{\rm top}, \quad T_{h_2} = T_{\rm top}\;, \nonumber
\end{align}
where $T_{\rm top}$ is  Cartan generator of $U(1)$. The $U(1)_x$ is enhanced to $SU(2)_x$, which is embedded into $SU(3)$ as in \eqref{SU(2) in SU(3)}. 
Integer Dehn filling corresponds to gauging this $SU(2)_x$ flavor symmetry:
\begin{align} 
\begin{split}
&T[(S^3 \backslash \mathbf{5}_2)_{[P \mu + \lambda]}] \sim \frac{T[S^3\backslash \mathbf{5}_2;\mu]}{SU(2)_P}\;
\\
&= U(1)_{-\frac{1}2} \times SU(2)_{P-4} \textrm{ coupled to a chiral multiplet $\Phi$ in $(\mathbf{3},1)$}.
\end{split}
\end{align}
Here, $(\mathbf{3},1)$ denotes the fundamental representation of $SU(3)$ with charge $+1$ under $U(1)$. 
The Chern--Simons term, $\frac{1}{4\pi } \int d^4 \theta (-2 \Sigma_1 V_1)$ in \eqref{L for 5_2}, 
which corresponds to level $-4$ for the background gauge field of the $SU(2)_x$ symmetry in $T[S^3\backslash \mathbf{5}_2]$, is taken into account. 
The gauged $SU(2)_x$ symmetry does not commute with $U(1)_{h_1}$, but commutes with $U(1)_{h_2}$, and thus we have
\begin{align}
\mathfrak{F}^{\rm hard}_{\rm d}[M] = \textrm{Span}\{ T_{h_2} \}\;.
\end{align}
Since we are considering integer Dehn filling, it follows from \eqref{F-ref[M]} that
\begin{align}
\mathfrak{F}^{\rm dehn}[M] = \{ \mathbf{0} \}
\quad \textrm{and} \quad
\mathfrak{F}^{\rm ref}[M] = \textrm{Span}\{ T_{h_2} \}\;.
\end{align}
Therefore, we expect that $T[(S^3\backslash \mathbf{5}_2)_{[P\mu+\lambda]}]$ generically has a rank-one flavor symmetry. In particular, when $P=5$, the UV field theory is IR dual to $T[(S^3\backslash \mathbf{4}_1)_{[-5 \mu+\lambda]}]$ in \eqref{T[4-1]} and \eqref{T[(4-1)P]}, which has a manifest $U(1)$ flavor symmetry.

\section{Refined 3D Index} \label{sec : refined 3D index-main}
    Motivated by the field-theoretic construction of \(T[M;\vec{\alpha}]\) in the previous section, we introduce a refined 3D index that computes the superconformal index of the theory in the presence of fugacities for the additional flavor symmetries in \eqref{F-ref[M]}. We further reformulate the refined index in terms of normal surface counting.
    
   \subsection{3D index} \label{sec : unrefined 3D index}
   As a warm-up, we first review the formula for the (unrefined) 3D index \cite{Dimofte:2011py,Gang:2018wek}. The 3D index computes the superconformal index, see \eqref{urefined index as SCI}, of the theory \(T[M;\vec{\alpha}]\) using only the symmetries manifest in the 6D compactification picture \eqref{T[M] from 6D}--\eqref{T[M] from 6D-3}. Via the 3D--3D correspondence \cite{Dimofte:2011py,Yagi:2013fda,Lee:2013ida}, the 3D index can alternatively be given as the partition function of \(SL(2,\mathbb{C})\) Chern--Simons theory on \(M\) at quantized level \(k=0\).

The 3D index, when it converges, defines a map
\begin{align}
\begin{split}
&\CI_M \;:\; H_1(\partial M;\tfrac12\mathbb{Z}) \longrightarrow \mathbb{Z}[[q^{1/2}]],\\
&b=A_1\alpha_1+\cdots + A_n \alpha_n + B_1\beta_1+\cdots+B_n\beta_n
\;\longmapsto\;
\CI_M^{b}
\in \mathbb{Z}[[q^{1/2}]] \; .
\end{split}
\end{align}
The index is given by \cite{Dimofte:2011py,Gang:2018wek,Celoria:2025lqm}
\begin{align}
\begin{split}
&\CI^{\vec{A}\cdot \vec{\alpha} + \vec{B}\cdot \vec{\beta}}_M
=
\CI_M^{(\vec{\alpha};\vec{\beta})}
(m_1,\ldots,m_n,e_1,\ldots,e_n)
\Big|_{\, e_i = A_i,\; m_i = -2 B_i } \; ,
\\[6pt]
&\CI_M^{(\vec{\alpha};\vec{\beta})}
(m_1,\ldots,m_n;e_1,\ldots,e_n)
=
\sum_{\{(m_{n+I},e_{n+I})\}_{I=1}^d }
\left(
\prod_{I=1}^d
\CK(P_I,Q_I;m_{n+I},e_{n+I})
\right)
\\
&\qquad\qquad \qquad\qquad  \qquad\qquad  \qquad \qquad \times\;
\CI_N^{(\vec{\alpha},\vec{\gamma};\vec{\beta},\vec{\delta})}
(m_1,\ldots,m_{n+d}, e_1,\ldots,e_{n+d}) \; .
\end{split}
\label{unrefined 3D index 1}
\end{align}
The summation ranges are
\begin{align}
m_{n+I} \in \mathbb{Z}\;, 
\qquad 
e_{n+I} \in \tfrac{1}{2}\mathbb{Z} \;. \label{sum range}
\end{align}
Here $\CI_N$ is the 3D index of $N$ written in the charge basis:
\begin{align}
\begin{split}
&\CI_N^{(\vec{\alpha},\vec{\gamma};\vec{\beta},\vec{\delta})}
(m_1,\ldots,m_{n+d},e_1,\ldots,e_{n+d})
\\
&=
\sum_{e_{n+d+1},\ldots,e_r \in \mathbb{Z}/2}
\left(
(-q^{1/2})^{\mathbf{m}\cdot \boldsymbol{\nu}_p - \mathbf{e}\cdot \boldsymbol{\nu}_x}
\prod_{a=1}^r
\CI_\Delta\!\left(
(g_{\rm NZ}^{-1}\boldsymbol{\kappa})_a,
(g_{\rm NZ}^{-1}\boldsymbol{\kappa})_{r+a}
\right)
\right)
\Bigg|_{\, m_{a>n+d} \rightarrow 0 } ,
\\
&\text{where }\;
\mathbf{m}:=(m_1,\ldots,m_r)^T\;,\quad
\mathbf{e}:=(e_1,\ldots,e_r)^T\;,\quad
\boldsymbol{\kappa}:=(m_1,\ldots,m_r,e_1,\ldots,e_r)^T \; . \label{unrefined 3D index 2}
\end{split}
\end{align}
The matrices $(g_{\rm NZ},\boldsymbol{\nu}_x,\boldsymbol{\nu}_p)$ denote the
Neumann--Zagier data introduced in \eqref{NZ matrices}. For the computation of the unrefined index, in constructing the matrices, the condition that $\gamma_I$ be non-closable is not required, and there is no need to distinguish between hard and easy internal edges.
The $\CI_\Delta$ denotes the tetrahedron index \cite{Dimofte:2011py}:
\begin{align}
\begin{split}
&\CI_\Delta(m,e)
:=
\begin{cases}
\displaystyle
\sum_{n=\lfloor e \rfloor}^{\infty}
\frac{(-1)^n\, q^{\frac12 n(n+1) - (n+\frac12 e)m}}
{(q)_n (q)_{n+e}}\;,
& m,e \in \mathbb{Z}\;,
\\[6pt]
0\;, & \text{otherwise}\;,
\end{cases}
\\&\text{where }\lfloor e\rfloor\equiv \frac{1}{2}(|e|-e)\text{ and }(q)_n\equiv\prod_{k=1}^n(1-q^k)\;.
\label{tetrahedron index}
\end{split}
\end{align}
Equivalently, the tetrahedron index in fugacity basis is given by
\begin{align}
&\CI^{\rm fug}_\Delta(m,u)
:=
\sum_{e\in\mathbb{Z}} \CI_\Delta(m,e)\, u^e
=
\prod_{r=0}^{\infty}
\frac{1 - q^{\,r-\frac12 m+1} u^{-1}}
     {1 - q^{\,r-\frac12 m} u}\;, \label{tetrahedron index -fugacity}
\end{align}
which compute the generalized superconformal index of $T_\Delta$ theory. 
The unrefined Dehn filling kernels are given by \cite{Gang:2018wek,Celoria:2025lqm}
\begin{align}
\begin{split}
&\CK(P,Q;m,e)
\\
&:=
\frac12 (-1)^{R m + 2 S e}
\Bigl[
\delta_{P m + 2 Q e,\,0}
\left(
q^{\frac{R m + 2 S e}{2}} + q^{-\frac{R m + 2 S e}{2}}
\right)
-\delta_{P m + 2 Q e,\,-2}
-\delta_{P m + 2 Q e,\,2}
\Bigr]. \label{unrefined Dehn filling kernel}
\end{split}
\end{align}
Here the integers $(R,S)$ are chosen such that
\begin{align}
\begin{pmatrix}
R & S \\
P & Q
\end{pmatrix}
\in SL(2,\mathbb{Z}) \; .
\end{align}
\paragraph{Convergence of the 3D index}
In summary, combining \eqref{unrefined 3D index 1} and \eqref{unrefined 3D index 2}, the 3D index is schematically given by the following infinite sum:
\begin{align}
\sum_{\{(m_{n+I},e_{n+I})\}_{I=1}^d} \sum_{e_{n+d+1}, \ldots, e_r} (\textrm{formal $q^{1/2}$ series}) .
\end{align}
The expression depends on the  auxiliary choices in \eqref{auxiliary choices}. In some cases this infinite sum does not converge as a formal $q^{1/2}$ series, in which case the 3D index is ill-defined. The convergence or divergence may depend on the auxiliary choices rather than on the manifold $M$ itself. When the sum converges, the resulting 3D index is expected to be independent of these choices and depend only on the 3-manifold $M$. It is also believed that for hyperbolic $M$, there always exists a choice of auxiliary data for which the index converges.  However, for some non-hyperbolic manifolds the 3D index may fail to converge for any choice of auxiliary data. Typical examples include reducible or toroidal manifolds, i.e. manifolds containing essential spheres or tori \cite{Choi:2022dju}.

The 3D index computes the following generalized superconformal index of the $T[M;\vec{\alpha}]$ theory:
\begin{align}
\begin{split}
&\CI_M^{(\vec{\alpha},\vec{\beta});{\rm fug}} (m_1, \ldots, m_n,  u_1, \ldots, u_n)
\\
&:= \sum_{e_1, \ldots, e_n\in \mathbb{Z}/2} 
\CI_M^{(\vec{\alpha},\vec{\beta})} (m_1, \ldots, m_n; e_1, \ldots, e_n)
\, u_1^{e_1}\ldots u_n^{e_n}
\\
&= \textrm{Tr}_{\CH (S^2; m_1, \ldots, m_n)}
(-1)^{R_{\rm geo}} 
q^{\frac{R_{\rm geo}}{2} + j_3} 
u_1^{T_{x_1}}\ldots u_n^{T_{x_n}} \;. \label{urefined index as SCI}
\end{split}
\end{align}
Here $\CH(S^2; m_1, \ldots, m_n)$ denotes the radially quantized Hilbert space of the $T[M;\vec{\alpha}]$ theory on $S^2$ in the presence of background monopole fluxes $\{m_i\}_{i=1}^n$ for the $U(1)$ flavor symmetries generated by $\{T_{x_i}\}_{i=1}^n$. When $m_1=\cdots=m_n=0$, the radially quantized Hilbert space
$\CH(S^2; m_1,\ldots,m_n)$ is isomorphic to the space of local operators
of the $T[M;\vec{\alpha}]$ theory (equivalently, to the spectrum of local operators) via the state--operator correspondence. The choice of $\vec{\beta}$ only affects the background Chern-Simons levels
for the $U(1)$ flavor symmetries, and its dependence can be factorized in the form $\prod_{i,j=1}^{n} u_i^{K_{ij} m_j}$.

\subsection{Refined 3D index} \label{sec : refined 3D index}
The refined index, when it converges, defines a map
\begin{align}
\begin{split}
&\CI_{M; {\rm ref}}
\;:\;
H_1\!\left(\partial M;\tfrac{1}{2} \mathbb{Z}\right)
\times
\Gamma\!\bigl(\mathfrak{F}^{\rm ref}[M]\bigr)
\longrightarrow
\mathbb{Z}[[q^{\frac{1}{2}}]][[\eta,\eta^{-1}]]\;,
\\
&\CI^{A_1\alpha_1+\cdots + A_n \alpha_n + B_1\beta_1+\cdots+B_n\beta_n}_{M; {\rm ref}}(\mathbf{v})
\in
\mathbb{Z}[[q^{\frac{1}{2}}]][[\eta,\eta^{-1}]]\;,
\end{split}
\end{align}
where
\begin{align}
\mathbf{v}
=
\sum_{\substack{I=1 \\ Q_I \neq \pm 1,\ \gamma_I\ \text{not marginal}}}^{d}
V_I \, T_{a_I}
+
\sum_{i=1}^{\sharp_H} W_i \, T_{h_i}
\;\in\;
\Gamma\!\bigl(\mathfrak{F}^{\rm ref}[M]\bigr)\,.
\end{align}
with $V_I, W_i \in \mathbb{Z}$. The vector space $\mathfrak{F}^{\rm ref}[M]$ is given in \eqref{F-ref[M]}. One nontrivial task is to identify the subspace $\mathfrak{F}^{\rm hard}_{\rm d}[M]$ of $\mathfrak{F}^{\rm hard}[N]$; this will be addressed below. The refined 3D index computes the following generalized refined superconformal index of the $T[M;\vec{\alpha}]$ theory:
\begin{align}
\begin{split}
&\CI_{M;{\rm ref}}^{(\vec{\alpha},\vec{\beta});{\rm fug}} (\mathbf{v};\vec{m},  \vec{u}) := \sum_{e_1, \ldots, e_n\in \mathbb{Z}/2} 
\CI_{M;{\rm ref}}^{(\vec{\alpha},\vec{\beta})} (\mathbf{v};\vec{m}, \vec{e})
\, \prod_{i=1}^n u_i^{e_i}
\\
&= \textrm{Tr}_{\CH (S^2; \vec{m})}
(-1)^{R_{\rm geo}} 
q^{\frac{R_{\rm geo}}{2} + j_3} 
u_1^{T_{x_1}}\ldots u_n^{T_{x_n}} \eta^{\mathbf{v}} \;.
\label{SCI interpretation of refined 3D index}
\end{split}
\end{align}
Here we define
\begin{align}
\CI^{(\vec{\a};\vec{\b})}_{M;{\rm ref}}(\mathbf{v};\vec{m}, \vec{e}):=
\CI^{\vec{A}\cdot \vec{\a}+\vec{B}\cdot \vec{\beta} }_{M;{\rm ref}} (\mathbf{v})|_{A_i =  e_i, B_i =- m_i/2}\;. \label{refined index in boundary 1-cycles}
\end{align}
For a given Dehn filling presentation $\mathcal{D}$ of $M$
in~\eqref{Dehn filling rep of M}, an ideal triangulation $\mathcal{T}$ of $N$
in~\eqref{Ideal triangulation of N}, and a choice of non-closable cycles $\vec{\gamma}$
in~\eqref{basis (a,b,g,d)-2}, the refined index is defined as follows:
\begin{align}
\begin{split}
&\CI^{(\vec{\alpha};\vec{\beta})}_{M;{\rm ref}} 
(\mathbf{v};m_1, \ldots, m_n, e_1,\ldots, e_n) 
\\
&= \sum_{\{(m_{n+I},e_{n+I})\}_{I=1}^d }
\;  \prod_{I=1}^d 
\begin{cases}
\CK_{\rm ref}\!\left(
P_I,Q_I;\,m_{n+I},e_{n+I};\,\eta^{V_I}
\right), 
& Q_I \neq \pm 1,\ \gamma_I \text{ is not marginal}
\\
\CK\!\left(
P_I,Q_I;\,m_{n+I},e_{n+I}
\right), 
& \text{otherwise}
\end{cases}
\\
&\qquad \times 
\CI^{(\vec{\alpha},\vec{\gamma}; \vec{\beta},\vec{\delta})}_{N ;{\rm ref}}
(\mathbf{v};m_1, \ldots, m_{n+d},e_1,\ldots, e_{n+d}, \mathbf{a}_{\rm d}, \mathbf{b}_{\rm d})
\;. 
\label{refined index in charge basis}
\end{split}
\end{align}
The ranges of summation are the same as in \eqref{sum range}.   Here, $\CK_{\rm ref}(P,Q; m, e;\eta)$ denotes the refined Dehn filling kernel, whose explicit expression is given in \eqref{refined Dehn filling kernel}. It reduces to the unrefined kernel in \eqref{unrefined Dehn filling kernel} when $\eta=1$ (i.e.\ when the refinement is turned off) or when $Q=\pm 1$ (i.e.\ for integral Dehn filling).   $\CI_{N; {\rm ref}}$ is the refined 3D index for $N$ in charge basis
\begin{align}
\begin{split}
&\CI^{(\vec{\alpha},\vec{\g};\vec{\b}, \vec{\d})}_{N;{\rm ref}} (\mathbf{v};\vec{m},\vec{e}) 
\\
&= \sum_{e_{n+d+1}, \ldots, e_r \in \mathbb{Z}/2}\left(  (-q^{1/2})^{\mathbf{m}\cdot \boldsymbol{\nu}_p - \mathbf{e}\cdot \boldsymbol{\nu}_x } \prod_{a=1}^r \CI_\Delta \left((g_{\rm NZ}^{-1}\boldsymbol{\kappa})_a, (g_{\rm NZ}^{-1}
\boldsymbol{\kappa})_{r+a}\right) \prod_{i=1}^{\sharp_H} \eta^{2 W_i e_{n+d+i}}\right)\bigg{|}_{m_{a>n+d}\rightarrow 0}\;, \label{refined index for N in charge basis}
\end{split}
\end{align}
where  $\mathbf{m}:= (m_1, \ldots ,m_r)^T\;, \quad \mathbf{e}:=(e_1, \ldots, e_r)^T\;, \quad \boldsymbol{\kappa} :=(m_1, \ldots, m_r, e_1, \ldots, e_r)^T$.
The matrices $(g_{\rm NZ},\boldsymbol{\nu}_x,\boldsymbol{\nu}_p)$ denote the
Neumann--Zagier data introduced in \eqref{NZ matrices}. For the computation of the refined index, unlike in the unrefined case, one must choose $\gamma_I$ to be non-closable in constructing these matrices, and distinguish between hard and easy internal edges, as in \eqref{NZ matrices}. As in the construction of $T[M;\vec{\alpha}]$, the refined index cannot be computed using a Dehn surgery representation in \eqref{Dehn filling rep of M} if $N$ does not admit non-closable cycles. We further define, for
\begin{align}
\begin{split}
%&a^{(i)}_{1\le I\le n}\in \frac{\mathbb{Z}}{\Omega(\alpha_I)},\quad
%b^{(i)}_{1\le I\le n}\in \frac{\Omega(\alpha_I)}{2}\mathbb{Z},
%\\
&\mathbf{a} = (\vec{a}_1,\ldots, \vec{a}_{d})  \textrm{ with }\vec{a}_{1\leq I \leq d} = (a_I^{(1)},\ldots, a_I^{(\sharp_H)})   \in  \mathbb{Z}^{\sharp_H}\;,
\\
&\mathbf{b} = (\vec{b}_1,\ldots, \vec{b}_{d})  \textrm{ with }\vec{b}_{1\leq I \leq d} = (b_I^{(1)},\ldots, b_I^{(\sharp_H)})   \in \left( \mathbb{Z}/2 \right)^{\sharp_H}\;,
\end{split}
\end{align}
the refined index with background shifts $(\mathbf{a},\mathbf{b})$ by
\begin{align}
\begin{split}
&\CI^{(\vec{\alpha},\vec{\gamma};\vec{\beta},\vec{\delta})}_{N;{\rm ref}}
(\mathbf{v};\vec{m},\vec{e};\mathbf{ a},\mathbf{b})
\\
&:=
\eta^{\,2\sum_{i=1}^{\sharp_H}\sum_{I=1}^{d}
W_i\bigl(a^{(i)}_I e_{n+I} + b^{(i)}_I m_{n+I}\bigr)}
\CI^{(\vec{\alpha},\vec{\gamma};\vec{\beta},\vec{\delta})}_{N;{\rm ref}}
(\mathbf{v};m_1,\ldots,m_{n+d},e_1,\ldots,e_{n+d})\; . \label{refined index with (a,b)}
\end{split}
\end{align}
The specific vector $(\mathbf{a}_{\rm d}, \mathbf{b}_{\rm d})$ appearing in \eqref{refined index in charge basis} will be determined below using the compatibility conditions in \eqref{Dehn filling compatibility}. At the level of the gluing equations, the vectors $(\mathbf{a},\mathbf{b})$ can be
implemented through the shifts
\begin{align}
\begin{split}
hol(\gamma_I) &\;\longrightarrow\;
hol(\gamma_I)
-\sum_i a^{(i)}_{I}\,(H_i-2\pi \mathbf{i})\;,
\\
hol(\delta_I) &\;\longrightarrow\;
hol(\delta_I)
+ 2\sum_i b^{(i)}_{I}\,(H_i-2\pi \mathbf{i})\; .
\end{split}
\end{align}
These shifts encode the ambiguity associated with the choice of representatives
of boundary $1$-cycles, since adding loops around internal edges leaves their
classes in $H_1(\partial N;\mathbb{Z})$ unchanged.
Unlike the unrefined index, the refined index is not invariant under these
shifts and may change by an overall factor of the form $\eta^{k}$,
$k\in\mathbb{Z}$, as shown in \eqref{refined index with (a,b)}.
Therefore, the refined 3D index is defined only up to such overall
multiplicative factors. In the field theory $T[N;\vec{\alpha},\vec{\gamma}]$, the shift $\mathbf{a}$
corresponds to a redefinition of background gauge fields,
\begin{align}
(A_I,\vec{B}) \;\longrightarrow\; (A_I -  \vec{a}_I \cdot \vec{B}, \vec{B} ) \; . \label{mixing : a}
\end{align}
Here $A_{I=1,\ldots,d}$ denote the background gauge fields for the Cartan
$U(1)^d $ flavor symmetry in \eqref{flavor of T[N]}, while
$\vec{B}= (B_1, \ldots, B_{\sharp_H})$ are the background gauge fields for the
$U(1)^{\sharp_H}$ flavor symmetry in \eqref{flavor of T[N]}. The $\mathbf{b}$ corresponds to  adding the following background mixed CS terms 
\begin{align}
\delta \mathcal{L} = \frac{1}{2\pi }\sum_{I, i}b^{(i)}_I  B_i \wedge dA_I\;. \label{mixing : b}
\end{align}
\paragraph{The vector spaces $ \mathfrak{F}^{\rm ref}[M]:= \mathfrak{F}^{\rm hard}_{\rm d}[M] \oplus \mathfrak{F}^{\rm dehn}[M] , \mathfrak{F}^{\rm ref}_{\rm dec}[M]$, and $\mathfrak{F}^{\rm ref}_{\rm IR}[M]$}

We now determine the vector space $\mathfrak{F}^{\rm ref}_{\rm IR}[M]$ in \eqref{F-ref[M]-IR}.  This vector space captures the subset of IR flavor symmetries that are not manifest in the 6D construction \eqref{T[M] from 6D}--\eqref{T[M] from 6D-3}, but are manifest in the UV field-theoretic description in \eqref{T[M] theory}. To this end, we first identify the subspace $\mathfrak{F}^{\rm hard}_{\rm d}[M] \subset \mathfrak{F}^{\rm hard}[N]$ appearing in \eqref{F-ref[M]}. We say that a vector $\mathbf{v} \in \Gamma(\mathfrak{F}^{\rm hard}[N])$ is
\emph{compatible with Dehn filling} if there exists a choice of $(\mathbf{a},\mathbf{b})$ such that the following two conditions are satisfied:
\begin{align}
\begin{split}
\text{1)}\;& 
\CI_{N;{\rm ref}}(\mathbf{v};\vec{m},\vec{e};\mathbf{a},\mathbf{b})
\ \text{is invariant under the Weyl $\mathbb{Z}_2$ action }
(m_i,e_i)\mapsto (-m_i,-e_i)\;,
\\
&\text{for each } i=n+1,\ldots,n+d\;,
\\[4pt]
\text{2)}\;& 
\CI_{q^1}(\eta,u_1,\ldots,u_{n+d};\mathbf{a},\mathbf{b})
\Big|_{(\mathrm{adj}\; SU(2)_I)}
\leq -1\;,
\\
&\text{or }
\CI_{N;{\rm ref}}^{(\vec{\a},\vec{\g};\vec{\beta},\vec{\d})}
(\mathbf{v};\vec{m}=\vec{0},\vec{e})
\neq 0
\quad \text{only if } e_{n+I}=0\;,
\\
&\text{for all } I=1,\ldots,d \;.
\label{Dehn filling compatibility}
\end{split}
\end{align}
We  define 
\begin{align}
\mathfrak{F}_{\rm d}^{\rm hard}= \textrm{Span}\left\{
\mathbf{v} \in \Gamma(\mathfrak{F}^{\rm hard}) \;\middle|\;   \mathbf{v}\textrm{ is compatible with Dehn filling}
\right\}\;, \label{F-hard-0}
\end{align}
which represent the Cartan subalgebra of a subgroup $F^{\rm hard}_{\rm d}[M]$ of $F^{\rm hard}[N] = U(1)^{\sharp_H}$
that survives after the gauging of $SU(2)^d$. We expect that 
there exists a choice of $(\mathbf{a},\mathbf{b})$, denoted by
$(\mathbf{a}_{\rm d},\mathbf{b}_{\rm d})$, for which the two conditions in
\eqref{Dehn filling compatibility} are satisfied for all
$\mathbf{v}\in\Gamma(\mathfrak{F}_{\rm d}^{\rm hard})$.
This choice appears in the refined index
\eqref{refined index in charge basis}.

In the above, $\CI_{q^1}(\eta, \vec{u}; \mathbf{a},\mathbf{b})$ denotes the coefficient of $q^1$ of the refined index in the fugacity basis
\begin{align}
\begin{split}
& \CI_{q^1}(\eta; \vec{u}; \mathbf{a},\mathbf{b})
 :=
 \text{coeff}_{q^1}
 \left[
\sum_{\vec{e} \in (\mathbb{Z}/2)^{n+d}} \CI^{(\vec{\alpha},\vec{\g};\vec{\b}, \vec{\d})}_{N;{\rm ref}}
 (\mathbf{v};\vec{m}= \vec{0}, \vec{e};\mathbf{a},\mathbf{b})
 \prod_{i=1}^{n+d} u_i^{e_i} 
 \right]\;.
 \end{split}
\end{align}
The notation $|_{(\mathrm{adj}\; SU(2)_I)}$ denotes the projection onto the
adjoint representation of $SU(2)_I$ and the trivial representation of the
other flavor symmetries. Explicitly,
\begin{align}
\begin{split}
&f(\eta,u_1,\ldots,u_{n+d})\Big|_{(\mathrm{adj}\; SU(2)_I)}
\\
&:=
\left(\oint\frac{d\eta}{2\pi \mathbf{i}\,\eta}\right)
\left(\prod_{i=1}^{n}\oint\frac{du_i}{2\pi \mathbf{i}\,u_i}\right)
\left(\frac{1}{2^{d}}
\prod_{J=1}^{d}
\oint
\frac{du_{n+J}}{2\pi \mathbf{i}\,u_{n+J}}
(2-u_{n+J}^{2}-u_{n+J}^{-2})\right)
\\
&\qquad\times
(1+u_{n+I}^{2}+u_{n+I}^{-2})\,
f\!\left(\eta,u_1^2,\ldots, u^2_{n+d}\right).
\end{split}
\end{align}
The integration contour is along the unit circles, $|\eta|=|u_i|=|u_{n+J}|=1$. Condition~1) is required so that the symmetry generated by the Cartan element
$\mathbf{v}$ commutes with the Weyl $(\mathbb{Z}_2)^d$ subgroup of $SU(2)^d$.
The first condition in 2) ensures that the IR SCFT contains conserved current
multiplets of $SU(2)^d$ that are neutral under both the symmetry generated by
the Cartan element $\mathbf{v}$ and the $U(1)^n$ symmetries, whose Cartan
generators are $\{T_{x_i}\}_{i=1}^n$ in \eqref{Cartans of T[N]}.
In 3D $\mathcal{N}=2$ superconformal field theories, the conserved current
multiplet for a flavor symmetry $F$ contributes to the superconformal index as
\begin{align}
-\,\chi_{\mathrm{adj}}(F)\, q \, ,
\end{align}
where $\chi_{\mathrm{adj}}(F)$ denotes the character of the adjoint
representation of $F$.
For the $SU(2)_I$ ($I=1,\ldots,d$) symmetry to commute with the other symmetries—
namely the $U(1)$ symmetry generated by the Cartan element $\mathbf v$ and the
$U(1)^n$ symmetries with Cartan generators $\{T_{x_i}\}_{i=1}^n$—the conserved
current multiplets of $SU(2)_I$ must be neutral under these abelian symmetries.
In terms of the superconformal index, this implies that the 
contribution from the $SU(2)^d$ conserved current multiplet,
\begin{align}
-\,\chi_{\mathrm{adj}}\!\left(SU(2)^d\right)\, q \; ,
\end{align}
appears without any additional fugacity factors, $\eta$ or
$\{u_i\}_{i=1}^n$, associated with $U(1)_{\mathbf v}$ or $U(1)^n$.
Equivalently, the coefficient of the adjoint character
$\chi_{\mathrm{adj}}\!\left(SU(2)^d\right)$ must be independent of
$\eta$ and $\{u_i\}_{i=1}^n$. The second condition in 2) is intended for the case in which the \(SU(2)_I\) flavor symmetry decouples in the IR. In that case, the first condition in 2) is not imposed, since the corresponding flavor current multiplet is absent in the IR, and the \(SU(2)_I\) symmetry commutes trivially with all other symmetries.

Since the $U(1)^d$ symmetries are the Cartan subalgebra of
$SU(2)^d$, they cannot mix with other abelian symmetries and
cannot have mixed Chern--Simons terms with them. Consequently, only special
values of the mixing parameters $(\mathbf{a},\mathbf{b})$ in
\eqref{mixing : a} and \eqref{mixing : b}—those corresponding to the absence of
such mixing—are allowed.  Only for these values do we expect that
the conditions in 1) and 2) can be simultaneously satisfied. 

We finally define
\begin{align}
\begin{split}
\mathfrak{F}^{\rm ref}_{\rm dec}[M]
&:= 
\operatorname{Span}\Bigl\{
\mathbf{v}\in \Gamma\!\bigl(\mathfrak{F}^{\rm ref}[M]\bigr)
\;\Big|\;
\CI^{b}_{M;{\rm ref}}(\mathbf{v})
=
\eta^{\,f(b)}\,\CI^{b}_{M;{\rm ref}}(\mathbf{0})
\\
&\hspace{4.5cm}
\text{for some linear map }
f: H_1\!\left(\partial M;\tfrac{1}{2}\mathbb{Z}\right)\to \mathbb{Z}
\Bigr\},
\\[4pt]
\mathfrak{F}^{\rm ref}_{\rm IR}[M]
&:=
\mathfrak{F}^{\rm ref}[M]\big/\mathfrak{F}^{\rm ref}_{\rm dec}[M] \;.
\label{F-ref[M]-IR-2}
\end{split}
\end{align}
Here, the vector space \(\mathfrak{F}^{\rm ref}[M] = \mathfrak{F}^{\rm dehn}[M]\oplus \mathfrak{F}^{\rm hard}_{\rm d}[M]\) is defined in \eqref{F-ref[M]}. The condition implies that the superconformal index is independent of the fugacity associated with the vector \(\mathbf{v}\), up to background Chern--Simons terms \eqref{mixing : b} or a redefinition of the background vector field \eqref{mixing : a}.
Physically, $\mathfrak{F}^{\rm ref}_{\rm dec}[M]$ captures the Cartan subalgebra of
the flavor symmetry that decouples in the IR, while
$\mathfrak{F}^{\rm ref}_{\rm IR}[M]$ corresponds to the Cartan subalgebra of the
faithful flavor symmetry acting nontrivially in the IR. Effectively, the refined index can be regarded as a map
\begin{align}
\CI_{M;{\rm ref}}
\;:\;
H_1(\partial M;\frac{1}2 \mathbb{Z})
\times
\Gamma\!\left(\mathfrak{F}^{\rm ref}_{\rm IR}[M]\right)
\;\longrightarrow\;
\mathbb{Z} [[q^{ \frac{1}2}]] [[\eta, \eta^{-1}]] \; .
\end{align}
The vector spaces $\mathfrak{F}^{\rm ref}[M]$, $\mathfrak{F}^{\rm ref}_{\rm dec}[M]$,
and $\mathfrak{F}^{\rm ref}_{\rm IR}[M]$, and hence the 3D refined index
$\CI_{M;{\rm ref}}$, are not intrinsic to $M$; rather, they generally depend on
auxiliary choices in \eqref{auxiliary choices}. These vector spaces are well-defined only when the refined index converges for the chosen data. 

We now present the main conjecture of the paper.
\begin{align}
\begin{split}
&\textbf{Conjecture.}
\\
&\text{For a given $3$-manifold $M$, there exists a choice }
(\mathcal{D}_{0},\mathcal{T}_{0},\vec{\gamma}_{0})
\text{ such that}
\\[4pt]
& \text{i)}\;
\dim\!\Bigl(
\mathfrak{F}^{\rm ref}_{\rm IR}\!\bigl[M;(\mathcal{D},\mathcal{T},\vec{\gamma})\bigr]
\Bigr)
\;\le\;
\dim\!\Bigl(
\mathfrak{F}^{\rm ref}_{\rm IR}\!\bigl[M;(\mathcal{D}_{0},\mathcal{T}_{0},\vec{\gamma}_{0})\bigr]
\Bigr)
< \infty\;,
\\
&\hspace{0.5cm}\text{for all admissible }(\mathcal{D},\mathcal{T},\vec{\gamma})\;,
\\[6pt]
&\text{ii)}\;
\text{For a general admissible choice }(\mathcal{D},\mathcal{T},\vec{\gamma}),
\text{ there exists an injective linear map}
\\
&\quad
F :
\mathfrak{F}^{\rm ref}_{\rm IR}\!\bigl[M;(\mathcal{D},\mathcal{T},\vec{\gamma})\bigr]
\;\hookrightarrow \;
\mathfrak{F}^{\rm ref}_{\rm IR}\!\bigl[M;(\mathcal{D}_{0},\mathcal{T}_{0},\vec{\gamma}_{0})\bigr]
\\
&\quad
\text{such that}
\quad
\CI_{M;{\rm ref}}\!\bigl[\mathbf{v};(\mathcal{D},\mathcal{T},\vec{\gamma})\bigr]
\sim
\CI_{M;{\rm ref}}\!\bigl[F\cdot \mathbf{v};(\mathcal{D}_0,\mathcal{T}_0,\vec{\gamma}_0)\bigr]\;.
\label{main conjecture}
\end{split}
\end{align}
Here we assume that the $3$-manifold $M$ admits at least one choice of $(\mathcal{D},\mathcal{T},\vec{\gamma})$ for which the 3D index converges, and that all triples appearing in the conjecture are restricted to such admissible choices. For some non-hyperbolic 3-manifolds, however, as we will see in Section~\ref{sec : examples}, this condition is not guaranteed, and the conjecture cannot be formulated. The equivalence relation appearing in the conjecture is defined as follows:
\begin{align}
\begin{split}
&\CI^{b}_1 [\mathbf{v}]
\sim
\CI^{b}_2 [F\cdot \mathbf{v}]
\\
&
\text{if there exist a bilinear map }
f : H_1 (\partial M, \frac{1}2 \mathbb{Z}) \times \Gamma(\mathfrak{F}_1) \rightarrow \mathbb{Z}
\text{ such that}
\\
&
\CI^b_1 [\mathbf{v}]
=
\eta^{f(b,\mathbf{v})}
\,\CI^b_2 [F\cdot \mathbf{v}]\;,
\qquad
\forall \mathbf{v}\in \Gamma\!\left(\mathfrak{F}_1\right)\;.
\end{split}
\end{align}
In other words, the conjecture asserts that the refined index is
\emph{independent of the auxiliary choices up to embedding}.
While different choices of $(\mathcal{D},\mathcal{T},\vec{\gamma})$
may lead to different vector spaces
$\mathfrak{F}^{\rm ref}_{\rm IR}[M;(\mathcal{D},\mathcal{T},\vec{\gamma})]$,
all such spaces admit injective linear maps into a maximal one,
$\mathfrak{F}^{\rm ref}_{\rm IR}[M;(\mathcal{D}_{0},\mathcal{T}_{0},\vec{\gamma}_{0})]$.
Under these maps, the refined index is preserved (up to the equivalence
relation defined above). 

In Section~\ref{sec : examples}, for some non-hyperbolic  3-manifolds we
propose the fully refined index. We also identify certain non-hyperbolic
manifolds that do not admit any triple for which the 3D index converges.
For hyperbolic manifolds, as illustrated in the examples, the distinguished
choice $(\mathcal{D}_{0},\mathcal{T}_{0},\vec{\gamma}_{0})$ need not be special
or finely tuned. In practice, a generic choice of $(\mathcal{D},\mathcal{T},\vec{\gamma})$ already achieves the maximal refinement.

We  define
\begin{align}
r_{\rm ref}[M]
:=
\dim\!\Bigl(
\mathfrak{F}^{\rm ref}_{\rm IR}\!\bigl[M;(\mathcal{D}_{0},\mathcal{T}_{0},\vec{\gamma}_{0})\bigr]
\Bigr)
=
\max_{(\mathcal{D},\mathcal{T},\vec{\gamma})}
\dim\!\Bigl(
\mathfrak{F}^{\rm ref}_{\rm IR}\!\bigl[M;(\mathcal{D},\mathcal{T},\vec{\gamma})\bigr]
\Bigr)\;.
\end{align}
The quantity $r_{\rm ref}[M]$, the maximal number of refinements, is itself an invariant of the 3-manifold. In general, checking whether a given choice $(\CD, \CT, \vec{\gamma})$ corresponds to $(\CD_0, \CT_0, \vec{\gamma}_0)$ is a nontrivial task. In the following, we provide a necessary condition for $(\CD, \CT, \vec{\gamma})$ to be a $(\CD_0, \CT_0, \vec{\gamma}_0)$.

\paragraph{Necessary condition for $(\CD_0, \CT_0, \vec{\gamma}_0)$} For closed 3-manifolds without torus boundary, there is a sufficient condition for
$r_{\rm ref}[M]\ge 1$:
\begin{align}
\begin{split}
&\textrm{coeff}_{q^1} (\CI_M) :=\bigl(\text{the coefficient of $q$ in the 3D index }\CI_M\bigr)<0
\\
&\Rightarrow \quad
r_{\rm ref}[M]\ge 1 \, .
\end{split}
\end{align}
This follows from the fact that, in a 3D $\CN=2$ SCFT, the only superconformal multiplet that contributes
a term $-q$ to the superconformal index is the flavor current multiplet.\footnote{The R-charge $R_{\rm geo}$ used in the 3D index \eqref{urefined index as SCI} need not coincide with the
superconformal R-charge, so the $-q$ term is not necessarily due to a flavor
current multiplet. However, such a discrepancy can arise only in the presence
of a flavor symmetry that mixes with the R-symmetry, implying the existence
of a $U(1)$ flavor symmetry.} As an example, consider
\begin{align}
M=(S^3\backslash \mathbf{4}_1)_{[-5\mu+\lambda]}
  =(S^3\backslash \mathbf{5}_2)_{[5\mu+\lambda]}\;.
\end{align}
The 3D index of this 3-manifold is
\begin{align}
\CI_M(q)=1-q-2q^2-q^3+\ldots \; .
\end{align}
Since the coefficient of $q$ is negative, we expect $r_{\rm ref}[M]\ge 1$. 
In fact, one finds that $r_{\rm ref}[M]=1$, and this refinement is manifest for the following choice of auxiliary data:
\begin{align}
\begin{split}
\CD_0 \;&:\; M=(S^3\backslash \mathbf{5}_2)_{[5\mu+\lambda]}\;,
\\
\CT_0 \;&:\; \textrm{the ideal triangulation in \eqref{gluing eqs for 5_2}}\;,
\\
\vec{\gamma}_0 \;&=\; (\mu)\;.
\end{split}
\end{align}
Generalizing the above argument to a 3-manifold $M$ with $n$ torus boundaries,
we claim that
\begin{align}
\text{coeff}_{q^1} \bigl(\CI^{\mathbf{0}}_M\bigr)< -n
\quad \Rightarrow \quad
r_{\rm ref}[M]\ge 1 \, .
\end{align}
The 3D index $\CI_M^{\mathbf{0}}$ counts local operators of the theory
$T[M;\vec{\alpha}]$ that are neutral under the $U(1)^n$ flavor symmetry,
whose Cartan generators are $T_{x_1},\ldots,T_{x_n}$ in
\eqref{F-tori[M]}. The flavor current multiplets associated with
$U(1)^n$ contribute $-n q$ to the index. Therefore, if the coefficient of
$q$ in $\CI_M^{\mathbf{0}}$ is smaller than $-n$, it implies the presence
of additional flavor symmetry commuting with $U(1)^n$. Such an enhanced
flavor symmetry leads to $r_{\rm ref}[M]\ge 1$.
Further generalizing the argument, we obtain the following necessary condition:
\begin{align} \label{maximal refinement necessary}
\begin{split}
&\text{If }(\CD,\CT,\vec{\gamma}) \text{ is a choice corresponding to }
(\CD_0,\CT_0,\vec{\gamma}_0) \text{ in the conjecture \eqref{main conjecture}}\;,\\
&\text{coeff}_{q^1}\left( \CI^{\mathbf{0}}_{M;{\rm ref}}|_{\rm neutral}\right)
\;\ge\;
- \dim\!\bigl(\mathfrak{F}_{\rm IR}[M;(\CD,\CT,\vec{\gamma})]\bigr) - n \, .
\end{split}
\end{align}
Here $\CI^{\mathbf{0}}_{M;{\rm ref}}\big|_{\rm neutral}$ denotes the part of
$\CI^{\mathbf{0}}_{M;{\rm ref}}(\mathbf{v})$ that is independent of $\eta$ for
any $\mathbf{v}\in \Gamma\bigl(\mathfrak{F}^{\rm ref}[M;\CD,\CT,\vec{\gamma}]\bigr)$.
It counts local operators of $T[M;\vec{\alpha}]$ that are neutral under both
the $U(1)^n$ flavor symmetry in \eqref{F-tori[M]} and the additional symmetries
in \eqref{F-ref[M]}. If the inequality is violated, it implies the existence of
further emergent flavor symmetries.

\paragraph{$(\CD_0,\CT_0,\vec{\gamma}_0)$ under integral Dehn filling}
Let $N$ be a $3$-manifold with $(n+1)$ torus boundaries, and let $(\CD_0,\CT_0,\vec{\gamma}_0)$ be the choice appearing in the conjecture \eqref{main conjecture}, which yields the maximal refinement. Consider a $3$-manifold $M$ obtained from $N$ by
Dehn filling along the $(n+1)$-th boundary,
\begin{align}
M = N_{[p\gamma_{n+1}+\delta_{n+1}]}\,,
\end{align}
where $\gamma_{n+1}\in H_1(T_{n+1},\mathbb{Z})$ is a non-closable cycle.
Then, for sufficiently large $|p|$, we expect that the following data
\begin{align}
\begin{split}
&\widetilde{\CD}_0
= \bigl(\text{the Dehn surgery presentation of $M$ obtained by combining }
\CD_0 \text{ for $N$}
\\
&\qquad\quad \text{with the filling } M = N_{[p\gamma_{n+1}+\delta_{n+1}]}\bigr)\;,
\\
&\widetilde{\CT}_0 = \CT_0\;, \qquad
\widetilde{\vec{\gamma}}_0
= \vec{\gamma}_0 \cup \{\gamma_{n+1}\}
\end{split}
\end{align}
provide the choice appearing in the conjecture \eqref{main conjecture}. The reason is that, at sufficiently large Chern--Simons level $p$, the quantum fluctuations of the gauge field introduced by the Dehn filling are suppressed. Consequently, the naive UV flavor-symmetry counting remains reliable, and no accidental symmetry emerges from the gauging.

As an example, consider $N=S^3\backslash \mathbf{5}^2_1$ (the Whitehead link complement), with the following choices:
\begin{align}
\begin{split}
&\CD_0 \;:\; N = S^3\backslash \mathbf{5}^2_1 \qquad (d=0)\;,
\\
&\CT_0 \;:\;\text{an ideal triangulation provided by \texttt{SnapPy}; see \eqref{Gluing equations : 5_2 and 5^2_1}}\;,
\\
&\vec{\gamma}_0 = \emptyset\;.
\end{split}
\end{align}
There are four ideal tetrahedra in $\CT_0$, with one easy internal edge $E_1$ and one hard internal edge $H_1$. The hard edge gives rise to an additional $U(1)^{\sharp_H=1}$ symmetry, and this symmetry turns out not to decouple in the IR. Thus, we have
\begin{align}
\mathfrak{F}^{\rm ref}_{\rm IR}[S^3\backslash \mathbf{5}^2_1;\CD_0,\CT_0,\vec{\gamma}_0]
=
\mathfrak{F}^{\rm hard}[S^3\backslash \mathbf{5}^2_1;\CD_0,\CT_0,\vec{\gamma}_0] = \textrm{Span} \{T_{h_1}\}\;.
\end{align}
We further claim that this choice realizes the maximal refinement. In particular,
\begin{align}
\begin{split}
r_{\rm ref}[S^3\backslash \mathbf{5}^2_1]
&=
\dim\!\Bigl(
\mathfrak{F}^{\rm ref}_{\rm IR}[S^3\backslash \mathbf{5}^2_1;\CD_0,\CT_0,\vec{\gamma}_0]
\Bigr)
=1\;.
\end{split}
\end{align}
Now consider
\begin{align}
M_p := (S^3\backslash \mathbf{5}^2_1)_{[p\mu_2+\lambda_2]}\;,
\qquad p\ge 1\,.
\end{align}
Here $\mu_2$ and $\lambda_2$ denote the meridian and longitude of the second boundary torus of the link complement.
For $M_p$, we choose
\begin{align}
\begin{split}
&\widetilde{\CD} \;:\; N = S^3\backslash \mathbf{5}^2_1\;,
\qquad \{P_I\tilde{\gamma}_I+Q_I\tilde{\delta}_I\}_{I=1}^{d=1} = \{p\mu_2+\lambda_2\}\;,
\\
&\widetilde{\CT} \;:\; \text{an ideal triangulation provided by \texttt{SnapPy}; see \eqref{Gluing equations : 5_2 and 5^2_1}}\;,
\\
&\widetilde{\vec{\gamma}} = (\mu_2)\;.
\end{split}
\end{align}
For sufficiently large $p$, we expect that this choice provides the distinguished data for $M_p$, realizing the maximal refinement. Since the Dehn filling is integral (i.e.\ $(P_1,Q_1)=(p,1)$ with $(\tilde{\gamma}_1,\tilde{\delta}_1)=(\mu_2,\lambda_2)$), there is no additional flavor symmetry arising from the Dehn filling. Hence,
\begin{align}
\mathfrak{F}^{\rm ref}[M_p;\widetilde{\CD},\widetilde{\CT},\widetilde{\vec{\gamma}}]
=
\mathfrak{F}^{\rm hard}_{\rm d} [M_p;\widetilde{\CD},\widetilde{\CT},\widetilde{\vec{\gamma}}]\;.
\end{align}
It turns out that $\mathbf{v}=T_{h_1}$ is not compatible with the Dehn filling, and therefore
\begin{align}
\mathfrak{F}^{\rm ref}_{\rm IR}[M_p;\widetilde{\CD},\widetilde{\CT},\widetilde{\vec{\gamma}}]
=
\{\mathbf{0}\}\;.
\end{align}
Accordingly, we expect that
\begin{align}
r_{\rm ref}[M_p]=0
\qquad \text{for sufficiently large } p\;.
\end{align}
The manifold $M_p$ is homeomorphic to the once-punctured torus bundle with monodromy matrix
\[
\varphi = LR^p \in SL(2,\mathbb{Z})\;,
\]
where
\begin{align}
L=\begin{pmatrix}1&1\\0&1\end{pmatrix}\;,
\qquad
R=\begin{pmatrix}1&0\\1&1\end{pmatrix}\;.
\end{align}
One can indeed check that, in an ideal triangulation provided by \texttt{SnapPy}, this once-punctured torus bundle has no hard internal edge for $p\ge 2$. In \texttt{SnapPy}, the manifold $M_p$ can be input as `$b++ L\overbrace{R\cdots R}^{p\ \text{times}}$'.
This is consistent with the expectation that $r_{\rm ref}[M_p]=0$ for sufficiently large $p$, and we therefore claim that
\begin{align}
r_{\rm ref}[M_p]=0
\qquad \text{for } p\ge 2\;.
\end{align}
By the same reasoning, one can conclude that
\begin{align}
r_{\rm ref}[(S^3\backslash \mathbf{4}_1)_{[p \mu+\lambda]}]=0\;, 
\quad 
r_{\rm ref}[(S^3\backslash \mathbf{5}_2)_{[p \mu+\lambda]}]=1\;, 
\quad 
\text{for sufficiently large } |p| \;.
\label{r-ref for 4-1-p and 5-2-p}
\end{align}
\subsection{Refined 3D Index from Normal Surface Counting} \label{sec : index from normal surface counting}
In this section, we briefly review the construction of 3D index via normal surface counting proposed in \cite{GHRS15} and \cite{GHHR16}. Then we follow the same procedure to reconstruct the refined 3D index and prove the invariance under Pachner 2-3 moves. 
\subsubsection{\(Q\)-Normal Surface theory}
Let \(N\) be a 3-manifold with \((n+d)\) torus boundaries equipped with an ideal triangulation \(\CT\) consisting of \(r\) tetrahedra \(\Delta_a\). A \emph{normal surface} is an embedded surface in \(N\) that intersects each ideal tetrahedron in a collection of elementary disks. Up to isotopy, these disks come in two topologies: triangles, which cut off the ideal vertices of the tetrahedron, and quadrilaterals, which separate pairs of edges of the tetrahedron as illustrated in \ref{fig : quad types}.

In classical normal surface theory, surfaces are strictly embedded, restricting each tetrahedron to contain at most one type of quadrilateral disk to prevent self-intersection. \(Q\)-normal surfaces relax this constraint, allowing multiple quadrilateral types to coexist and intersect along \emph{double arcs}. These intersections are crucial in our construction of refined index. Furthermore, because we work with ideal triangulations whose vertices form torus boundaries, these surfaces do not terminate but rather spin asymptotically around the boundary tori. This spinning behavior geometrically encodes the boundary magnetic and electric charges via the peripheral components. Because the triangular disks are completely determined by how the surface meets the boundary of \(N\), the non-trivial bulk geometry of the surface is entirely captured by the quadrilateral disks. For these reasons, the 3D index counts specifically \emph{spun \(Q\)-normal surfaces}. 

\begin{figure}[h]
    \centering
    \includegraphics[width=0.7\linewidth]{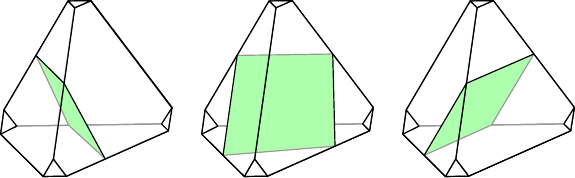}
    \caption{Quad types in single tetrahedra}
    \label{fig : quad types}
\end{figure}
There are three types of quadrilaterals in a single tetrahedron and each of them exactly separates the corresponding pair of opposite edges. Hence, each quad can be labeled by edge parameters \((Z_a, Z_a', Z_a'') \) and \(Q\)-normal classes are represented by the quad coordinate \((f_a, g_a, h_a)\). That is, for example, \(S=(0,0,1,2,1,0)\) is a surface that intersects with the first tetrahedron in third type once, intersects with the second tetrahedron in first type twice and second type once. 

For these isolated quadrilaterals to form a globally well-defined, continuous surface across the entire 3-manifold, they must glue together consistently at the boundaries of the tetrahedra. This geometric requirement imposes a set of linear constraints known as the \emph{\(Q\)-matching equations}. Specifically, across any shared triangular face between two adjacent tetrahedra, the quadrilateral disks intersect the face in distinct topological types of normal arcs. The \(Q\)-matching equations dictate that for each of the three possible arc types on the shared face, the number of arcs orginating from the quadrilateral disks in one tetrahedron must exactly equal the number originating from the adjacent tetrahedron. 

Algebraically, the real solution space to the \(Q\)-matching equations is generated by three fundamental types of solutions: edge solutions \(C_a\), peripheral curve solutions \(M_i\) and \(L_i\), and tetrahedral solutions \(T_a\). The edge solutions \(C_a\) and peripheral curve solutions \(M_i\), and \(L_i\) correspond directly to the internal edge classes and the chosen boundary curves (meridians \(hol(\alpha_i)\), \(hol(\gamma_I)\), and longitudes \(hol(\beta_i)\), \(hol(\delta_I)\), for each cusp \(i\), and \(I\)), with their coordinates completely determined by the standard Thurston gluing equations. Meanwhile, the tetrahedral solutions \(T_a\) represent local, internal degrees of freedom. Explicitly, the tetrahedral solution for the \(a\)-th ideal tetrahedron is defined as \(T_a=(0,\dots,0, 1,1,1, 0,\dots,0)\), which uniformly assigns a unit weight to all three of its quadrilateral types. 

By adapting this geometric basis, the solution of \(Q\)-matching equations are in the following form: 
\begin{equation}
    S = \sum_{i=1}^{n+d} (A_i M_i + B_i L_i) + \sum_{j=1}^{r-n-d}x_j C_j+\sum_{a=1}^r y_a T_a
    \label{surface_basis}
\end{equation}
At this stage, one must be careful. While \eqref{surface_basis} spans the \emph{real} solution space of the \(Q\)-matching equations, geometrically valid spun normal surfaces require strictly integer quad coordinates \((f,g,h)\in\mathbb{Z}^{3r}\). Therefore, we must intersect this real solution space with integer lattice. 

Crucially for \(S\) to be an integer \(Q\)-normal class, the coefficients \(A_i\), \(B_i\), \(x_j\), and \(y_a\) are not necessarily required to be strict integers. Because the basis vectors possess specific even/odd coordinate structures, valid integer surfaces can be formed using half-integer coefficients, provided that their fractional parts cancel out modulo \(\mathbb{Z}\). As highlighted in Remark 7.3, 7.5, and 7.6 in \cite{GHHR16}, for any valid integer \(Q\)-normal class \(S\), its double \(2S\) is always an integer linear combination of the basis elements, meaning the coordinates naturally carry a modulo 2 parity structure. 

\subsubsection{Refined 3D Index from Normal Surface Counting}
To construct the 3D index from normal surface, we define the \emph{formal Euler characteristic}: 
\begin{equation}
    -\chi(S) = \sum_{j=1}^{r-n-d}2x_j +\sum_{a=1}^r y_a.
\end{equation}
which gives the usual Euler characteristic for embedded closed and spun normal surfaces. 

Following the structure of Section 8 in \cite{GHHR16} with \(a=0\), the 3D index is formulated as a state-sum over all geometrically valid integer \(Q\)-normal surfaces. The contribution of each surface \(S\) is determined by local weights assigned to each tetrahedron, dressed by global topological factors.

For a single ideal tetrahedron \(\Delta_a\), let \((f_a, g_a, h_a)\) denote its non-negative integer quad coordinates. The local tetrahedral weight is given by the tetrahedron index \(J_{\Delta}(f_a, g_a, h_a)\), defined as 
\begin{equation}
    J_\Delta(a,b,c) = (-q^{1/2})^{-b}\CI_\Delta(b-c,a-b)
\end{equation}
which is invariant under all permutations of its arguments. Note that it satisfies an identity 
\begin{equation}
    J_\Delta(a+s, b+s, c+s) = (-q^{1/2})^{-s}J_\Delta(a,b,c).     
\end{equation}
We then define the single normal surface index:
\begin{equation}
    J(S) = (-q^{1/2})^{-\chi(S)}\prod_{a=1}^r J_{\Delta}(f_a, g_a, h_a). 
\end{equation}

Note that the tetrahedral solutions \(T_a\) in \eqref{surface_basis} act as trivial redundancies. Shifting a surface by \(T_a\) uniformly increases the quad coordinates \((f_a, g_a, h_a)\) of the \(a\)-th tetrahedron by \(1\), which multiplies the local weight \(J_\Delta\) by \((-q^{1/2})^{-1}\). However, it simultaneously increases the formal Euler characteristic \(-\chi(S)\) by \(1\), multiplying the global prefactor by \((-q^{1/2})^1\). These two effects perfectly cancel, ensuring that \(J(S + T_a) = J(S)\). 

Because the surface weight is strictly invariant under the addition of tetrahedral solutions, we must evaluate the state-sum over equivalence classes of normal surfaces modulo these \(T_a\) shifts to avoid infinite over counting. The standard 3D index evaluated at a fixed boundary topological sector \(b \in H_1(\partial N; \mathbb{Z}/2)\) is then defined by summing over the equivalence classes of all valid integer \(Q\)-normal surfaces \([S]\) whose boundary class matches \(b\):
\begin{equation}
    \CJ_N(b) = \sum_{[S],\, \partial S=b} J(S) = \sum_{[S],\, \partial S=b} (-q^{1/2})^{-\chi(S)} \prod_{a=1}^r J_{\Delta}(f_a, g_a, h_a).
\end{equation}
Assuming the edge solutions \(\mathbf{D}=(D_1, D_2, \cdots, D_{r-n-d})\) are taken to be linearly independent, then the index can be calculated as 
\begin{equation}
    \CJ_N(b) = \sum_{\mathbf{k}\in (\mathbb{Z}/2)^{r-n-d}}J( \mathbf{A}\cdot\mathbf{M}+\mathbf{B}\cdot\mathbf{L}+\mathbf{k}\cdot \mathbf{D}).
\end{equation}

Now suppose the edge solutions \(\mathbf{D}=(\mathbf{H}, \mathbf{E})\) are taken to maximally contain linearly independent easy edges \(\mathbf{E} = (E_1, E_2, \cdots, E_{\sharp _E})\) and remaining hard edges \(\mathbf{H}=(H_1, H_2, \cdots, H_{\sharp_H})\). In this construction, the hard edge solutions \(\mathbf{H}\) act as the primary sources of double arcs, representing intrinsic singularities on the \(Q\)-normal surfaces that cannot be removed by isotopy. While the unrefined index integrates over all such configurations, our refined 3D index distinguishes normal surfaces by their singular structure, weighting each surface by the number of hard edge components it contains:
\begin{equation}
    \mathcal{J}_{N;\text{ref}}(b;\mathbf{v}) = \sum_{\mathbf{k} \in (\mathbb{Z}/2)^{r-n-d}} J(\mathbf{A}\cdot\mathbf{M}+\mathbf{B}\cdot\mathbf{L}+ \mathbf{k} \cdot \mathbf{D}) \prod_{i=1}^{\sharp_H} \eta^{2\mathbf{k}\cdot\mathbf{W}} \;.
\end{equation}
By grading the state-sum by the hard-edge weights \(k_{i}\), the refined index captures the contribution of the double-arc singularities generated by the internal bulk geometry. This provides a more sensitive topological invariant that resolves the divergent sectors of the standard index and reflects the accidental symmetry structure of the underlying 3D gauge theory.

\begin{theorem}[Equivalence between two Refined Indices]
For \(b\in H_1(\partial N,\mathbb{Z}/2)\), and \(\mathbf{v}\in\Gamma(\mathfrak{F}^{\textrm{hard}}(N))\)
\begin{equation}
    \mathcal{I}_{N;\textrm{ref}}^b(\mathbf{v}) = \CJ_{N;\textrm{ref}}(b;\mathbf{v})
\end{equation}
\end{theorem}
\begin{proof}
    Lets denote the Neumann-Zaiger matrix as 
    \begin{equation}
        g_{\textrm{NZ}}= \begin{pmatrix}
            \mathbf{P}&\mathbf{Q}\\
            \mathbf{R}&\mathbf{S}
        \end{pmatrix}
    \end{equation}
    such that in terms of the fundamental solutions, 
    \begin{equation}
        \begin{split}
            P_{ia}&=\begin{cases}
            M_{i,\,3a-2}-M_{i,\,3a-1},\quad & 1\leq i\leq n+d\\
            D_{j,\,3a-2}-D_{j,\,3a-1},\quad&n+d+1\leq i\leq r
            \end{cases}\\
            Q_{ia}&=\begin{cases}
            M_{i,\,3a}-M_{i,\,3a-1},\quad & 1\leq i\leq n+d\\
            D_{j,\,3a}-D_{j,\,3a-1},\quad&n+d+1\leq i\leq r
            \end{cases}\\
            R_{ia}&=\begin{cases}
            L_{i,\,3a-2}/2-L_{i,\,3a-1}/2,\quad & 1\leq i\leq n+d\\
            \Gamma_{j,\,3a-2}-\Gamma_{j,\,3a-1},\quad&n+d+1\leq i\leq r
            \end{cases}\\
            S_{ia}&=\begin{cases}
            L_{i,\,3a}/2-L_{i,\,3a-1}/2,\quad & 1\leq i\leq n+d\\
            \Gamma_{j,\,3a}-\Gamma_{j,\,3a-1},\quad&n+d+1\leq i\leq r
            \end{cases},
        \end{split}
    \end{equation}
    where \(j=i-n-d\) and \(\Gamma\) is a solutions to make the Neumann-Zaiger matrix to be true. Since \(g_{NZ}\) matrix is an symplectic matrix, its inverse is 
    \begin{equation}
        g_{\textrm{NZ}}^{-1} = \begin{pmatrix}
            \mathbf{S}^T & -\mathbf{Q}^T\\
            -\mathbf{R}^T&\mathbf{P}^T
        \end{pmatrix}. 
    \end{equation}
    Similarly, we can express Affine shift vectors as 
    \begin{equation}
        \begin{split}
            \boldsymbol{\nu}_{x,\,i}&=\begin{cases}
                \sum_{a=1}^rM_{i,\,3a-1},\quad&1\leq i\leq n+d\\
                \sum_{a=1}^rD_{j,\,3a-1}-2,\quad&n+d+1\leq i\leq r
            \end{cases}\\
            \boldsymbol{\nu}_{p,\,i}&=\begin{cases}
                \sum_{a=1}^rL_{i,\,3a-1}/2,\quad&1\leq i\leq n+d\\
                \sum_{a=1}^r\Gamma_{j,\,3a-1},\quad&n+d+1\leq i\leq r
            \end{cases}
        \end{split}
    \end{equation}
    On the other hand, a normal surface \(S=\mathbf{A}\cdot\mathbf{M}+\mathbf{B}\cdot\mathbf{L}+ \mathbf{k} \cdot \mathbf{D}\) has quad coordinate coefficients  
    \begin{equation}
        \begin{split}
            f_a &= \sum_{i=1}^{n+d}(A_i M_{i,\,3a-2}+B_i L_{i,\,3a-2})+\sum_{j=1}^{r-n-d}k_j D_{j,\,3a-2}, \\
            g_a &=  \sum_{i=1}^{n+d}(A_i M_{i,\,3a-1}+B_i L_{i,\,3a-1})+\sum_{j=1}^{r-n-d}k_j D_{j,\,3a-1}, \\
            h_a &=  \sum_{i=1}^{n+d}(A_i M_{i,\,3a}+B_i L_{i,\,3a})+\sum_{j=1}^{r-n-d}k_j D_{j,\,3a}.
        \end{split}
    \end{equation}
    After substituting \(A_i =e_i\), \(B_i=-m_i/2\), and setting \(e_{j+n+d}=-k_j\), \(m_{j+n+d}=0\) for \(j=1,2,\cdots,r-n-d\), and setting \(\boldsymbol{\kappa}=(\mathbf{m},\mathbf{e})^T=(m_1, \cdots, m_r, e_1, \cdots, e_r)^T\), one finds that 
    \begin{equation}
        \begin{split}
            f_a - g_a &= \sum_{i=1}^{n+d}(A_i P_{ia}+2B_i R_{ia})+\sum_{j=1}^{r-n-d}k_j P_{j+n+d,\,a}=(\mathbf{P}^T \mathbf{e}-\mathbf{R}^T\mathbf{m})_a\\
            g_a - h_a &= \sum_{i=1}^{n+d}(-A_i Q_{ia}-2B_i S_{ia})-\sum_{j=1}^{r-n-d}k_j Q_{j+n+d,\,a}=(-\mathbf{Q}^T\mathbf{e}+\mathbf{S}^T \mathbf{m})_a\\
            \sum_{a=1}^rg_a&=-\mathbf{m}\cdot\boldsymbol{\nu}_p +\mathbf{e}\cdot \boldsymbol{\nu}_x-\chi(S).        \end{split}
    \end{equation}
    Then further simplifying as such, 
    \begin{equation}
        \begin{split}
            J(S) &= (-q^{1/2})^{-\chi(S)}\prod_{a=1}^r J_{\Delta}(f_a,g_a,h_a) \\
            &=(-q^{1/2})^{-\chi(S)}\prod_{a=1}^r (-q^{1/2})^{-g_a} \CI_{\Delta}(g_a - h_a, f_a - g_a)\\
            &=(-q^{1/2})^{\mathbf{m}\cdot\boldsymbol{\nu}_p-\mathbf{e}\cdot\boldsymbol{\nu}_p}\prod_{a=1}^r \CI_\Delta((-\mathbf{Q}^T\mathbf{e}+\mathbf{S}^T \mathbf{m})_a,(\mathbf{P}^T \mathbf{e}-\mathbf{R}^T\mathbf{m})_a)\\
            &= (-q^{1/2})^{\mathbf{m}\cdot\boldsymbol{\nu}_p-\mathbf{e}\cdot\boldsymbol{\nu}_p}\prod_{a=1}^r \CI_\Delta((g_{\textrm{NZ}}^{-1}\boldsymbol{\kappa})_a,(g_{\textrm{NZ}}^{-1}\boldsymbol{\kappa})_{r+a})
        \end{split}
    \end{equation}
    one obtain 
    \begin{equation}
        \CJ_{N;\textrm{ref}}(b;\mathbf{v})=\sum_{\mathbf{k}\in(\mathbb{Z}/2)^{r-n-d}}J(\mathbf{A}\cdot\mathbf{M}+\mathbf{B}\cdot\mathbf{L}+\mathbf{k}\cdot\mathbf{D}) \prod_{i=1}^{\sharp_H}\eta^{2\mathbf{k}\cdot\mathbf{W}} = \CI_{N;\textrm{ref}}^b(\mathbf{v}).
    \end{equation}
\end{proof}

\subsubsection{Refined 3D index under 2-3 moves}
In this subsection, we prove that the number of hard internal edges does not increase under a Pachner \(2\to 3\) move, and use this together with the pentagon identity to establish that the refined 3D index is invariant under \(2\to 3\) moves. 
\begin{figure}[h]
    \centering
    \includegraphics[width=0.9\linewidth]{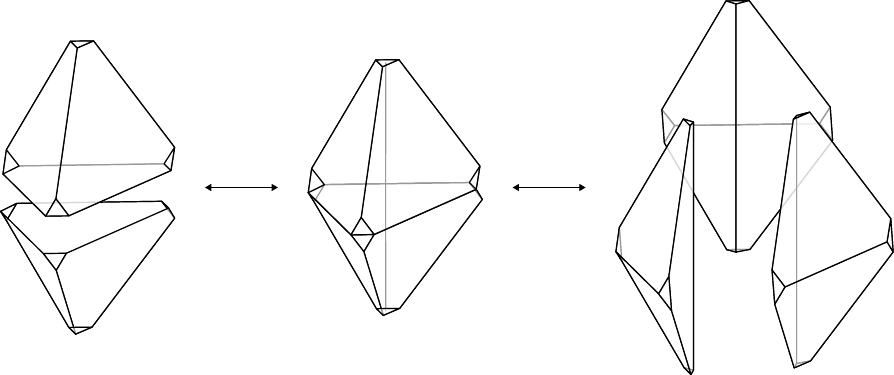}
    \caption{Pachner 2-3 move. }
    \label{fig : 2-3 move}
\end{figure}
    \paragraph{Setup.} Let \(\CT\) be an ideal triangulation of \(N\) with \(r\) tetrahedra, equipped with edge shape parameters \(Z_a\), \(Z_a'\), and \(Z_a''\) for \(a=1,2,\cdots ,r\). Let \(\widetilde{\CT}\) be the ideal triangulation obtained from \(\CT\) by a Pachner \(2\to 3\) move, which replaces the \((r-1)\)-th and \(r\)-th tetrahedra by three new tetrahedra; without loss of generality the move is applied to the last two tetrahedra. The triangulation \(\widetilde{\CT}\) has \(r+1\) tetrahedra, with edge parameters \(\widetilde{Z}_a\), \(\widetilde{Z}_a'\), and \(\widetilde{Z}_a''\) for \(a=1,2,\cdots , r+1\). 

    The first \(r-2\) tetrahedra are unaffected: 
    \begin{equation}
        Z_a = \widetilde{Z}_a, \quad Z_a' = \widetilde{Z}_a', \quad Z_a'' = \widetilde{Z}_a'', \quad a=1,2,\cdots,r-2
    \end{equation}
    while the last two old tetrahedra are expressed in terms of the three new ones via 
    \begin{equation}
        \begin{pmatrix}
            Z_{r-1}\\
            Z_{r-1}'\\
            Z_{r-1}''\\
            Z_r\\
            Z_r'\\
            Z_r''
        \end{pmatrix} = \begin{pmatrix}
            0&0&0&0&1&0&0&0&1\\
            0&1&0&0&0&1&0&0&0\\
            0&0&1&0&0&0&0&1&0\\
            0&1&0&0&0&0&0&0&1\\
            0&0&1&0&1&0&0&0&0\\
            0&0&0&0&0&1&0&1&0
        \end{pmatrix}\begin{pmatrix}
            \widetilde{Z}_{r-1}\\
            \widetilde{Z}_{r-1}'\\
            \widetilde{Z}_{r-1}''\\
            \widetilde{Z}_r\\
            \widetilde{Z}_r'\\
            \widetilde{Z}_r''\\
            \widetilde{Z}_{r+1}\\
            \widetilde{Z}_{r+1}'\\
            \widetilde{Z}_{r+1}''
        \end{pmatrix}.\label{23 move edge parameter transformation}
    \end{equation}
    The move introduces a new internal edge, the \emph{central edge}
    \begin{equation}
        D_c:= \widetilde{Z}_{r-1}+ \widetilde{Z}_r + \widetilde{Z}_{r+1}, 
    \end{equation}
    which is manifestly easy: it carries exactly the unprimed quad coordinate in each of the three new tetrahedra and zero in all others. 

    \paragraph{Transformation of internal edges.} Let \(D = \sum_{a=1}^r (f_a Z_a + g_a Z_a' + h_a Z_a'')\) be an internal edge of \(\CT\). Substituting via \eqref{23 move edge parameter transformation}, the same edge in \(\widetilde{T}\) is 
    \begin{equation}
        \begin{split}
            \widetilde{D}=&\sum_{a=1}^{r-2}\bigl(f_a\widetilde{Z}_a+g_a\widetilde{Z}_a'+h_a\widetilde{Z}_a''\bigr)+\underbrace{(f_r+g_{r-1})}_{\widetilde{g}_{r-1}}\widetilde{Z}_{r-1}'+\underbrace{(g_r+h_{r-1})}_{\widetilde{h}_{r-1}}\widetilde{Z}_{r-1}''\\
            &+\underbrace{(f_{r-1}+g_r)}_{\widetilde{g}_r}\widetilde{Z}_r'+\underbrace{(g_{r-1}+h_r)}_{\widetilde{h}_r}\widetilde{Z}_r''+\underbrace{(h_{r-1}+h_r)}_{\widetilde{g}_{r+1}}\widetilde{Z}_{r+1}'+\underbrace{(f_{r-1}+f_r)}_{\widetilde{h}_{r+1}}\widetilde{Z}_{r+1}''.
            \label{Dtilde}
        \end{split}
    \end{equation}
    Note that the unprimed quad coordinate \(\widetilde{f}_a\) vanishes for each new tetrahedron \(a\in \{r-1,r,r+1\}\).  

    We are now ready to state and prove the key structural lemma. 
    \begin{lemma}[Embedding of hard-edge symmetries under 2-3 moves]\label{lem:easy-23}
    Let \(\mathfrak{F}^{\rm hard}[N;\CT]\) denote the Cartan subalgebra associated with the hard internal edges of the triangulation \(\CT\). A Pachner 2-3 move \(\CT \to \widetilde{\CT}\) induces a natural injective linear map (an embedding) between the hard-edge symmetry algebras:
    \begin{equation}
        \iota: \mathfrak{F}^{\rm hard}[N;\widetilde{\CT}] \hookrightarrow \mathfrak{F}^{\rm hard}[N;\CT] \;.
    \end{equation}
    Combinatorially, this implies that the number of easy edges strictly increases, \(\sharp_E(\widetilde{\CT}) \geq \sharp_E(\CT)+1\), while the number of hard edges is non-increasing, \(\sharp_H(\widetilde{\CT}) \leq \sharp_H(\CT)\).
    \end{lemma}
    
    \begin{proof}
        Since the dimension of the Cartan subalgebra \(\mathfrak{F}^{\rm hard}[N;\CT]\) is precisely the number of hard edges \(\sharp_H\), and the total number of linearly independent internal edges obeys the relation \(\sharp_H(\CT) + \sharp_E(\CT) = r - n - d\) before the move and \(\sharp_H(\widetilde{\CT}) + \sharp_E(\widetilde{\CT}) = (r+1) - n - d\) after the move, the existence of the embedding \(\iota\) follows directly if we can establish the strict inequality \(\sharp_E(\widetilde{\CT})\geq \sharp_E(\CT)+1\). We demonstrate this by explicitly constructing the easy edge subspace for \(\widetilde{\CT}\).

        \medskip
        \noindent\textit{Step 1: The central edge generates a new easy direction.}\\
        The newly introduced edge \(D_c=\widetilde{Z}_{r-1}+\widetilde{Z}_{r}+\widetilde{Z}_{r+1}\) has quad coordinates \((1,0,0)\) in tetrahedra \(r-1\), \(r\), and \(r+1\), and \((0,0,0)\) in all others. Because at most one coordinate is nonzero per tetrahedron, \(D_c\) is manifestly easy. Moreover, \(D_c\) is linearly independent of any image \(\widetilde{E}\) of an old internal edge, since any such \(\widetilde{E}\) has a vanishing unprimed coordinate in all three new tetrahedra. Hence, \(D_c\) contributes one entirely new dimension to the easy edge subspace that was not present in \(\CT\). 

        \medskip
        \noindent\textit{Step 2: The old easy edge subspace survives.}\\
        Let \(E\) be an easy internal edge of \(\CT\) with quad coefficients \((f_a,g_a,h_a)_{a=1}^r\). By the easy constraint, at most one of \(\{f_a,g_a,h_a\}\) is nonzero for each \(a\). We claim there exists a valid easy internal edge \(\widetilde{E}'\) for \(\widetilde{\CT}\) within the \(\BZ\)-span of \(\widetilde{E}\), \(D_c\), and the tetrahedral solutions \(T_j\) of \(\widetilde{\CT}\) (recall that \(T_j\) has quad coordinates \((1,0,0)\) in the \(j\)-th tetrahedron and \((0,0,0)\) elsewhere). 

        From \eqref{Dtilde}, the quad coordinates of \(\widetilde{E}\) in the three new tetrahedra are:
        \begin{align}
            \begin{array}{ll}
                \textrm{tetrahedron }r-1:& (0,\; f_r + g_{r-1},\;g_r+h_{r-1}),\\
                \textrm{tetrahedron }r:&(0,\;f_{r-1}+g_r,\;g_{r-1}+h_r),\\
                \textrm{tetrahedron }r+1:&(0,\;h_{r-1}+h_r,\;f_{r-1}+f_r).
            \end{array}\label{new tetrahedron coordinates}
        \end{align}
        If every one of these triples has at most one nonzero entry, then \(\widetilde{E}\) is already strictly easy and we set $\widetilde{E}' = \widetilde{E}$. Otherwise, at least one triple has two nonzero entries. Under the original easy edge constraint on $\CT$, the admissible configurations that produce two simultaneously nonzero entries in \eqref{new tetrahedron coordinates} fall exactly into the following three cases:

        \begin{enumerate}
            \item[(i)] \(f_r\neq 0\) and \(h_{r-1}\neq 0\), all others zero (tetrahedra \(r-1\) and \(r+1\) temporarily violate easiness).\\
            The linear combination
            \begin{equation}
                \widetilde{E}':= \widetilde{E}- \textrm{min}(f_r,\;h_{r-1})(T_{r-1}+T_{r+1}-D_c)\label{eq: modification of easy edge first}
            \end{equation}
            is a valid internal edge with quad coordinates in the three new tetrahedra: 
            \begin{align}
                \begin{array}{ll}
                    \textrm{tetrahedron }r-1:& (0,\;|f_r-h_{r-1}|,\;0),\\
                    \textrm{tetrahedron }r:&(\textrm{min}(f_r,\;h_{r-1}),\;0,\;0),\\
                    \textrm{tetrahedron }r+1:&(0,\;0,\;|f_r-h_{r-1}|).
                \end{array}
            \end{align}
            At most one entry is nonzero per tetrahedron, so \(\widetilde{E}'\) is successfully projected back to being easy. 

            \item [(ii)] \(f_{r-1}\neq 0\) and \(h_r\neq 0\), all others zero (tetrahedra \(r\) and \(r+1\) temporarily violate easiness). \\
            The linear combination 
            \begin{equation}
                \widetilde{E}' := \widetilde{E}-\textrm{min}(f_{r-1},\;h_{r})(T_r + T_{r+1} - D_c) 
            \end{equation}
            is a valid internal edge with quad coordinates:
            \begin{align}
                \begin{array}{ll}
                    \textrm{tetrahedron }r-1:& (\textrm{min}(f_{r-1},\;h_{r}),\;0,\;0),\\
                    \textrm{tetrahedron }r:&(0,\;|f_{r-1}-h_r|,\;0),\\
                    \textrm{tetrahedron }r+1:&(0,\;0,\;|f_{r-1}-h_{r}|).
                \end{array}
            \end{align}
            At most one entry is nonzero per tetrahedron, making \(\widetilde{E}'\) easy. 

            \item [(iii)] \(g_{r-1}\neq 0\) and \(g_r\neq 0\), all others zero (tetrahedra \(r-1\) and \(r\) temporarily violate easiness). \\
            The linear combination 
            \begin{equation}
                \widetilde{E}' = \widetilde{E} - \textrm{min}(g_{r-1},\;g_{r})(T_{r-1}+T_r-D_c)\label{eq: modification of easy edge last}
            \end{equation}
            is a valid internal edge with quad coordinates:
            \begin{align}
                \begin{array}{ll}
                    \textrm{tetrahedron }r-1:& (0,\;0,\;|g_r - g_{r-1}|),\\
                    \textrm{tetrahedron }r:&(0,\;|g_r - g_{r-1}|,\;0),\\
                    \textrm{tetrahedron }r+1:&(\textrm{min}(g_{r-1},\;g_{r}),\;0,\;0).
                \end{array}
            \end{align}
            At most one entry is nonzero per tetrahedron, making \(\widetilde{E}'\) easy. 
        \end{enumerate}
        In every case, the operation \(\widetilde{E}\mapsto\widetilde{E}'\) provides a valid easy internal edge in the new triangulation. 

        Combining Steps 1 and 2, the subspace of easy internal edges in \(\widetilde{\CT}\) contains at least \(\sharp_E(\CT)\) independent modified edges along with the newly independent central edge \(D_c\). Thus, \(\sharp_E(\widetilde{\CT})\geq \sharp_E(\CT)+1\). As discussed, this strictly bounds the dimension of the hard edge space: \(\sharp_H(\widetilde{\CT}) \leq \sharp_H(\CT)\). Consequently, the hard-edge symmetry algebra of the new triangulation naturally embeds into that of the old triangulation, establishing the map \(\iota\).
    \end{proof}
    Correspondingly, the embedding $\iota$ of the Cartan subalgebras naturally restricts to an embedding of their associated integer weight lattices:
    \begin{equation}
        \iota_{\Gamma}: \Gamma\left(\mathfrak{F}^{\rm hard}[N;\widetilde{\CT}]\right) \hookrightarrow \Gamma\left(\mathfrak{F}^{\rm hard}[N;\CT]\right) \;.
    \end{equation}
    Let \(\mathbf{v}\) denote the refinement vector living in this lattice. The map \(\iota_{\Gamma}\) embeds these refinement parameters into the parameter space of the old triangulation by padding the lost directions with zeros. Physically, evaluating the refined index of the original triangulation on the image \(\iota_{\Gamma}(\widetilde{\mathbf{v}})\) evaluates it at a configuration where the components corresponding to the lost hard edges are identically zero. This successfully mimics 'turning off' the refinements for any hard edges in \(\CT\) that relaxes into an easy edge in \(\widetilde{\CT}\).

    \paragraph{Invariance of the refined index.} We now use Lemma \ref{lem:easy-23} and \ref{lem:symmetrized pentagon} which will be provided right below to prove that \(\CI_{N; {\rm ref}}\) is invariant under the \(2\to 3\) move. 

    \begin{lemma}[Pentagon identity of symmetrized tetrahedron index]\label{lem:symmetrized pentagon}
        \begin{equation}
            \begin{split}
                &\sum_{k_c\in\mathbb{Z}}(-q^{1/2})^{2k_c}J_{\Delta}(k_c, f_r+g_{r-1},g_r+h_{r-1})J_{\Delta}(k_c, f_{r-1}+g_r, g_{r-1}+h_r)\\
                &\quad \quad J_{\Delta}(k_c, h_{r-1}+h_r, f_{r-1}+f_r) = J_\Delta(f_{r-1}, g_{r-1}, h_{r-1})J_{\Delta}(f_r , g_r, h_r).
            \end{split}
        \end{equation}
    \end{lemma}
    \begin{proof}
        Recall the pentagon identity of ordinary tetrahedron index: 
        \begin{equation} \label{Tetra pentagon}
            \begin{split}
                &\sum_{e_0\in\mathbb{Z}}q^{e_0}\CI_\Delta(m_1, x_1 +e_0)\CI_\Delta(m_2, x_2+e_0)\CI_\Delta(m_1 + m_2, x_3 + e_0)\\
                &=q^{-x_3}\CI_\Delta(m_1 - x_2 + x_2, x_1 - x_3) \CI_\Delta(m_2 - x_1 + x_3, x_2 - x_3). 
            \end{split}        
        \end{equation}
        Setting \(m_1 = f_r- h_{r-1}+g_{r-1}-g_r\), \(m_2 = f_{r-1}- h_r+g_r-g_{r-1}\), \(e_0 = k_c\), \(x_1 = -f_r - g_{r-1}\), \(x_2 = -f_{r-1}-g_r\), and \(x_3 = -f_{r-1}-f_r\) gives 
        \begin{equation}
            \begin{split}
                &J_{\Delta}(k_c, f_r+g_{r-1},g_r+h_{r-1})=(-q^{1/2})^{x_1}\CI_\Delta(m_1 , x_1 + e_0)\\
                &J_{\Delta}(k_c, f_{r-1}+g_r, g_{r-1}+h_r)=(-q^{1/2})^{x_2}\CI_\Delta(m_2 , x_2 + e_0)\\
                &J_{\Delta}(k_c, h_{r-1}+h_r, f_{r-1}+f_r) =(-q^{1/2})^{x_3}\CI_\Delta(m_1+m_2 , x_3 + e_0).
            \end{split}
        \end{equation}
        Therefore, 
        \begin{equation}
            \begin{split}
                (\textrm{LHS)}&=(-q^{1/2})^{x_1 +x_2 + x_3}\sum_{k_c\in\mathbb{Z}}q^{e_0}\CI_\Delta(m_1, x_1 +e_0)\CI_\Delta(m_2, x_2+e_0)\\
                &\quad\;\;\CI_\Delta(m_1 + m_2, x_3 + e_0)\\
                &= (-q^{1/2})^{x_1 +x_2 - x_3}\CI_\Delta(m_1 - x_2 + x_2, x_1 - x_3) \CI_\Delta(m_2 - x_1 + x_3, x_2 - x_3)\\
                &= (-q^{1/2})^{-g_{r-1}-g_r}\CI_\Delta(g_{r-1}-h_{r-1},f_{r-1}-g_{r-1})\CI_\Delta(g_r - h_r, f_r - g_r)\\
                &= J_\Delta(f_{r-1}, g_{r-1}, h_{r-1}) J_\Delta(f_r, g_r, h_r)\\
                &=(\textrm{RHS)}.
            \end{split}
        \end{equation}
    \end{proof}

    \begin{theorem}[Functoriality Refined 3D index under \(2\to 3\) moves] 
    Let \(\mathbf{v} \in \Gamma\left(\mathfrak{F}^{\rm hard}[N;\widetilde{\CT}]\right)\) be the refinement vector for the new triangulation. The refined 3D index \(\CJ_{N; {\rm ref}}\) is invariant under the Pachner \(2\to 3\) move up to the lattice embedding \(\iota_{\Gamma}\) induced by the move: 
    \begin{equation}
        \widetilde{\CJ}_{N;{\rm ref}}(b; \mathbf{v}) = \CJ_{N;{\rm ref}}(b; \iota_{\Gamma}(\mathbf{v})) \;.
    \end{equation}
    When no hard edges are lost during the move (i.e., \(\iota_{\Gamma}\) is an isomorphism), the refined index is strictly invariant.
    \end{theorem}
    \begin{proof}
        Under a \(2\to 3\) move, we choose a basis of internal edges for \(\widetilde{\CT}\) as \(\widetilde{\mathbf{D}}=(\widetilde{\mathbf{H}},\widetilde{\mathbf{E}},E_c)\). By Lemma \ref{lem:easy-23}, the new hard edges \(\widetilde{\mathbf{H}}\) are formed by a subset (or linear combinations) of the original hard edges \(\mathbf{H}\). The new easy edges \(\widetilde{\mathbf{E}}\) consist of the images of the old easy edges (modified according to \eqref{eq: modification of easy edge first}-\eqref{eq: modification of easy edge last} to remain valid) alongside the remaining elements of \(\mathbf{H}\) that relax into easy edges. 
        
        Recall that the single normal surface index \(J(S)\) is independent of tetrahedral solutions \(T_a\), and that \eqref{eq: modification of easy edge first}-\eqref{eq: modification of easy edge last} modify ordinary easy edges merely by adding multiples of \(T_a\) and the central edge \(E_c\). Since we sum over all integers \(k_c\), we can shift \(k_c\) to absorb these modifications:
        \begin{equation}
        \begin{split}
            \widetilde{J}(\mathbf{A}\cdot\widetilde{\mathbf{M}}&+\mathbf{B}\cdot\widetilde{\mathbf{L}}+(\mathbf{k},k_c)\cdot\widetilde{\mathbf{D}}) = (-q^{1/2})^{-\widetilde{\chi}(\mathbf{A}\cdot\widetilde{\mathbf{M}}+\mathbf{B}\cdot\widetilde{\mathbf{L}}+(\mathbf{k},k_c)\cdot\widetilde{\mathbf{D}})}\prod_{a=1}^{r+1}J_{\Delta}(\widetilde{f}_a,\widetilde{g}_a,\widetilde{h}_a)\\
            &=(-q^{1/2})^{2k_c}J_{\Delta}(k_c, f_r+g_{r-1},g_r+h_{r-1})J_{\Delta}(k_c, f_{r-1}+g_r, g_{r-1}+h_r)\\
            &\quad \times J_{\Delta}(k_c, h_{r-1}+h_r, f_{r-1}+f_r)\prod_{a=1}^{r-2}J_{\Delta}(f_a,g_a,h_a). 
        \end{split}
        \end{equation}
        Applying the symmetrized pentagon identity (Lemma \ref{lem:symmetrized pentagon}) evaluates the sum over the central edge variable \(k_c\):
        \begin{equation}
            \sum_{k_c\in\mathbb{Z}}\widetilde{J}(\mathbf{A}\cdot\widetilde{\mathbf{M}}+\mathbf{B}\cdot\widetilde{\mathbf{L}}+(\mathbf{k},k_c)\cdot\widetilde{\mathbf{D}}) = \prod_{a=1}^{r}J_{\Delta }(f_a, g_a, h_a) \;.
        \end{equation}
        Thus, the unrefined geometric weights exactly reproduce those of the original triangulation. Furthermore, because the newly relaxed easy edges and the central edge \(E_c\) do not couple to the new refinement vector \(\mathbf{v}\), the refinement weighting evaluates trivially for these directions. This exactly mimics evaluating the original index at the embedded vector \(\iota_{\Gamma}(\mathbf{v})\), where the components corresponding to the lost hard edges are identically zero. Therefore, 
        \begin{equation}
            \widetilde{\CJ}_{N;\textrm{ref}}(b; \mathbf{v}) = \CJ_{N;\textrm{ref}}(b; \iota_{\Gamma}(\mathbf{v})) \;.
        \end{equation}
    \end{proof}

\section{Examples} \label{sec : examples}
In this section, we present explicit computations of the refined 3D index, defined in \eqref{refined index in boundary 1-cycles} and \eqref{refined index in charge basis}, for a variety of examples and verify the conjecture \eqref{main conjecture}.
\subsection{$M=S^3\backslash\mathbf{5}_2$}
As a representative example, consider the following two Dehn surgery representations of the knot complement $M=S^3\backslash\mathbf{5}_2$:
\begin{align}
\begin{split}
    \CD_0\;&:\;N=S^3\backslash\mathbf{5}_2\;,\quad \{P_I\g_I+Q_I\d_I\}^{d=0}_{I=1}=\emptyset\;,
    \\ \CD\;&:\;N=S^3\backslash\mathbf{5}^2_1\;,\quad \{P_I\g_I+Q_I\d_I\}^{d=1}_{I=1}=\{\mu_2-2\lambda_2\}\;.
\end{split}
\end{align}
Here, $\mathbf{5}^2_1$ denotes the Whitehead link, and $\{(\mu_i,\lambda_i)\}^{2}_{i=1}$ are meridian-longitude pairs associated with the two boundary components. We use the ideal triangulations $\CT_0$ and $\CT$ of $S^3\backslash\mathbf{5}_2$ and $S^3\backslash\mathbf{5}^2_1$, respectively, whose gluing equations are
\begin{align} \label{Gluing equations : 5_2 and 5^2_1}
    \begin{split}
        S^3\backslash\mathbf{5}_2\;:\;&C_1= Z_1'+Z_1''+2Z_2+Z_3'+Z_3''\;, \quad C_2 = Z_1+Z_1'+Z_2'+Z_3+Z_3'\;,\\
        &C_3=Z_1+Z_1''+Z_2'+2Z_2''+Z_3+Z_3''\;,\\
        &hol(\mu) = -Z_1+Z_2\;, \quad hol(\lambda)= 2 (Z_1-Z_1''-Z_2+Z_3'')\;,
        \\ S^3\backslash\mathbf{5}^2_1\;:\;&C_1=Z_1+Z_2+Z_3+Z_4\;,\quad C_2=2Z_1'+Z_1''+2Z_2'+Z_2''+Z_3''+Z_4''\;,\\
        &C_3=Z_1''+Z_2''+2Z_3'+Z_3''+2Z_4'+Z_4''\;,\quad C_4=C_1\;,
        \\&hol(\mu_1)=Z_1'-Z_2+Z_3''\;,\quad hol(\lambda_1)=2(Z_2'+Z_2''-Z_3)\;,\\
        &hol(\mu_2)=Z_1-Z_2'-Z_3''\;,\quad hol(\lambda_2)=2(Z_1-Z_3'-Z_3'')\;.
    \end{split}
\end{align}
For clarity, we rewrite the gluing equations for $S^3\backslash\mathbf{5}_2$ given in \eqref{gluing eqs for 5_2}. We first compute the refined index $\CI_{M;\rm ref}[\mathbf{v};(\mathcal{D}_0,\mathcal{T}_0)]$. 
There are two hard edges, $\{H_1,H_2\}=\{C_1,C_2\}$, in the ideal triangulation $\mathcal{T}_0$. 
Thus, the associated vector space $\mathfrak{F}^{\rm ref}[M;(\mathcal{D}_0,\mathcal{T}_0)]$ in \eqref{F-ref[M]} is two-dimensional and given by
\begin{align}
    \mathfrak{F}^{\rm ref}[M;(\mathcal{D}_0,\mathcal{T}_0)]
    =
    \mathfrak{F}^{\rm hard}[M;(\mathcal{D}_0,\mathcal{T}_0)]
    =
    \mathrm{Span}\{T_{h_1},T_{h_2}\}\;.
\end{align}
We choose 
$\{\alpha_1,\beta_1\}=\{\mu,\lambda\}$, i.e.\ the meridian and longitude. 
Then, using the NZ matrices \eqref{NZ matrices 5_2}, we obtain the refined 3D index, as given in 
\eqref{refined index in boundary 1-cycles}, 
\eqref{refined index in charge basis}, and 
\eqref{refined index for N in charge basis}:
\begin{align}
    \begin{split}
        &\CI^{A\mu+B\lambda}_{M;\rm ref}[\mathbf{v}=W_1T_{h_1}+W_2T_{h_2};(\CD_0,\CT_0)]
        \\&=\sum_{n_1,n_2\in\mathbb{Z}/2}(-q^{1/2})^{-n_2}\CI_\Delta(n_2-m,-e-n_1-m)
        \\&\quad\times\CI_\Delta(n_2,e+2n_1-n_2+m)\CI_\Delta(n_2+m,-n_1)\eta^{2W_1n_1+2W_2n_2}\biggr|_{e=A,m=-2B}\;.
    \end{split} \label{I-ref for 5_2 in D0}
\end{align}
Let us consider the case $(A,B)=(0,0)$:
\begin{align}
\begin{split}
    &\CI^{\mathbf{0}}_{M;\rm ref}[\mathbf{v}=W_1T_{h_1}+W_2T_{h_2};(\CD_0,\CT_0)]
    \\&=1+(-3-\eta^{-2W_2})q+(-3+\eta^{2W_1}-\eta^{-2W_2}+\eta^{2W_2}+\eta^{2W_1+2W_2})q^2+(5+\eta^{-2W_1}
    \\&\qquad+2\eta^{2W_1}+\eta^{-4W_2}+\eta^{-2W_1-4W_2}+2\eta^{-2W_2}+\eta^{-2W_1-2W_2}+2\eta^{2W_2}+\eta^{2W_1+2W_2})q^3+\dots \nonumber
\end{split}
\end{align}
From the index expression, we can confirm that $\mathfrak{F}^{\rm ref}_{\rm dec}[M;(\CD_0,\CT_0)]=\{\mathbf{0}\}$ and thus
\begin{align}
    \mathfrak{F}^{\rm ref}_{\rm IR}[M;(\CD_0,\CT_0)]=\mathfrak{F}^{\rm ref}\;,\quad \dim \mathfrak{F}^{\rm ref}_{\rm IR}=2\;.
\end{align}
We expect $r_{\rm ref}[S^3\backslash\mathbf{5}_2]=2$, so the maximally refined index for $M=S^3\backslash \mathbf{5}_2$ is $\CI_{M;\rm ref}[\mathbf{v}=W_1T_{h_1}+W_2T_{h_2};(\CD_0,\CT_0)]$ .

We now compare this with the refined index 
$\CI_{M;\rm ref}[\mathbf{v};(\mathcal{D},\mathcal{T},\vec{\gamma})]$ 
obtained from a different choice of auxiliary data.  In the ideal triangulation $\mathcal{T}$, there is one easy edge, $E_1=C_1$, 
and one hard edge, $H_1=C_2$. Thus, we have
\begin{align}
    \mathfrak{F}^{\rm hard}[M;(\mathcal{D},\mathcal{T})]
    =
    \mathrm{Span}\{T_{h_1}\}\;.
\end{align}
We choose the bases 
$\{\alpha_1,\beta_1\}=\{\mu_1,\lambda_1\}$ and 
$\{\gamma_1,\delta_1\}=\{\mu_2,\lambda_2\}$, so that 
$\vec{\gamma}=(\mu_2)$, where $\mu_2$ is a non-closable boundary 1-cycle. 
The corresponding refined NZ data are
\begin{align}
    g_{\rm NZ}=\left(\begin{array}{cccccccc}
    -1 & -1 & 0 & 0 & -1 & 0 & 1 & 0 \\
    1 & 1 & 0 & 0 & 0 & 1 & -1 & 0 \\
    -2 & -2 & 0 & 0 & -1 & -1 & 1 & 1 \\
    1 & 1 & 1 & 1 & 0 & 0 & 0 & 0 \\
    0 & -1 & -1 & 0 & 0 & 0 & 0 & 0 \\
    1 & 0 & 1 & 0 & 0 & 0 & 0 & 0 \\
    1 & 1 & 1 & 0 & 0 & 0 & 0 & 0 \\
    0 & 0 & 0 & 0 & 0 & 0 & 0 & 1 \\
    \end{array}\right),\quad 
    \boldsymbol{\nu}_x=\left(\begin{array}{c}
        1\\-1\\2\\-2
    \end{array}\right),\quad 
    \boldsymbol{\nu}_p=\left(\begin{array}{c}
        1\\-1\\0\\0
    \end{array}\right)\;.
\end{align}
Using \eqref{refined index for N in charge basis}, we can compute 
$\CI_{N=S^3\backslash\mathbf{5}_2;\rm ref}[\mathbf{v};(\mathcal{D},\mathcal{T})]$. 
Applying the criterion in \eqref{Dehn filling compatibility}, one finds that any nonzero element in 
$\mathfrak{F}^{\rm hard}$ is not compatible with the Dehn filling:
\begin{align}
    \mathfrak{F}^{\rm hard}_{\rm d}[M;(\mathcal{D},\mathcal{T})]
    =
    \{\mathbf{0}\}\;.
\end{align}
On the other hand, there is one refinement arising from the Dehn filling, since our choice 
$\{\gamma_1,\delta_1\}=\{\mu_2,\lambda_2\}$ leads to $(P_1,Q_1)=(-1,2)\neq(\,\cdot\,,\pm1)$. 
Therefore, the vector space $\mathfrak{F}^{\rm ref}[M;(\mathcal{D},\mathcal{T},\vec{\gamma})]$ is
\begin{align}
    \mathfrak{F}^{\rm ref}[M;(\mathcal{D},\mathcal{T},\vec{\gamma})]
    =
    \mathfrak{F}^{\rm dehn}
    \oplus
    \mathfrak{F}^{\rm hard}_{\rm d}
    =
    \mathrm{Span}\{T_a\}\;.
\end{align}
Using \eqref{refined index in boundary 1-cycles}, 
\eqref{refined index in charge basis}, and 
\eqref{refined index for N in charge basis}, 
the refined 3D index is given by
\begin{align}
    \begin{split}
        &\CI^{A \mu_1 + B \lambda_1 }_{M;{\rm ref}}[\mathbf{v} =V T_{a};\left(\CD, \CT, \vec{\gamma}\right)]
        \\&=\sum_{m_2\in\mathbb{Z},e_2\in\mathbb{Z}/2}\CK_{\rm ref}(-1,2;m_2,e_2;\eta^V)\CI_{N=S^3\backslash\mathbf{5}^2_1}(m_1,m_2,e_1,e_2)\bigr|_{e_1=A,m_1=-2B}\;,
 \\
 &\textrm{where}
 \\
& \CI_{N= S^3\backslash \mathbf{5}^2_1}  (m_1, m_2, e_1, e_2) 
\\
&= \sum_{n_1,n_2 \in \mathbb{Z}/2} (-q^{1/2})^{-e_1-2n_1+e_2+2n_2+m_1-m_2} \CI_\Delta (-n_1, n_2) \CI_\Delta (-e_1-n_1+e_2,n_2+m_1-m_2)
\\
& \quad \times \CI_\Delta(n_1-e_2,-e_1-2n_1+e_2+n_2+m_1)\CI_\Delta(e_1+n_1,-e_1-2n_1+e_2+n_2-m_2)\;. \label{I-ref for 5_2 in D}  \end{split}
\end{align}
Here $\CK_{\rm ref}$ is refined Dehn filling kernel given in \eqref{refined Dehn filling kernel}. 
Or equivalently,  using \eqref{refined Dehn filling Kernel in fugacity basis}, the refined 3D index can be given as (Note that $\frac{P_1}{Q_1} = 0-\frac{1}{2}$)
\begin{align}
\begin{split}
&\CI^{A \mu_1 + B \lambda_1 }_{M;{\rm ref}}[\mathbf{v} =V T_{a};\left(\CD, \CT, \vec{\gamma}\right)]
\\
&= \sum_{m_2,m_3 \in \mathbb{Z}/2} \oint \frac{du_2}{2\pi \mathbf{i} u_2 } \oint \frac{du_3}{2\pi \mathbf{i} u_3 } \Delta(m_2, u_2) \Delta(m_3, u_3)  
\\
& \quad \times \CI^{\rm fug}_{N= S^3\backslash \mathbf{5}^2_1} (m,2 m_2, e, u_2^2) \CI_{T[SU(2)]} (m_2 , u_2, m_3, u_3;\eta^{V})u_3^{4m_3}\big{|}_{e= A, m=-2 B}
\\
&\textrm{where }
\\
&\CI^{\rm fug}_{N= S^3\backslash \mathbf{5}^2_1} (m_1,m_2, e_1, u_2)  = \sum_{e_2\in \mathbb{Z}/2}  \CI_{N= S^3\backslash \mathbf{5}^2_1}  (m_1, m_2, e_1, e_2)u_2^{e_2}\;.
\end{split}
\end{align}
Using the expressions in \eqref{I-ref for 5_2 in D0} and \eqref{I-ref for 5_2 in D}, we can check
\begin{align}
\begin{split}
&  \CI^{A \mu + B \lambda }_{M;{\rm ref}}[\mathbf{v} =-W_2 T_{h_2};(\CD_0, \CT_0)] = \CI^{A \mu_1 - B \lambda_1 }_{M;{\rm ref}}[\mathbf{v} =W_2 T_{a};\left(\CD, \CT, \vec{\gamma}\right)] \;.
\end{split}
\end{align}
The map between boundary 1-cycles reflects a topological correspondence between boundary cycles: 
the cycle $(\mu,\lambda)$ of $S^3\backslash \mathbf{5}_2$ is identified with 
$(\mu_1,-\lambda_1)$ in the Dehn-filled manifold 
$(S^3\backslash \mathbf{5}^2_1)_{[\mu_2-2\lambda_2]}$. 
Using \texttt{SnapPy}, one can explicitly verify, as nontrivial evidence, that
\[
\mathrm{vol}\!\left(
(S^3\backslash \mathbf{5}^2_1)_{[A\mu_1-B\lambda_1],[\mu_2-2\lambda_2]}
\right)
=
\mathrm{vol}\!\left(
(S^3\backslash \mathbf{5}_2)_{[A\mu+B\lambda]}
\right).
\]
We list some comparisons of the refined 3D index:
\begin{align}
\begin{split}
&\CI^{0 \mu + 0 \lambda }_{M;{\rm ref}}[\mathbf{v} = -T_{h_2};(\CD_0, \CT_0)]  =  \CI^{0 \mu_1 + 0 \lambda_1 }_{M;{\rm ref}}[\mathbf{v} = T_{a};\left(\CD, \CT, \vec{\gamma} )\right)] 
\\
& = 1- (3+ \eta^2)q+ (-2+\frac{2}{\eta^2} - \eta^2) q^2 + (8 + \frac{3}{\eta^2}+3 \eta^2+2 \eta^4) q^3+(14-\frac{1}{\eta^2}+10\eta^2+3\eta^4)q^4
\\&\qquad +(15-\frac{12}{\eta^2}+15\eta^2+5\eta^4)q^5+(-15+\frac{3}{\eta^4}-\frac{27}{\eta^2}+7\eta^2-2\eta^6)q^6+\dots
\\
&\CI^{ \mu + 0 \lambda }_{M;{\rm ref}}[\mathbf{v} =- T_{h_2};(\CD_0, \CT_0)]  =  \CI^{ \mu_1 + 0 \lambda_1 }_{M;{\rm ref}}[\mathbf{v} = T_{a};\left(\CD, \CT, \vec{\gamma} \right)] 
\\
& = -(2+ \eta^2) q+ (\frac{1}{\eta^2}-\eta^2)q^2+(8+ \frac{2}{\eta^2}+3 \eta^2+ \eta^4 )q^3+(13-\frac{2}{\eta^2}+9\eta^2+2\eta^4)q^4
\\&\qquad +(10-\frac{11}{\eta^2}+14\eta^2+3\eta^4)q^5+(-18+\frac{2}{\eta^4}-\frac{23}{\eta^2}+6\eta^2-\eta^4-2\eta^6)q^6+\dots
\\
&\CI^{ \mu +  \lambda }_{M;{\rm ref}}[\mathbf{v} =- T_{h_2};(\CD_0, \CT_0)]  =  \CI^{ \mu_1 - \lambda_1 }_{M;{\rm ref}}[\mathbf{v} = T_{a};\left(\CD, \CT, \vec{\gamma}\right)] 
\\
& = (\eta^2+\eta^4) q^2+ (\eta^2+\eta^4) q^3+(-2-4\eta^2-4\eta^4-2\eta^6)q^5+(-4-11\eta^2-9\eta^4-4\eta^6)q^6\quad\quad
\\&\qquad+(-5-19\eta^2-16\eta^4-6\eta^6)q^7+(-1+\frac{1}{\eta^2}-22\eta^2-20\eta^4-6\eta^6)q^8+\dots\;.
    \end{split}
\end{align}
We can also verify from the index expression that
$\mathfrak{F}^{\rm ref}_{\rm dec}[M;(\mathcal{D},\mathcal{T},\vec{\gamma})]=\{\mathbf{0}\}$, and hence
\begin{align}
    \mathfrak{F}^{\rm ref}_{\rm IR}[M;(\mathcal{D},\mathcal{T},\vec{\gamma})]
    =
    \mathfrak{F}^{\rm ref}\;,
    \qquad
    \dim \mathfrak{F}^{\rm ref}_{\rm IR}=1\;.
\end{align}
The injective linear map in \eqref{main conjecture} can be given as:
\begin{align}
    F\;:\;
    \mathfrak{F}^{\rm ref}_{\rm IR}[M;(\mathcal{D},\mathcal{T},\vec{\gamma})]
    \;\hookrightarrow\;
    \mathfrak{F}^{\rm ref}_{\rm IR}[M;(\mathcal{D}_0,\mathcal{T}_0)]\;,
    \qquad
    W_2T_a \;\mapsto\; -W_2T_{h_2}\;. 
\end{align}
Note that
\(
\dim \mathfrak{F}^{\rm ref}_{\rm IR}[M;(\mathcal{D},\mathcal{T},\vec{\gamma})]
<
\dim \mathfrak{F}^{\rm ref}_{\rm IR}[M;(\mathcal{D}_0,\mathcal{T}_0)]
=2
\),
and hence the map is injective but not surjective. Note that, even without comparison with 
$\CI_{M;\rm ref}[\mathbf{v};(\mathcal{D}_0,\mathcal{T}_0)]$, 
we can check that the necessary condition for maximal refinement 
\eqref{maximal refinement necessary} is violated:
\begin{align}
    {\rm coeff}_{q^1}\!\left[
    \CI^{\mathbf{0}}_{M;\rm ref}
    [\mathbf{v};(\mathcal{D},\mathcal{T},\vec{\gamma})]
    \big|_{\rm neutral}
    \right]
    =
    -3
    <
    \underbrace{
    -\dim \mathfrak{F}_{\rm IR}[M;(\mathcal{D},\mathcal{T},\vec{\gamma})]-1
    }_{=-2}\;.
\end{align}
This indicates that additional refinements must be present. Indeed, we can identify at least one more from 
$\CI_{M;\rm ref}[\mathbf{v};(\mathcal{D}_0,\mathcal{T}_0)]$.

The analysis of the remaining examples follows the same procedure as above. 
We therefore present them focusing on the essential data and the verification of the conjecture.
\subsection{$M=m003$} \label{Example 2}
We consider the following three Dehn surgery representations of the $M=m003$:
\begin{align}
\begin{split}
    \CD_0\;&:\;N=m003\;,\quad \{P_I\g_I+Q_I\d_I\}^{d=0}_{I=1}=\emptyset\;,
    \\ \CD\;&:\;N=m125\;,\quad \{P_I\g_I+Q_I\d_I\}^{d=1}_{I=1}=\{\a_{\rm sp2}+2\b_{\rm sp2}\}\;,
    \\ \CD'\;&:\;N=m125\;,\quad \{P_I\g_I+Q_I\d_I\}^{d=1}_{I=1}=\{2\a_{\rm sp2}-\b_{\rm sp2}\}\;.
\end{split}
\end{align}
Here $(\alpha_{\rm sp},\beta_{\rm sp})$ is a basis of $H_1(\partial(m003);\mathbb{Z})$ and $\{(\alpha_{\textrm{sp}i},\beta_{\textrm{sp}i})\}_{i=1}^2$ are basis associated with the two boundary components of $m125$, as chosen by \texttt{SnapPy}. We use the ideal triangulations $\CT_0$ and $\CT$ of $m003$ and $m125$, respectively, as provided by \texttt{SnapPy}, whose gluing equations are
\begin{align}
\label{Gluing equation : m003, m125}
\begin{split}
m003 \;:\; & C_1 = 2Z_1+Z_1''+2Z_2+Z_2''\;, \quad C_2 = 2Z_1'+Z_1''+2Z_2'+Z_2''\;,
\\
& hol(\alpha_{\rm sp}) = -2Z_1'+2Z_2\;, \quad hol(\beta_{\rm sp}) = -Z_1'+2Z_2-Z_2'\;,
\\\text{m125} \;:\; &C_1 = Z_1+Z_1'+Z_2+Z_2''+Z_3+Z_3''+Z_4+Z_4''\;, 
\\&C_2 = Z_1'+Z_2'+Z_2''+Z_3+Z_4'+Z_4''\;,\quad C_3=Z_1''+Z_2'+Z_3'+Z_4'\;,
\\&C_4=Z_1+Z_1''+Z_2+Z_3'+Z_3''+Z_4\;,
\\&hol(\alpha_{\mathrm{sp}1}) = -Z_1+Z_1'-Z_2-Z_4+Z_4'\;, \quad hol(\beta_{\mathrm{sp}1}) = -Z_2+Z_3'-Z_4''\;,
\\&hol(\alpha_{\mathrm{sp}2})=-Z_1'-Z_3''+Z_4'\;,\quad hol(\beta_{\mathrm{sp}2})=Z_1''-Z_2-Z_4''\;.
\end{split}
\end{align}

\paragraph{The refined index $\CI_{M;\rm ref}[\mathbf{v};(\CD_0,\CT_0)]$} From the following choices
\begin{align}
    H_1=C_1\;,\quad \{\a_1,\b_1\}=\{\a_{\rm sp},\b_{\rm sp}\}\;,
\end{align}
we get the refined 3D index with the one dimensional vector space
\begin{align}
    \CI^{A\alpha_{\rm sp}+B\beta_{\rm sp}}_{M;\textrm{ref}}[\mathbf{v}=W_1T_{h_1};(\CD_0,\CT_0)]\;,\quad \mathfrak{F}^{\rm ref}[M;(\CD_0,\CT_0)]=\mathrm{Span}\{T_{h_1}\}\;.
\end{align}
Let's look
\begin{align}
\begin{split}
    &\CI^{\mathbf{0}}_{M;\rm ref}[\mathbf{v}=T_{h_1};(\CD_0,\CT_0)]
    \\&=1-2q-3q^2+(\frac{1}{\eta^2}+\eta^{2})q^3+(4+\frac{2}{\eta^2}+2\eta^{2})q^4+(12+\frac{3}{\eta^2}+3\eta^2)q^5+(14+\frac{2}{\eta^2}
    \\&\qquad+2\eta^2)q^6+(16-\frac{1}{\eta^2}-\eta^2)q^7+(4-\frac{8}{\eta^2}-8\eta^2)q^8+(-16-\frac{18}{\eta^2}-18\eta^2)q^9+\dots
\end{split}
\end{align}
From the index expression, we can confirm that $\mathfrak{F}_{\rm dec}^{\rm ref}[M;(\CD_0,\CT_0)]=\{\mathbf{0}\}$ and thus
\begin{align}
    \mathfrak{F}^{\rm ref}_{\rm IR}[M;(\CD_0,\CT_0)]=\mathfrak{F}^{\rm ref}\;,\quad \dim\mathfrak{F}^{\rm ref}_{\rm IR}=1\;.
\end{align}
We expect $r_{\rm ref}[m003]=1$, so the maximally refined index for $M=m003$ is $\CI_{M;\rm ref}[\mathbf{v}=W_1T_{h_1};(\CD_0,\CT_0)]$.

\paragraph{The refined index $\CI_{M;\rm ref}[\mathbf{v};(\CD,\CT,\vec{\gamma})]$} From the following choices
\begin{align} \label{choices for m125(2,1)}
\begin{split}
    &\{H_1,E_1\}=\{C_4,C_3\}\;,\quad \{\a_1,\b_1\}=\{\a_{\rm sp1},\b_{\rm sp1}\}\;,
    \\&\{\gamma_1,\delta_1\}=\{\a_{\rm sp2},\b_{\rm sp2}\}\rightarrow \vec{\gamma}=(\a_{\rm sp 2})\text{ and }(P_1,Q_1)=(1,2)\;,
    \\&\mathbf{a}_{\rm d}=\mathbf{b}_{\rm d}=\mathbf{0}
\end{split}
\end{align}
where $\a_{\rm sp2}$ is non-closable boundary 1-cycle, we get the refined 3D index
\begin{align}
    \begin{split}
        &\CI^{A\a_{\rm sp1}+B\b_{\rm sp1}}_{M;\rm ref}[\mathbf{v}=W_1T_{h_1}+VT_a;(\CD,\CT,\vec{\gamma})]
        \\&=\sum_{m_2\in\mathbb{Z},e_2\in\mathbb{Z}/2}\CK_{\rm ref}(1,2;m_2,e_2;\eta^V)\CI_{N=m125;\rm ref}(\mathbf{v}=W_1T_{h_1};m_1,m_2,e_1,e_2)\bigr|_{e_1=A,m_1=-2B}
    \end{split}
\end{align}
with the two dimensional vector space
\begin{align}
\begin{split}
    &\mathfrak{F}^{\rm hard}_{\rm d}[M;(\CD,\CT)]=\mathrm{Span}\{T_{h_1}\}\;,\quad \mathfrak{F}^{\rm dehn}[M;(\CD,\vec{\gamma})]=\mathrm{Span}\{T_a\}
    \\&\rightarrow \mathfrak{F}^{\rm ref}[M;(\CD,\CT,\vec{\gamma})]=\mathfrak{F}^{\rm dehn}\oplus \mathfrak{F}^{\rm hard}_{\rm d}=\mathrm{Span}\{T_a,T_{h_1}\}\;.
\end{split}
\end{align}
From the expression of the index with extra background shift\footnote{This background shift corresponds to the choice of $f$ in \eqref{F-ref[M]-IR-2}.} $\eta^{-2AW_1}$, we checked that $\mathfrak{F}^{\rm ref}_{\rm dec}[M;(\CD,\CT,\vec{\gamma})]=\{T_{h_1}+T_a\}$ and thus
\begin{align}
\begin{split}
    &\mathfrak{F}^{\textrm{ref}}_{\textrm{IR}}[M;(\CD,\CT,\vec{\gamma})]=\mathfrak{F}^{\textrm{ref}}/\mathfrak{F}^{\textrm{ref}}_{\textrm{dec}}=\textrm{Span}\{T_{h_1}\}\;,\quad \dim\mathfrak{F}^{\textrm{ref}}_{\textrm{IR}}=1
\end{split}
\end{align}
where we choose the representative $T_{h_1}$ among $T_{h_1}+\alpha(T_{h_1}+T_a)$. Since
\begin{align}
    \dim \mathfrak{F}^{\rm ref}_{\rm IR}[M;(\CD,\CT,\vec{\gamma})]=\dim \mathfrak{F}^{\rm ref}_{\rm IR}[M;(\CD_0,\CT_0)]\;,
\end{align}
the injective linear map in \eqref{main conjecture} will be a bijective for this case. We checked the following identification
\begin{align}
    \eta^{-2(A+B)W_1}\CI_{M;\rm ref}^{B\a_{\rm sp }+A\b_{\rm sp}}[\mathbf{v}=-W_1T_{h_1};(\CD_0,\CT_0)]=\eta^{-2AW_1}\CI_{M;\rm ref}^{A\a_{\rm sp1}+B\b_{\rm sp1}}[\mathbf{v}=W_1T_{h_1};(\CD,\CT,\vec{\gamma})]\;.
\end{align}
It also reflects a topological correspondence between the boundary cycles: the cycle $(\a_{\rm sp},\b_{\rm sp})$ of $m003$ is identified with $(\b_{\rm sp1},\a_{\rm sp1})$ in the Dehn-filled manifold $(m125)_{[\a_{\rm sp2}+2\b_{\rm sp2}]}$. For example,
\begin{align}
\begin{split}
    &\CI^{0\alpha_{\rm sp}+0\beta_{\rm sp}}_{M;\textrm{ref}}[\mathbf{v}=-T_{h_1};(\CD_0,\CT_0)]=\CI^{0\alpha_{\rm sp1}+0\beta_{\rm sp1}}_{M;\textrm{ref}}[\mathbf{v}=T_{h_1};(\CD,\CT,\vec{\gamma})]
    \\&=1-2q-3q^2+(\frac{1}{\eta^2}+\eta^2)q^3+(4+\frac{2}{\eta^2}+2\eta^2)q^4+(12+\frac{3}{\eta^2}+3\eta^2)q^5+(14+\frac{2}{\eta^2}
    \\&\qquad+2\eta^2)q^6+(16-\frac{1}{\eta^2}-\eta^2)q^7+(4-\frac{8}{\eta^2}-8\eta^2)q^8+(-16-\frac{18}{\eta^2}-18\eta^2)q^9+\dots
    \\&\eta^{-2}\CI^{\alpha_{\rm sp}+0\beta_{\rm sp}}_{M;\textrm{ref}}[\mathbf{v}=-T_{h_1};(\CD_0,\CT_0)]=\CI^{0\alpha_{\rm sp1}+\beta_{\rm sp1}}_{M;\textrm{ref}}[\mathbf{v}=T_{h_1};(\CD,\CT,\vec{\gamma})]
    \\&=-q-q^2+(1+\frac{1}{\eta^2}+\eta^2)q^3+(4+\frac{1}{\eta^2}+\eta^2)q^4+(8+\frac{2}{\eta^2}+2\eta^2)q^5+9q^6
    \\&\qquad +(7-\frac{2}{\eta^2}-2\eta^2)q^7+(-3-\frac{8}{\eta^2}-8\eta^2)q^8+(-20-\frac{15}{\eta^2}-15\eta^2)q^9+\dots
    \\&\eta^{-4}\CI^{\alpha_{\rm sp}+\beta_{\rm sp}}_{M;\textrm{ref}}[\mathbf{v}=-T_{h_1};(\CD_0,\CT_0)]=\eta^{-2}\CI^{\alpha_{\rm sp1}+\beta_{\rm sp1}}_{M;\textrm{ref}}[\mathbf{v}=T_{h_1};(\CD,\CT,\vec{\gamma})]
    \\&=\frac{q^2}{\eta^2}+(1+\frac{1}{\eta^2})q^3+(2+\frac{1}{\eta^2})q^4+3q^5+(2-\frac{2}{\eta^2})q^6+(-1-\frac{6}{\eta^2}-\eta^2)q^7+\dots
    \\&\qquad+(-8-\frac{11}{\eta^2}-2\eta^2)q^8+(-18+\frac{1}{\eta^4}-\frac{16}{\eta^2}-4\eta^2)q^9+\dots
\end{split}
\end{align}

\paragraph{The refined index $\CI_{M;\rm ref}[\mathbf{v};(\CD',\CT,\vec{\gamma})]$} Using the choices \eqref{choices for m125(2,1)}, we get the same $\CI_{N=m125;\rm ref}$ but a different vector space
\begin{align}
    \mathfrak{F}^{\rm ref}[M;(\CD',\CT,\vec{\gamma})]=\mathfrak{F}^{\rm hard}_{\rm d}[M;(\CD',\CT)]=\mathrm{Span}\{T_{h_1}\}
\end{align}
since $(P_1,Q_1)=(2,-1)=(\;\cdot\;,\pm 1)$ for this case. The resulting refined 3D index is
\begin{align}
    \begin{split}
        &\CI_{M;\rm ref}^{A\a_{\rm sp1}+B\b_{\rm sp1}}[\mathbf{v}=W_1T_{h_1};(\CD',\CT,\vec{\gamma})]
        \\&=\sum_{m_2\in\mathbb{Z},e_2\in\mathbb{Z}/2}\CK(2,-1;m_2,e_2)\CI_{N=m125;\rm ref}(\mathbf{v}=W_1T_{h_1};m_1,m_2,e_1,e_2)\bigr|_{e_1=A,m_1=-2B}\;.
    \end{split}
\end{align}
From the expression of the index, we checked that $\mathfrak{F}^{\rm ref}_{\rm dec}[M;(\CD',\CT,\vec{\gamma})]=\{\mathbf{0}\}$ and thus
\begin{align}
    \mathfrak{F}_{\rm IR}^{\rm ref}[M;(\CD',\CT,\vec{\gamma})]=\mathfrak{F}^{\rm ref}\;,\quad \dim\mathfrak{F}^{\rm ref}_{\rm IR}=1\;.
\end{align}
Since
\begin{align}
    \dim\mathfrak{F}^{\rm ref}_{\rm IR}[M;(\CD',\CT,\vec{\gamma})]=\dim \mathfrak{F}^{\rm ref}_{\rm IR}[M;(\CD_0,\CT_0)]\;,
\end{align}
the injective linear map in \eqref{main conjecture} will be a bijective for this case. We checked the following identification
\begin{align}
    \eta^{2BW_1}\CI^{A\alpha_{\rm sp}+B\beta_{\rm sp}}_{M;\textrm{ref}}[\mathbf{v}=W_1T_{h_1};(\CD_0,\CT_0)]=\CI^{-A\alpha_{\rm sp1}+B\beta_{\rm sp1}}_{M;\textrm{ref}}[\mathbf{v}=W_1T_{h_1};(\CD',\CT,\vec{\gamma})]\;.
\end{align}
It also reflects a topological correspondence between the boundary cycles: the cycle $(\a_{\rm sp},\b_{\rm sp})$ of $m003$ is identified with $(-\a_{\rm sp1},\b_{\rm sp1})$ in the Dehn-filled manifold $(m125)_{[2\a_{\rm sp2}-\b_{\rm sp2}]}$. For example,
\begin{align}
\begin{split}
    &\CI^{0\alpha_{\rm sp}+0\beta_{\rm sp}}_{M;\textrm{ref}}[\mathbf{v}=T_{h_1};(\CD_0,\CT_0)]=\CI^{0\alpha_{\rm sp1}+0\beta_{\rm sp1}}_{M;\textrm{ref}}[\mathbf{v}=T_{h_1};(\CD',\CT,\vec{\gamma})]
    \\&=1-2q-3q^2+(\frac{1}{\eta^2}+\eta^2)q^3+(4+\frac{2}{\eta^2}+2\eta^2)q^4+(12+\frac{3}{\eta^2}+3\eta^2)q^5+(14+\frac{2}{\eta^2}
    \\&\qquad+2\eta^2)q^6+(16-\frac{1}{\eta^2}-\eta^2)q^7+(4-\frac{8}{\eta^2}-8\eta^2)q^8+(-16-\frac{18}{\eta^2}-18\eta^2)q^9+\dots
    \\&\CI^{\alpha_{\rm sp}+0\beta_{\rm sp}}_{M;\textrm{ref}}[\mathbf{v}=T_{h_1};(\CD_0,\CT_0)]=\CI^{-\alpha_{\rm sp1}+0\beta_{\rm sp1}}_{M;\textrm{ref}}[\mathbf{v}=T_{h_1};(\CD',\CT,\vec{\gamma})]
    \\&=-\frac{q}{\eta^2}-\frac{q^2}{\eta^2}+(1+\frac{1}{\eta^4}+\frac{1}{\eta^2})q^3+(1+\frac{1}{\eta^4}+\frac{4}{\eta^2})q^4+(2+\frac{2}{\eta^4}+\frac{8}{\eta^2})q^5+\frac{9}{\eta^2}q^6
    \\&\qquad +(-2-\frac{2}{\eta^4}+\frac{7}{\eta^2})q^7+(-8-\frac{8}{\eta^4}-\frac{3}{\eta^2})q^8+(-15-\frac{15}{\eta^4}-\frac{20}{\eta^2})q^9+\dots
    \\&\eta^{2}\CI^{\alpha_{\rm sp}+\beta_{\rm sp}}_{M;\textrm{ref}}[\mathbf{v}=T_{h_1};(\CD_0,\CT_0)]=\CI^{-\alpha_{\rm sp1}+\beta_{\rm sp1}}_{M;\textrm{ref}}[\mathbf{v}=T_{h_1};(\CD',\CT,\vec{\gamma})]
    \\&=q^2+(1+\frac{1}{\eta^2})q^3+(1+\frac{2}{\eta^2})q^4+\frac{3}{\eta^2}q^5+(-2+\frac{2}{\eta^2})q^6+(-6-\frac{1}{\eta^4}-\frac{1}{\eta^2})q^7
    \\&\qquad +(-11-\frac{2}{\eta^4}-\frac{8}{\eta^2})q^8+(-16-\frac{4}{\eta^4}-\frac{18}{\eta^2}+\eta^2)q^9+(-20-\frac{6}{\eta^4}-\frac{30}{\eta^2}+2\eta^2)q^{10}+\dots
\end{split}
\end{align}

\subsection{Closed hyperbolic manifold with ${\rm vol}\simeq 0.98137$, $(S^3\backslash \mathbf{5}_2)_{[5 \mu+\lambda]}$} \label{Example 3}
We consider the following three Dehn surgery representations of the closed hyperbolic 3-manifold $M=(S^3\backslash\mathbf{5}_2)_{[5\mu+\lambda]}$:
\begin{align}
\begin{split}
    \CD_0\;&:\;N=S^3\backslash\mathbf{5}_2\;,\quad \{P_I\g_I+Q_I\d_I\}^{d=1}_{I=1}=\{5\mu+\lambda\}\;,
    \\ \CD\;&:\;N=S^3\backslash\mathbf{4}_1\;,\quad \{P_I\g_I+Q_I\d_I\}^{d=1}_{I=1}=\{-5\mu+\lambda\}\;,
    \\ \CD'\;&:\;N=m003\;,\quad \{P_I\g_I+Q_I\d_I\}^{d=1}_{I=1}=\{2\a_{\rm sp}-3\b_{\rm sp}\}\;.
\end{split}
\end{align}
Here $(\mu,\lambda)$ is the meridian-longitude pair associated with boundary components of the two knot complements $S^3\backslash \mathbf{5}_2$ and $S^3\backslash \mathbf{4}_1$. $(\a_{\rm sp},\b_{\rm sp})$ is the basis of $H_1(\partial(m003);\mathbb{Z})$ chosen by \texttt{SnapPy}. See $\eqref{gluing eqs for 5_2}$, \eqref{gluing eqs for 4_1} and $\eqref{Gluing equation : m003, m125}$, for the gluing equations associated with the ideal triangulation $\CT_0$, $\CT$ and $\CT'$, of $\mathbf{5}_2$, $\mathbf{4}_1$ and $m003$, respectively.
\paragraph{The refined index $\CI_{M;\rm ref}[\mathbf{v};(\CD_0,\CT_0,\vec{\gamma}_0)]$} From the following choices
\begin{align}
\begin{split}
    &\{H_1,H_2\}=\{C_1,C_2\}\;,
    \\& \{\gamma_1,\delta_1\}=\{\mu,\lambda\}\rightarrow \vec{\gamma}_0=(\mu)\text{ and }(P_1,Q_1)=(5,1)\;,
    \\&\mathbf{a}_{\rm d}=\mathbf{b}_{\rm d}=\mathbf{0}
\end{split}
\end{align}
where $\mu$ is non-closable boundary 1-cycle, we get the refined 3D index
\begin{align}
    \begin{split}
        &\CI_{M;\rm ref}[\mathbf{v}=W_2T_{h_2};(\CD_0,\CT_0,\vec{\gamma}_0)]
        \\&=\sum_{m\in\mathbb{Z},e\in\mathbb{Z}/2}\CK(5,1;m,e)\CI_{N=S^3\backslash\mathbf{5}_2;\rm ref}(\mathbf{v}=W_2T_{h_2};m,e)
        \\&=1-q-2q^2+(-2+\eta^{-4W_2})q^3+(-2+\eta^{-4W_2})q^4+\eta^{-4W_2}q^5+(1+\eta^{-4W_2})q^6+(5
        \\&\qquad +\eta^{-4W_2}+\eta^{4W_2})q^7+(7+\eta^{4W_2})q^8+(11-\eta^{-4W_2}+2\eta^{4W_2})q^9+\dots
    \end{split}
\end{align}
with the one dimensional vector space
\begin{align}
    \begin{split}
        &\mathfrak{F}^{\rm hard}_{\rm d}[M;(\CD_0,\cT_0)]=\mathrm{Span}\{T_{h_2}\}\;,\quad \mathfrak{F}^{\rm dehn}[M;(\CD_0,\vec{\gamma}_0)]=\{\mathbf{0}\}
        \\&\rightarrow \mathfrak{F}^{\rm ref}[M;(\CD_0,\CT_0,\vec{\gamma}_0)]=\mathfrak{F}^{\rm dehn}\oplus \mathfrak{F}^{\rm hard}_{\rm d}=\mathrm{Span}\{T_{h_2}\}\;.
    \end{split}
\end{align}
From the index expression, we can confirm that $\mathfrak{F}^{\rm ref}_{\rm dec}[M;(\CD_0,\CT_0,\vec{\gamma}_0)]=\{\mathbf{0}\}$ and thus
\begin{align}
    \mathfrak{F}^{\rm ref}_{\rm IR}[M;(\CD_0,\CT_0,\vec{\gamma}_0)]=\mathfrak{F}^{\rm ref}\;,\quad \dim \mathfrak{F}^{\rm ref}_{\rm IR}=1\;.
\end{align}
We expect $r_{\rm ref}[(S^3\backslash\mathbf{5}_2)_{[5\mu+\lambda]}]=1$, so the maximally refined index for $M=(S^3\backslash\mathbf{5}_2)_{[5\mu+\lambda]}$ is $\CI_{M;\rm ref}[\mathbf{v}=W_2T_{h_2};(\CD_0,\CT_0,\vec{\gamma}_0)]$.

\paragraph{The refined index $\CI_{M;\rm ref}[\mathbf{v};(\CD,\CT,\vec{\gamma})]$} From the following choices
\begin{align}
    \begin{split}
        &H_1=C_1\;,
        \\&\{\gamma_1,\delta_1\}=\{\mu,\lambda\}\rightarrow \vec{\gamma}=(\mu)\text{ and }(P_1,Q_1)=(-5,1)
    \end{split}
\end{align}
where $\mu$ is non-closable boundary 1-cycle, the vector space becomes zero dimensional
\begin{align}
    &\mathfrak{F}^{\rm hard}_{\rm d}[M;(\CD,\CT)]=\mathfrak{F}^{\rm dehn}[M;(\CD,\vec{\gamma})]=\{\mathbf{0}\}\rightarrow\mathfrak{F}^{\rm ref}[M;(\CD,\CT,\vec{\gamma})]=\mathfrak{F}^{\rm dehn}\oplus \mathfrak{F}^{\rm hard}_{\rm d}=\{\mathbf{0}\}
\end{align}
so there is no refinement. The resulting (un)refined 3D index is
\begin{align}
\begin{split}
    &\CI_{M;\textrm{ref}}[(\CD,\CT,\vec{\gamma})]=\CI_M
    \\&=\sum_{m\in\mathbb{Z},e\in\mathbb{Z}/2}\CK(-5,1;m,e)\CI_{N=S^3\backslash \mathbf{4}_1}(m,e)
    \\&=1-q-2q^2-q^3-q^4+q^5+2q^6+7q^7+8q^8+12q^9+14q^{10}+16q^{11}+12q^{12}+\dots
\end{split}
\end{align}
$\dim\mathfrak{F}^{\rm ref}_{\rm IR}[M;(\CD,\CT,\vec{\gamma})]<\dim\mathfrak{F}^{\rm ref}_{\rm IR}[M;(\CD_0,\CT_0,\vec{\gamma}_0)]$ and the injective linear map in \eqref{main conjecture} trivially exists as
\begin{align}
F\;:\; \mathfrak{F}^{\rm ref}_{\rm IR}[M;\left(\CD, \CT, \vec{\gamma}\right)]\;\hookrightarrow\;\mathfrak{F}^{\rm ref}_{\rm IR}[M;(\mathcal{D}_0,\mathcal{T}_0,\vec{\gamma}_0)]\;,
\quad
 \mathbf{0}\mapsto \mathbf{0}\;.
\end{align}
Note that the necessary condition for the maximal refinement \eqref{maximal refinement necessary} is violated as
\begin{align}
    {\rm coeff}_{q^1}\left[\mathcal{I}_{M;\rm ref}[(\CD,\CT,\vec{\gamma})]\right]=-1 < \underbrace{-\dim \mathfrak{F}_{\rm IR}[M;(\CD,\CT,\vec{\gamma})]-0}_{=0}\;.
\end{align}
It indicates that there must be extra refinements. We can find at least one more from $\CI_{M;\rm ref}[\mathbf{v};(\CD_0,\CT_0,\vec{\gamma}_0)]$.

\paragraph{The refined index $\CI_{M;\rm ref}[\mathbf{v};(\CD',\CT',\vec{\gamma}')]$} From the following choices
\begin{align}
    \begin{split}
        &H_1=C_1\;,
        \\&\{\gamma_1,\delta_1\}=\{\a_{\rm sp},\beta_{\rm sp}\}\rightarrow \vec{\gamma}'=(\a_{\rm sp})\text{ and }(P_1,Q_1)=(2,-3)\;,
        \\&\mathbf{a}_{\rm d}=1\;,\quad \mathbf{b}_{\rm d}=-\frac{1}{2}
    \end{split}
\end{align}
where $\a_{\rm sp}$ is non-closable boundary 1-cycle, we get the refined 3D index
\begin{align}
\begin{split}
    &\CI_{M;\rm ref}[\mathbf{v}=W_1T_{h_1}+VT_a;(\CD',\CT',\vec{\gamma}')]
    \\&=\sum_{m\in\mathbb{Z},e\in\mathbb{Z}/2}\CK_{\rm ref}(2,-3;m,e;\eta^V)\CI_{N=m003;\rm ref}(\mathbf{v}=W_1T_{h_1};m,e;\mathbf{a}_{\rm d},\mathbf{b}_{\rm d})
    \\&=1-q-2q^2+(-2+\eta^{2W_1+2V})q^3+(-2+\eta^{2W_1+2V})q^4+\eta^{2W_1+2V}q^5+(1+\eta^{2W_1+2V})q^6
    \\&\qquad +(5+\eta^{-2W_1-2V}+\eta^{2W_1+2V})q^7+(7+\eta^{-2W_1-2V})q^8+\dots
\end{split}
\end{align}
with the two dimensional vector space
\begin{align}
\label{m003(2,-3) vector space}
    \begin{split}
        &\mathfrak{F}^{\rm hard}_{\rm d}[M;(\CD',\CT')]=\mathrm{Span}\{T_{h_1}\}\;,\quad \mathfrak{F}^{\rm dehn}[M;(\CD',\vec{\gamma}')]=\mathrm{Span}\{T_a\}
        \\&\rightarrow \mathfrak{F}^{\rm ref}[M;(\CD',\CT',\vec{\gamma}')]=\mathfrak{F}^{\rm dehn}\oplus\mathfrak{F}^{\rm hard}_{\rm d}=\mathrm{Span}\{T_a,T_{h_1}\}\;.
    \end{split}
\end{align}
From the index expression, we can confirm that $\mathfrak{F}^{\rm ref}_{\rm dec}[M;(\CD',\CT',\vec{\gamma}')]=\mathrm{Span}\{T_{h_1}-T_a\}$ and thus
\begin{align}
    \mathfrak{F}^{\rm ref}_{\rm IR}[M;(\CD',\CT',\vec{\gamma}')]=\mathfrak{F}^{\rm ref}/\mathfrak{F}^{\rm ref}_{\rm dec}=\mathrm{Span}\{T_{h_1}\}\;,\quad \dim\mathfrak{F}^{\rm ref}_{\rm IR}=1
\end{align}
where we choose the representative $T_{h_1}$ among $T_{h_1}+\a(T_{h_1}-T_a)$. Since
\begin{align}
    \dim\mathfrak{F}^{\rm ref}_{\rm IR}[M;(\CD',\CT',\vec{\gamma}')]=\dim\mathfrak{F}^{\rm ref}_{\rm IR}[M;(\CD_0,\CT_0,\vec{\gamma}_0)]\;,
\end{align}
the injective linear map in \eqref{main conjecture} will be a bijective for this case. We can check it from the following identification
\begin{align}
    \CI_{M;\textrm{ref}}[\mathbf{v}=W_2T_{h_2};(\CD_0,\CT_0,\vec{\gamma}_0)]=\CI_{M;\textrm{ref}}[\mathbf{v}=-2W_2T_{h_1};(\CD',\CT',\vec{\gamma}')]\;.
\end{align}
Alternatively, we can choose
\begin{align} \label{NC cycles in m003}
\begin{split}
    &\{\gamma_1,\delta_1\}=\{\b_{\rm sp},-\a_{\rm sp}\}\rightarrow \vec{\gamma}''=(\b_{\rm sp})\text{ and }(P_1,Q_1)=(-3,-2)\;,
    \\&\mathbf{a}_{\rm d}=1\;,\quad \mathbf{b}_{\rm d}=\frac{1}{2}
    \\&\text{or}
    \\&\{\gamma_1,\delta_1\}=\{\b_{\rm sp}-\a_{\rm sp},-\a_{\rm sp}\}\rightarrow \vec{\gamma}'''=(\b_{\rm sp}-\a_{\rm sp})\text{ and }(P_1,Q_1)=(-3,1)\;,
    \\&\mathbf{a}_{\rm d}=\mathbf{0}\;,\quad \mathbf{b}_{\rm d}=\frac{1}{2}
\end{split}
\end{align}
from the fact that $\b_{\rm sp}$ and $\b_{\rm sp}-\a_{\rm sp}$ are also non-closable boundary 1-cycles. For $\vec{\gamma}''$, we find the two dimensional vector space $\mathfrak{F}^{\rm ref}$ which is a direct sum of one dimensional $\mathfrak{F}^{\rm dehn}$ and $\mathfrak{F}^{\rm hard}_{\rm d}$, like \eqref{m003(2,-3) vector space}. Then again, we can confirm that the one dimensional $\mathfrak{F}^{\rm ref}_{\rm dec}=\mathrm{Span}\{2T_{h_1}-T_{a}\}$ decoupled, from the index expression
\begin{align}
\begin{split}
    &\CI_{M;\textrm{ref}}[\mathbf{v}=W_1T_{h_1}+VT_a;(\CD',\CT',\vec{\gamma}'')]
    \\&=1-q-2q^2+(-2+\eta^{2W_1+4V})q^3+(-2+\eta^{2W_1+4V})q^4+\eta^{2W_1+4V}q^5+\dots
\end{split}
\end{align}
so
\begin{align}
    \mathfrak{F}^{\rm ref}_{\rm IR}[M;(\CD',\CT',\vec{\gamma}'')]=\mathfrak{F}^{\rm ref}/\mathfrak{F}^{\rm ref}_{\rm dec}=\mathrm{Span}\{T_{h_1}\}\;,\quad \dim\mathfrak{F}^{\rm ref}_{\rm IR}=1
\end{align}
where we choose the representative $T_{h_1}$ among $T_{h_1}+\a(2T_{h_1}-T_a)$. Since
\begin{align}
    \dim\mathfrak{F}^{\rm ref}_{\rm IR}[M;(\CD',\CT',\vec{\gamma}'')]=\dim\mathfrak{F}^{\rm ref}_{\rm IR}[M;(\CD_0,\CT_0,\vec{\gamma}_0)]\;,
\end{align}
the injective linear map in \eqref{main conjecture} will be a bijective for this case. We can check it from the following identification
\begin{align}
    \CI_{M;\textrm{ref}}[\mathbf{v}=W_2T_{h_2};(\CD_0,\CT_0,\vec{\gamma}_0)]=\CI_{M;\textrm{ref}}[\mathbf{v}=-2W_2T_{h_1};(\CD',\CT',\vec{\gamma}'')]\;.
\end{align}
For $\vec{\gamma}'''$, we get the refined 3D index
\begin{align}
    \begin{split}
        &\CI_{M;\textrm{ref}}[\mathbf{v}=W_1T_{h_1};(\CD',\CT',\vec{\gamma}''')]
        \\&=1-q-2q^2+(-2+\eta^{2W_1})q^3+(-2+\eta^{2W_1})q^4+\eta^{2W_1}q^5+(1+\eta^{2W_1})q^6+\dots
    \end{split}
\end{align}
with one dimensional vector space $\mathfrak{F}^{\rm ref}=\mathfrak{F}^{\rm hard}_{\rm d}$ since $(P_1,Q_1)=(-3,1)=(\;\cdot\;,\pm 1)$ so $\mathfrak{F}^{\rm dehn}=\{\mathbf{0}\}$. From the index expression, we can confirm that $\mathfrak{F}^{\rm ref}_{\rm dec}=\{\mathbf{0}\}$ and thus
\begin{align}
    \mathfrak{F}^{\rm ref}_{\rm IR}[M;(\CD',\CT',\vec{\gamma}''')]=\mathfrak{F}^{\rm ref}=\mathrm{Span}\{T_{h_1}\}\;,\quad \dim \mathfrak{F}^{\rm ref}_{\rm IR}=1\;.
\end{align}
Since
\begin{align}
    \dim\mathfrak{F}^{\rm ref}_{\rm IR}[M;(\CD',\CT',\vec{\gamma}''')]=\dim\mathfrak{F}^{\rm ref}_{\rm IR}[M;(\CD_0,\CT_0,\vec{\gamma}_0)]\;,
\end{align}
the injective linear map in \eqref{main conjecture} will be a bijective for this case. We can check it from the following identification
\begin{align}
    \CI_{M;\textrm{ref}}[\mathbf{v}=W_2T_{h_2};(\CD_0,\CT_0,\vec{\gamma}_0)]=\CI_{M;\textrm{ref}}[\mathbf{v}=-2W_2T_{h_1};(\CD',\CT',\vec{\gamma}''')]\;.
\end{align}

\subsection{Closed hyperbolic manifold with ${\rm vol}\simeq 2.06567$, $(S^3\backslash\mathbf{5}_2)_{[3\mu+2\lambda]}$}
We consider the following two Dehn surgery representations of the closed hyperbolic manifold $M=(S^3\backslash\mathbf{5}_2)_{[3\mu+2\lambda]}$:
\begin{align}
\begin{split}
    \CD_0\;&:\;N=S^3\backslash\mathbf{5}_2\;,\quad \{P_I\g_I+Q_I\d_I\}^{d=1}_{I=1}=\{3\mu+2\lambda\}\;,
    \\ \CD\;&:\;N=m007\;,\quad \{P_I\g_I+Q_I\d_I\}^{d=1}_{I=1}=\{-5\a_{\rm sp}+2\b_{\rm sp}\}\;.
\end{split}
\end{align}
Here $(\mu,\lambda)$ is the meridian-longitude pair associated with the boundary components of the knot complement $S^3\backslash\mathbf{5}_2$. $(\a_{\rm sp},\b_{\rm sp})$ is the basis of $H_1(\partial(m007);\mathbb{Z})$ chosen by \texttt{SnapPy}. See \eqref{gluing eqs for 5_2} for the gluing equations associated with the ideal triangulation $\CT_0$ of $S^3\backslash\mathbf{5}_2$. The gluing equations for the ideal triangulation $\CT$ of $m007$ are
\begin{align}
\label{Gluing equation : m007}
\begin{split}
m007 \;:\; & C_1 = Z_1+2Z_2+Z_3\;, \quad C_2 = 2Z_1'+Z_1''+Z_2'+2Z_2''+2Z_3'+Z_3''\;,
\\&C_3=Z_1+Z_1''+Z_2'+Z_3+Z_3''\;,
\\
& hol(\alpha_{\rm sp}) = Z_1'-Z_2+Z_3''\;, \quad hol(\beta_{\rm sp}) = 2Z_1+Z_1''-Z_2'-Z_2''-Z_3-Z_3''\;,
\end{split}
\end{align}
provided by \texttt{SnapPy}.
\paragraph{The refined index $\CI_{M;\rm ref}[\mathbf{v};(\CD_0,\CT_0,\vec{\gamma}_0)]$} From the following choices
\begin{align}
\begin{split}
    &\{H_1,H_2\}=\{C_1,C_2\}\;,
    \\&\{\gamma_1,\delta_1\}=\{\mu,\lambda\}\rightarrow \vec{\gamma}_0=(\mu)\text{ and }(P_1,Q_1)=(3,2)\;,
    \\&\mathbf{a}_{\rm d}=\mathbf{b}_{\rm d}=\mathbf{0}
\end{split}
\end{align}
where $\mu$ is non-closable boundary 1-cycle, we get the refined 3D index
\begin{align}
\begin{split}
    &\CI_{M;\rm ref}[\mathbf{v}=W_2T_{h_2}+VT_a;(\CD_0,\CT_0,\vec{\gamma}_0)]
    \\&=\sum_{m\in\mathbb{Z},e\in\mathbb{Z}/2}\CK_{\rm ref}(3,2;m,e;\eta^V)\CI_{N=S^3\backslash\mathbf{5}_2;\rm ref}(\mathbf{v}=W_2T_{h_2};m,e)
    \\&=1+(-2-\eta^{4V}-\eta^{-2W_2+2V})q+(-2+\eta^{-2V}-\eta^{4V}+\eta^{2W_2}+\eta^{4V+2W_2}-\eta^{-2W_2}
    \\&\qquad -\eta^{2V-2W_2}-\eta^{4V-2W_2}-\eta^{2V-4W_2})q^2+(4+2\eta^{2V}+\eta^{8V}+\eta^{2V+2W_2}+\eta^{4V+2W_2}
    \\&\qquad +\eta^{6V+2W_2}+2\eta^{-2W_2-2V}-\eta^{2V-2W_2}-\eta^{4V-2W_2}-\eta^{2V-4W_2})q^3+\dots
\end{split}
\end{align}
with the two dimensional vector space
\begin{align}
    \begin{split}
        &\mathfrak{F}^{\rm hard}_{\rm d}[M;(\CD_0,\cT_0)]=\mathrm{Span}\{T_{h_2}\}\;,\quad \mathfrak{F}^{\rm dehn}[M;(\CD_0,\vec{\gamma}_0)]=\mathrm{Span}\{T_a\}
        \\&\rightarrow \mathfrak{F}^{\rm ref}[M;(\CD_0,\CT_0,\vec{\gamma}_0)]=\mathfrak{F}^{\rm dehn}\oplus \mathfrak{F}^{\rm hard}_{\rm d}=\mathrm{Span}\{T_a,T_{h_2}\}\;.
    \end{split}
\end{align}
From the index expression, we can confirm that $\mathfrak{F}^{\rm ref}_{\rm dec}[M;(\CD_0,\CT_0,\vec{\gamma}_0)]=\{\mathbf{0}\}$ and thus
\begin{align}
    \mathfrak{F}^{\rm ref}_{\rm IR}[M;(\CD_0,\CT_0,\vec{\gamma}_0)]=\mathfrak{F}^{\rm ref}\;,\quad \dim \mathfrak{F}^{\rm ref}_{\rm IR}=2\;.
\end{align}
We expect $r_{\rm ref}[(S^3\backslash\mathbf{5}_2)_{[3\mu+2\lambda]}]=2$, so the maximally refined index for $M=(S^3\backslash\mathbf{5}_2)_{[3\mu+2\lambda]}$ is $\CI_{M;\rm ref}[\mathbf{v}=W_2T_{h_2}+VT_a;(\CD_0,\CT_0,\vec{\gamma}_0)]$.

\paragraph{The refined index $\CI_{M;\rm ref}[\mathbf{v};(\CD,\CT,\vec{\gamma})]$} From the following choices
\begin{align}
\begin{split}
    &\{H_1,E_1\}=\{C_3,C_1\}\;,
    \\&\{\gamma_1,\delta_1\}=\{\a_{\rm sp},\b_{\rm sp}\}\rightarrow \vec{\gamma}=(\a_{\rm sp})\text{ and }(P_1,Q_1)=(-5,2)\;,
    \\&\mathbf{a}_{\rm d}=\mathbf{b}_{\rm d}=\mathbf{0}
\end{split}
\end{align}
where $\a_{\rm sp}$ is non-closable boundary 1-cycle, we can find the two dimensional vector space
\begin{align}
\begin{split}
    &\mathfrak{F}^{\rm hard}_{\rm d}[M;(\CD,\CT)]=\mathrm{Span}\{T_{h_1}\}\;,\quad \mathfrak{F}^{\rm dehn}[M;(\CD,\vec{\gamma})]=\mathrm{Span}\{T_a\}
    \\&\rightarrow \mathfrak{F}^{\rm ref}[M;(\CD,\CT,\vec{\gamma})]=\mathfrak{F}^{\rm dehn}\oplus\mathfrak{F}^{\rm hard}_{\rm d}=\mathrm{Span}\{T_{a},T_{h_1}\}
\end{split}
\end{align}
and compute the refined 3D index
\begin{align} \label{Refined index : m007}
    \begin{split}
        &\CI_{M;\rm ref}[\mathbf{v}=W_1T_{h_1}+VT_a;(\CD,\CT,\vec{\gamma})]
        \\&=\sum_{m\in\mathbb{Z},e\in\mathbb{Z}/2}\CK_{\rm ref}(-5,2;m,e;\eta^V)\CI_{N=m007;\rm ref}(\mathbf{v}=W_1T_{h_1};m,e)\;.
    \end{split}
\end{align}
We checked $\mathfrak{F}^{\rm ref}_{\rm dec}[M;(\CD,\CT,\vec{\gamma})]=\{\mathbf{0}\}$ from the index expression and thus
\begin{align}
    \mathfrak{F}^{\rm ref}_{\rm IR}[M;(\CD,\CT,\vec{\gamma})]=\mathfrak{F}^{\rm ref}\;,\quad \dim \mathfrak{F}^{\rm ref}_{\rm IR}=2\;.
\end{align}
Since
\begin{align}
    \dim\mathfrak{F}^{\rm ref}_{\rm IR}[M;(\CD,\CT,\vec{\gamma})]=\dim\mathfrak{F}^{\rm ref}_{\rm IR}[M;(\CD_0,\CT_0,\vec{\gamma}_0)]\;,
\end{align}
the injective linear map in \eqref{main conjecture} will be a bijective for this case. We also checked it from the following identification
\begin{align}
    \CI_{M;\rm ref}[\mathbf{v}=W_2T_{h_2}+VT_a;(\CD_0,\CT_0,\vec{\gamma}_0)]=\CI_{M;\rm ref}[\mathbf{v}=-VT_{h_1}-W_2T_a;(\CD,\CT,\vec{\gamma})]\;.
\end{align}

\subsection{Other hyperbolic manifolds} \label{sec : example - other hyperbolic Ms}
We have tested the conjecture \eqref{main conjecture} on additional examples of closed hyperbolic $3$-manifolds using the app \textsc{Refined Index Calculator}, which will be introduced in Appendix~\ref{sec : Refined Index Calculator}. A list of examples is presented below.
\begin{align}
\begin{split}
&(m003)_{[-3 \alpha_{\rm sp} +\beta_{\rm sp}]} = (m006)_{[2\alpha_{\rm sp} +\beta_{\rm sp}]} = (m029)_{[-2 \alpha_{\rm sp} + \beta_{\rm sp}]}, \quad \textrm{vol} \simeq 0.9427073628 , \quad r_{\rm ref} =1,
\\
&(m004)_{[5 \alpha_{\rm sp} +\beta_{\rm sp}]} = (m011)_{[\alpha_{\rm sp} +2\beta_{\rm sp}]} = (m019)_{[2 \alpha_{\rm sp} + \beta_{\rm sp}]}, \quad \textrm{vol} \simeq 0.9813688289 , \quad r_{\rm ref} =1,
\\
&(m007)_{[3\alpha_{\rm sp} +\beta_{\rm sp}]} =  (m010)_{[- \alpha_{\rm sp} +2\beta_{\rm sp}]} , \quad \textrm{vol} \simeq 1.01494160641 , \quad r_{\rm ref} =1,
\\
&(m003)_{[\alpha_{\rm sp} +3\beta_{\rm sp}]} =  (m006)_{[- \alpha_{\rm sp} +2\beta_{\rm sp}]} = (m043)_{[- 2\alpha_{\rm sp} +\beta_{\rm sp}]} , \quad \textrm{vol} \simeq 1.26370923866 , \quad r_{\rm ref} =1,
\\
&(m004)_{[\alpha_{\rm sp} +2\beta_{\rm sp}]} =  (m015)_{[- 3\alpha_{\rm sp} +\beta_{\rm sp}]} = (m032)_{[3\alpha_{\rm sp} +\beta_{\rm sp}]} , \quad \textrm{vol} \simeq 1.3985088842 , \quad r_{\rm ref} =1,
\\
&(m009)_{[4\alpha_{\rm sp} +\beta_{\rm sp}]} =  (m015)_{[ 4\alpha_{\rm sp} +\beta_{\rm sp}]} = (m023)_{[-3\alpha_{\rm sp} +\beta_{\rm sp}]} , \quad \textrm{vol} \simeq 1.41406104417, \quad r_{\rm ref} =1,
\\
&(m003)_{[-\alpha_{\rm sp} +4\beta_{\rm sp}]} =  (m044)_{ [-2\alpha_{\rm sp} +\beta_{\rm sp}]} = (m045)_{[3\alpha_{\rm sp} +\beta_{\rm sp}]} , \quad \textrm{vol} \simeq 1.41406104417, \quad r_{\rm ref} =1,
\\
&(m003)_{[3\alpha_{\rm sp} +\beta_{\rm sp}]} =  (m011)_{[ -2\alpha_{\rm sp} +\beta_{\rm sp}]} = (m026)_{[2\alpha_{\rm sp} +\beta_{\rm sp}]} , \quad \textrm{vol} \simeq 1.42361190029, \quad r_{\rm ref} =1,
\\
&(m004)_{[3\alpha_{\rm sp} +2\beta_{\rm sp}]} =  (m007)_{[ -4\alpha_{\rm sp} +\beta_{\rm sp}]} = (m034)_{[3\alpha_{\rm sp} +\beta_{\rm sp}]} , \quad \textrm{vol} \simeq 1.44069900673, \quad r_{\rm ref} =1,
\\
&(m004)_{[7\alpha_{\rm sp} +\beta_{\rm sp}]} =  (m022)_{ [3\alpha_{\rm sp} +\beta_{\rm sp}]}, \quad \textrm{vol} \simeq 1.4637766449, \quad r_{\rm ref} =0,
\\
&(m003)_{[3\alpha_{\rm sp} +2\beta_{\rm sp}] }=(m019)_{[\alpha_{\rm sp} -2\beta_{\rm sp}]}  = (m037)_{[3\alpha_{\rm sp} +\beta_{\rm sp}]} ,  \quad \textrm{vol} \simeq 1.58316666062,  \quad r_{\rm ref} =2.
\end{split}
\end{align}
For the \(3\)-manifold with \({\rm vol}\simeq 1.4637766449\), one can use the Dehn surgery representation 
\((m022)_{[3\alpha_{\rm sp} + \beta_{\rm sp}]}\) as described above. 
The cusped \(3\)-manifold \(m022\) is a rare example that admits marginal non-closable cycles defined in \eqref{Def : marginal NC}. 
It has two non-closable cycles, \(\alpha_{\rm sp}\) and \(\alpha_{\rm sp}+\beta_{\rm sp}\), 
of which the latter is marginal. 
Choosing \(\{\gamma,\delta\} = \{\alpha_{\rm sp}+\beta_{\rm sp}, \beta_{\rm sp}\}\), 
the Dehn filling becomes non-integral, since
\begin{align}
3\alpha_{\rm sp} + \beta_{\rm sp} = 3\gamma - 2\delta \;.
\end{align}
Despite the non-integral Dehn filling, no additional flavor symmetry arises since \(\gamma\) is marginal, see \eqref{F-ref[M]}, 
and hence \(r_{\rm ref}=0\). This is also consistent with an alternative Dehn surgery representation, $(m004)_{[7 \alpha_{\rm sp} + \beta_{\rm sp}]} 
= (S^3\backslash \mathbf{4}_1)_{[7 \mu + \lambda]}$. 
Treating $p=7$ as sufficiently large, this is compatible with the expectation in \eqref{r-ref for 4-1-p and 5-2-p}.

\subsection{Some non-hyperbolic manifolds}
The 3D index for non-hyperbolic 3-manifolds exhibits qualitatively different behavior from the hyperbolic case \cite{Gang:2018gyt}. In this subsection, we present the fully refined 3D index for several classes of non-hyperbolic 3-manifolds. For closed non-hyperbolic $3$-manifolds obtained via exceptional Dehn fillings on hyperbolic one-cusped $3$-manifolds, we use the database of \cite{dunfield2019censusexceptionaldehnfillings}, available at \cite{DVN/6WNVG0_2018}.

\subsubsection{$\Sigma_{0,3}\times S^1$}
Many non-hyperbolic closed $3$-manifolds can be obtained by Dehn filling on $\Sigma_{0,3}\times S^1$, where $\Sigma_{0,3}$ denotes the $3$-punctured sphere. We choose $\vec{\gamma} = (\g_1 , \g_2, \g_3)$ to be the boundary $1$-cycles enclosing the three punctures, and $\vec{\d} = (\d_1, \d_2, \d_3)$ to be the $1$-cycles along the $S^1$ direction. Note that $\g_i$ are non-closable, since $\Sigma_{0,2}\times S^1$
does not admit any (adjoint-)irreducible $SL(2,\mathbb{C})$ flat
connections. We claim that $r_{\rm ref}[\Sigma_{0,3}\times S^1]=0$ and the refined 3D index for
$N=\Sigma_{0,3}\times S^1$ is given by
\begin{align}
\CI_{N;{\rm ref}}
( \mathbf{v} = \mathbf{0} ; m_1, m_2, m_3, e_1, e_2, e_3)
=
\begin{cases}
\delta_{e_1,0}\delta_{e_2,0}\delta_{e_3,0} & m_1+m_2+m_3 \in 2\mathbb{Z}\;,\\
0 & \text{otherwise}\;.
\end{cases}
\end{align}
We do not have a direct computation, such as one based on an ideal triangulation of the $3$-manifold $N$, to justify this claim. Nevertheless, as we will see, it passes nontrivial consistency checks by comparing the refined 3D index of the Dehn-filled manifolds predicted by this claim with those obtained from alternative Dehn surgery representations.

One can check that 
\begin{align} 
\label{Refined kernel summation property} \sum_{m\in 2\mathbb{Z},2\mathbb{Z}+1}\CK_{\rm ref}(P,Q+\alpha P;m,e=0;\eta)\,,\quad P,Q\in\mathbb{Z}\;,\quad |P|\ge2 
\end{align} 
is independent of the choice of $\alpha \in \mathbb{Z}$. It will be frequently used in the subsequent examples when performing Dehn filling operations on $\Sigma_{0,3}\times S^1$.
Here are some examples of $\sum_{m}\CK_{\rm ref}$ for $Q\neq \pm1 \mod P$:
\begin{align}
\label{Magnetic sum of refined kernel}
    \begin{split}
        (P,Q
        )=(5,2)\;:\;&1+(-1+\eta^2)q+(-2+\frac{1}{\eta^2}+\eta^2)q^2+(-2+\frac{1}{\eta^2}+\eta^2)q^3+(-2+\frac{1}{\eta^2}
        \\&+\eta^4)q^4 +(-\eta^2+\eta^4)q^5+(1+\frac{1}{\eta^4}-\frac{1}{\eta^2}-3\eta^2+2\eta^4)q^6+\dots\;\text{for }m\in 2\mathbb{Z}\;,
        \\&0\quad\text{for }m\in 2\mathbb{Z}+1\;,
        \\(P,Q)=(5,3)\;:\;&\text{Same with }(P,Q)=(5,2)\;,
        \\(P,Q)=(7,2)\;:\;&1+(-1+\eta^2)q+(-2+\frac{1}{\eta^2}+\eta^4)q^2+(-1+\frac{1}{\eta^2}-\eta^2+\eta^4)q^3+(\frac{1}{\eta^2}
        \\&-3\eta^2+2\eta^4)q^4+(4-\frac{1}{\eta^2}-5\eta^2+\eta^4+\eta^6)q^5+\dots\;\text{for }m\in2\mathbb{Z}\;,
        \\&0\quad\text{for }m\in 2\mathbb{Z}+1\;,
        \\(P,Q)=(7,3)\;:\;&1+(-1+\eta^2)q+(-2+\frac{1}{\eta^2}+\eta^2)q^2+(-1+\eta^2)q^3+(\frac{1}{\eta^4}-\frac{1}{\eta^2}-\eta^2
        \\&+\eta^4)q^4+(4+\frac{1}{\eta^4}-\frac{3}{\eta^2}-3\eta^2+\eta^4)q^5+\dots\;\text{for }m\in2\mathbb{Z}\;,
        \\&0\quad\text{for }m\in2\mathbb{Z}+1\;,
        \\(P,Q)=(7,4)\;:\;&\text{Same with }(P,Q)=(7,3)\;,
        \\(P,Q)=(7,5)\;:\;&\text{Same with }(P,Q)=(7,2)\;,
        \\(P,Q)=(8,3)\;:\;&1+(-1+\eta^2)q+(-2+\frac{1}{\eta^2}+\eta^4)q^2+(1-2\eta^2+\eta^4)q^3+(3+\frac{1}{\eta^4}-\frac{2}{\eta^2}
        \\&-5\eta^2+3\eta^4)q^4+(10+\frac{1}{\eta^4}-\frac{5}{\eta^2}-8\eta^2+\eta^4+\eta^6)q^5+\dots\;\text{for }m\in2\mathbb{Z}\;,\quad
        \\&(-1+\eta^2)q^{1/2}+(-2+\frac{1}{\eta^2}+\eta^2)q^{3/2}+(-1+\frac{1}{\eta^2}-\eta^2+\eta^4)q^{5/2}+(2
        \\&-\frac{1}{\eta^2}-3\eta^2+2\eta^4)q^{7/2}+\dots\;\text{for }m\in2\mathbb{Z}+1\;,
        \\(P,Q)=(8,5)\;:\;&\text{Same with }(P,Q)=(8,3)\;.
    \end{split}
\end{align}
When $Q=\pm 1\mod P$, the result is identified with the one evaluated by the unrefined Kernel \eqref{unrefined Dehn filling kernel}
\begin{align}
    \sum_{m\in2\mathbb{Z},2\mathbb{Z}+1}\CK(P,\pm 1;m,e=0)=\begin{cases}
        1&(|P|=2)
        \\\begin{cases}
            1 & m\in 2\mathbb{Z}
            \\0 & m\in 2\mathbb{Z}+1
        \end{cases}&(|P|\ge 3)
    \end{cases}
\end{align}
since there exists a $\alpha_*$ which makes $Q+\alpha_* P=\pm1$. For other choices of $\alpha$, the flavor symmetry generator $T_a$ associated with the Dehn filling decouples in the IR.

\subsubsection{Seifert fibered spaces $M=S^2\bigl((P_1,Q_1),(P_2,Q_2),(P_3,Q_3)\bigr)$}

The Seifert fibered spaces can be represented as
\begin{align}
    \mathcal{D}_0\;:\;
    S^2\!\bigl((P_1,Q_1),(P_2,Q_2),(P_3,Q_3)\bigr)
    \;=\;
    \bigl(\Sigma_{0,3}\times S^1\bigr)_{[P_i\gamma_i+Q_i\delta_i]_{i=1}^3}\;.
\end{align}
The corresponding 3D gauge theory has been studied in \cite{Cho:2020ljj,Choi:2022dju,Gang:2024tlp,Baek2025abc}, whose results suggest that
\begin{align}
\label{Number of refinement, SFS}
\begin{split}
r_{\rm ref}[M]
&=
\sharp(P_1,Q_1)
+\sharp(P_2,Q_2)
+\sharp(P_3,Q_3)\;,
\\
\sharp(P,Q)
&:=
\begin{cases}
0\;, & Q \equiv \pm 1 \pmod{P}\;,\\
1\;, & \text{otherwise}\;.
\end{cases}
\end{split}
\end{align}
In the IR, the gauge theory flows to a product of $r_{\rm ref}$ rank-0 SCFTs, and each factor gives a refinement associated with the $U(1)_A \subset SO(4)_R$ symmetry.

Let the Cartan subalgebra of the refined flavor symmetry be
\begin{align}
\mathfrak{F}^{\rm ref}[M]
=
\mathrm{Span}\!\left\{
T_{a_i}
:\;
i=1,2,3
\ \text{with } \sharp(P_i,Q_i)=1
\right\}\;.
\end{align}
Then, from the Dehn surgery representation $\mathcal{D}_0$, we obtain the fully refined 3D index
\begin{align}
\begin{split}
\label{Refined index : General SFS}
\CI_{M;{\rm ref}}\Bigl[
\mathbf{v}
=
\sum_{i=1}^{3}
\sharp(P_i,Q_i)\, V_i\, T_{a_i}
\;;\;\mathcal{D}_0\Bigr]
=
\sum_{\substack{m_i \in \mathbb{Z} \\ m_1+m_2+m_3 \in 2\mathbb{Z}}}
\prod_{i=1}^{3}
\CK_{\rm ref}(P_i,Q_i; m_i, e_i=0; \eta^{V_i})\;.
\end{split}
\end{align}

\paragraph{$\boldsymbol{M=S^2\bigr((3,1),(3,2),(3,-1)\bigr)}$} We consider the following Dehn surgery representation
\begin{align}
    \CD\;:\;m036\;,\quad \{P_I\gamma_I+Q_I\delta_I\}^{d=1}_{I=1}=\{0\a_{\rm sp}+\b_{\rm sp}\}\;.
\end{align}
Here $(\a_{\rm sp},\b_{\rm sp})$ is a basis of $H_1(\partial(m036);\mathbb{Z})$, as chosen by \texttt{SnapPy}. We use the ideal triangulation $\CT$ of $m036$, as provided by \texttt{SnapPy}, whose gluing equations are
\begin{align}
\label{Gluing equation : m036}
\begin{split}
m036 \;:\;&C_1=2Z_1+Z_1''+Z_2+2Z_2'+Z_3+Z_3'+Z_4+Z_4'\;,\quad C_2=Z_1'+Z_2''+Z_3+Z_4\;,
\\&C_3=Z_1'+Z_1''+Z_3'+Z_3''+Z_4''\;,\quad C_4=Z_2+Z_2''+Z_3''+Z_4'+Z_4''\;,
\\&hol(\alpha_\mathrm{sp})=Z_1''+2Z_2'-Z_3''\;,\quad hol(\beta_\mathrm{sp})=-2Z_1+2Z_2'\;.
\end{split}
\end{align}
From the following choices
\begin{align}
    \begin{split}
        &\{H_1,H_2,E_1\}=\{C_3,C_4,C_2\}\;,
        \\&\{\gamma_1,\delta_1\}=\{\a_{\rm sp},\b_{\rm sp}\}\rightarrow \vec{\gamma}=(\a_{\rm sp})\text{ and }(P_1,Q_1)=(0,1)\;,
        \\&\mathbf{a}_{\rm d}=(0,-1)\;,\quad \mathbf{b}_{\rm d}=(-\frac{1}{2},\frac{1}{2})
    \end{split}
\end{align}
where $\a_{\rm sp}$ is non-closable boundary 1-cycle, we get the refined 3D index
\begin{align} \label{Refined index : m036}
\begin{split}
    &\CI_{M;\rm ref}[\mathbf{v}=W_1T_{h_1}+W_2T_{h_2};(\CD,\CT,\vec{\gamma})]
    \\&=\sum_{m\in\mathbb{Z},e\in\mathbb{Z}/2}\CK(0,1;m,e)\CI_{N=m036;\rm ref}(\mathbf{v}=W_1T_{h_1}+W_2T_{h_2};m,e;\mathbf{a}_{\rm d},\mathbf{b}_{\rm d})
    \\&=1
\end{split}
\end{align}
with the two dimensional vector space
\begin{align}
    \begin{split}
        &\mathfrak{F}^{\rm hard}_{\rm d}[M;(\CD,\CT)]=\mathrm{Span}\{T_{h_1},T_{h_2}\}\;,\quad \mathfrak{F}^{\rm dehn}[M;(\CD,\vec{\gamma})]=\{\mathbf{0}\}
        \\&\rightarrow \mathfrak{F}^{\rm ref}[M;(\CD,\CT,\vec{\gamma})]=\mathfrak{F}^{\rm dehn}\oplus \mathfrak{F}^{\rm hard}_{\rm d}=\mathrm{Span}\{T_{h_1},T_{h_2}\}\;.
    \end{split}
\end{align}
By comparison with $\CI_{M}[\CD_0]=1$, we can confirm that the conjecture \eqref{main conjecture} holds in a non-trivial way, as all elements in the vector space $\mathfrak{F}^{\rm ref}$ are decoupled. We can also find the trivial bijective map : $\mathbf{0}\leftrightarrow \mathbf{0}$.

\paragraph{$\boldsymbol{M=S^2\bigr((2,1),(3,1),(9,-7)\bigr)}$} We consider the following Dehn surgery representation
\begin{align}
    \CD\;:\;m007\;,\quad \{P_I\gamma_I+Q_I\delta_I\}_{I=1}^{d=1}=\{-2\a_{\rm sp}+\b_{\rm sp}\}\;.
\end{align}
Here $(\a_{\rm sp},\b_{\rm sp})$ is a basis of $H_1(\partial(m007);\mathbb{Z})$, as chosen by \texttt{SnapPy}. See \eqref{Gluing equation : m007} for the gluing equations associated with the ideal triangulation $\CT$ of $m007$. From the following choices
\begin{align}
\begin{split}
    &\{H_1,E_1\}=\{C_3,C_1\}\;,
    \\&\{\gamma_1,\delta_1\}=\{\a_{\rm sp},\b_{\rm sp}\}\rightarrow \vec{\gamma}=(\a_{\rm sp})\text{ and }(P_1,Q_1)=(-2,1)\;,
    \\&\mathbf{a}_{\rm d}=\mathbf{b}_{\rm d}=\mathbf{0}
\end{split}
\end{align}
where $\a_{\rm sp}$ is non-closable boundary 1-cycle, we get the refined 3D index
\begin{align}
\begin{split}
    &\CI_{M;\rm ref}[\mathbf{v}=W_1T_{h_1};(\CD,\CT,\vec{\gamma})]
    \\&=\sum_{m\in\mathbb{Z},e\in\mathbb{Z}/2}\CK(-2,1;m,e)\CI_{N=m007;\rm ref}(\mathbf{v}=W_1T_{h_1};m,e)
    \\&=1+(-1+\eta^{-2W_1})q+(-2+\eta^{-4W_1}+\eta^{2W_1})q^2+(-1+\eta^{-6W_1}-\eta^{-2W_1}+\eta^{2W_1})q^3
    \\&\qquad +(\eta^{-6W_1}-2\eta^{-2W_1}+\eta^{2W_1})q^4+(4+2\eta^{-6W_1}-2\eta^{-4W_1}-3\eta^{-2W_1}-\eta^{2W_1})q^5
    \\&\qquad +(5+\eta^{-8W_1}+2\eta^{-6W_1}-4\eta^{-4W_1}-2\eta^{-2W_1}-3\eta^{2W_1}+\eta^{4W_1})q^6+\dots
\end{split}
\end{align}
with the one dimensional vector space
\begin{align}
\begin{split}
    &\mathfrak{F}^{\rm hard}_{\rm d}[M;(\CD,\CT)]=\mathrm{Span}\{T_{h_1}\}\;,\quad \mathfrak{F}^{\rm dehn}[M;(\CD,\vec{\gamma})]=\{\mathbf{0}\}
    \\&\rightarrow \mathfrak{F}^{\rm ref}[M;(\CD,\CT,\vec{\gamma})]=\mathfrak{F}^{\rm dehn}\oplus\mathfrak{F}^{\rm hard}_{\rm d}=\mathrm{Span}\{T_{h_1}\}\;.
\end{split}
\end{align}
From the index expression, we can confirm that $\mathfrak{F}^{\rm ref}_{\rm dec}[M;(\CD,\CT,\vec{\gamma})]=\{\mathbf{0}\}$ and thus
\begin{align}
    \mathfrak{F}^{\rm ref}_{\rm IR}[M;(\CD,\CT,\vec{\gamma})]=\mathfrak{F}^{\rm ref}\;,\quad \dim\mathfrak{F}^{\rm ref}_{\rm IR}=1\;.
\end{align}
Since
\begin{align}
    \mathfrak{F}^{\rm ref}_{\rm IR}[M;(\CD,\CT,\vec{\gamma})]=\mathfrak{F}^{\rm ref}_{\rm IR}[M;\CD_0]\;,
\end{align}
the injective linear map in \eqref{main conjecture} will be a bijective for this case. We checked the following identification
\begin{align}
    \CI_{M;\rm ref}[\mathbf{v}=V_3T_{a_3};\CD_0]=\CI_{M;\rm ref}[\mathbf{v}=-V_3T_{h_1};(\CD,\CT,\vec{\gamma})]\;.
\end{align}

\paragraph{$\boldsymbol{M=S^2\bigr((2,1),(5,2),(5,-3)\bigr)}$} We consider the following Dehn surgery representation
\begin{align}
    \CD\;:\;N=m036\;,\quad \{P_I\gamma_I+Q_I\delta_I\}^{d=1}_{I=1}=\{\a_{\rm sp}+\b_{\rm sp}\}\;.
\end{align}
Here $(\a_{\rm sp},\b_{\rm sp})$ is a basis of $H_1(\partial(m036);\mathbb{Z})$, as chosen by \texttt{SnapPy}. See \eqref{Gluing equation : m036} for the gluing equations associated with the ideal triangulation $\CT$ of $m036$. From the following choices
\begin{align}
    \begin{split}
        &\{H_1,H_2,E_1\}=\{C_3,C_4,C_2\}\;,
        \\&\{\gamma_1,\delta_1\}=\{\a_{\rm sp},\b_{\rm sp}\}\rightarrow \vec{\gamma}=(\a_{\rm sp})\text{ and }(P_1,Q_1)=(1,1)\;,
        \\&\mathbf{a}_{\rm d}=(0,-1)\;,\quad \mathbf{b}_{\rm d}=(-\frac{1}{2},\frac{1}{2})
    \end{split}
\end{align}
where $\a_{\rm sp}$ is non-closable boundary 1-cycle, we can find the two dimensional vector space
\begin{align}
    \begin{split}
        &\mathfrak{F}^{\rm hard}_{\rm d}[M;(\CD,\CT)]=\mathrm{Span}\{T_{h_1},T_{h_2}\}\;,\quad \mathfrak{F}^{\rm dehn}[M;(\CD,\vec{\gamma})]=\{\mathbf{0}\}
        \\&\rightarrow \mathfrak{F}^{\rm ref}[M;(\CD,\CT,\vec{\gamma})]=\mathfrak{F}^{\rm dehn}\oplus\mathfrak{F}^{\rm hard}_{\rm d}=\mathrm{Span}\{T_{h_1},T_{h_2}\}
    \end{split}
\end{align}
and compute the refined 3D index
\begin{align}
    \begin{split}
        &\CI_{M;\rm ref}[\mathbf{v}=W_1T_{h_1}+W_2T_{h_2};(\CD,\CT,\vec{\gamma})]
        \\&=\sum_{m\in\mathbb{Z},e\in\mathbb{Z}/2}\CK(1,1;m,e)\CI_{N=m036;\rm ref}(\mathbf{v}=W_1T_{h_1}+W_2T_{h_2};m,e)\;.
    \end{split}
\end{align}
We checked $\mathfrak{F}^{\rm ref}_{\rm dec}[M;(\CD,\CT,\vec{\gamma})]=\{\mathbf{0}\}$ from the index expression and thus
\begin{align}
    \mathfrak{F}^{\rm ref}_{\rm IR}[M;(\CD,\CT,\vec{\gamma})]=\mathfrak{F}^{\rm ref}\;,\quad \dim \mathfrak{F}^{\rm ref}_{\rm IR}=2\;.
\end{align}
Since
\begin{align}
    \mathfrak{F}^{\rm ref}_{\rm IR}[M;(\CD,\CT,\vec{\gamma})]=\mathfrak{F}^{\rm ref}_{\rm IR}[M;\CD_0]\;,
\end{align}
the injective linear map in \eqref{main conjecture} will be a bijective for this case. We checked the following identification
\begin{align}
\begin{split}
    &\CI_{M;\textrm{ref}}[\mathbf{v}=V_2T_{a_2}+V_3T_{a_3};\CD_0]=\CI_{M;\rm ref}[\mathbf{v}=-V_2T_{h_1}-V_3T_{h_2};(\CD,\CT,\vec{\gamma})]
    \\&=\prod_{i=2}^3\biggr[(1+(-1+\eta^{2V_i})q+(-2+\eta^{2V_i}+\eta^{-2V_i})q^2+(-2+\eta^{2V_i}+\eta^{-2V_i})q^3
    \\&\qquad\qquad+(-2+\eta^{4V_i}+\eta^{-2V_i})q^4+(-\eta^{2V_i}+\eta^{4V_i})q^5+(1+\eta^{-4V_i}-\eta^{-2V_i}-3\eta^{2V_i}
    \\&\qquad\qquad+2\eta^{4V_i})q^6+\dots\biggr]\;. \nonumber
\end{split}
\end{align}

\paragraph{$\boldsymbol{M=S^2\bigr((5,2),(5,2),(5,-3)\bigr)}$} We consider the following Dehn surgery representation
\begin{align}
    \CD\;:\;N=v1683\;,\quad \{P_I\gamma_I+Q_I\delta_I\}_{I=1}^{d=1}=\{0\a_{\rm sp}+\b_{\rm sp}\}\;.
\end{align}
Here $(\a_{\rm sp},\b_{\rm sp})$ is a basis of $H_1(\partial(v1683);\mathbb{Z})$, as chosen by \texttt{SnapPy}. We use the ideal triangulation $\CT$ of $v1683$, as provided by \texttt{SnapPy}, whose gluing equations are
\begin{align}
\label{Gluing equation : v1683}
\begin{split}
v1683 \;:\; & C_1 = 2Z_1+Z_1''+Z_2+Z_3+Z_4+Z_7\;,
\\&C_2 = 2Z_1'+2Z_2''+2Z_3''+Z_4'+2Z_5'+2Z_6'+Z_7'\;,\quad C_3=Z_1''+Z_2'+Z_3'+Z_6''\;,
\\&C_4=Z_2'+Z_4''+Z_5''+Z_7\;,\quad C_5=Z_2+Z_3+Z_5+2Z_6+Z_6''\;,
\\&C_6=Z_3'+Z_4+Z_5''+Z_7''\;,\quad C_7=Z_4'+Z_4''+Z_5+Z_7'+Z_7''\;,
\\
& hol(\alpha_{\rm sp}) = -Z_2'+2Z_3+Z_4'+Z_5-Z_5''+Z_6-Z_6'\;,
\\&hol(\beta_{\rm sp}) = Z_2'-Z_2''+Z_3''-Z_4'-Z_5-Z_5'\;.
\end{split}
\end{align}
From the following choices
\begin{align}
    \begin{split}
        &\{H_1,H_2,E_1,E_2,E_3,E_4\}=\{C_1,C_5,C_2,C_3,C_4,C_6\}\;,
        \\&\{\gamma_1,\delta_1\}=\{\a_{\rm sp},\b_{\rm sp}\}\rightarrow \vec{\gamma}=(\a_{\rm sp})\text{ and }(P_1,Q_1)=(0,1)\;,
        \\&\mathbf{a}_{\rm d}=(0,1)\;,\quad \mathbf{b}_{\rm d}=\mathbf{0}
    \end{split}
\end{align}
where $\a_{\rm sp}$ is non-closable boundary 1-cycle, we can find the two dimensional vector space
\begin{align}
    \begin{split}
        &\mathfrak{F}^{\rm hard}_{\rm d}[M;(\CD,\CT)]=\mathrm{Span}\{T_{h_1},T_{h_2}\}\;,\quad \mathfrak{F}^{\rm dehn}[M;(\CD,\vec{\gamma})]=\{\mathbf{0}\}
        \\&\rightarrow \mathfrak{F}^{\rm ref}[M;(\CD,\CT,\vec{\gamma})]=\mathfrak{F}^{\rm dehn}\oplus \mathfrak{F}^{\rm hard}_{\rm d}=\mathrm{Span}\{T_{h_1},T_{h_2}\}
    \end{split}
\end{align}
and compute the refined 3D index
\begin{align}
    \begin{split}
        &\CI_{M;\rm ref}[\mathbf{v}=W_1T_{h_1}+W_2T_{h_2};(\CD,\CT,\vec{\gamma})]
        \\&=\sum_{m\in\mathbb{Z},e\in\mathbb{Z}/2}\CK(0,1;m,e)\CI_{N=v1683;\rm ref}(\mathbf{v}=W_1T_{h_1}+W_2T_{h_2};m,e;\mathbf{a}_{\rm d},\mathbf{b}_{\rm d})\;.
    \end{split}
\end{align}
We checked $\mathfrak{F}^{\rm ref}_{\rm dec}[M;(\CD,\CT,\vec{\gamma})]=\{\mathbf{0}\}$ from the index expression and thus
\begin{align}
    \mathfrak{F}^{\rm ref}_{\rm IR}[M;(\CD,\CT,\vec{\gamma})]=\mathfrak{F}^{\rm ref}\;,\quad \dim \mathfrak{F}^{\rm ref}_{\rm IR}=2\;.
\end{align}
Note that $\dim\mathfrak{F}^{\rm ref}_{\rm IR}[M;(\CD,\CT,\vec{\gamma})]<\dim \mathfrak{F}^{\rm ref}_{\rm IR}[M;\CD_0]=3$. Consequently, the injective linear map in \eqref{main conjecture} exists as
\begin{align}
\begin{split}
    &F\;:\;\mathfrak{F}^{\rm ref}_{\rm IR}[M;(\CD,\CT,\vec{\gamma})]\;\hookrightarrow\;\mathfrak{F}^{\rm ref}_{\rm IR}[M;\CD_0]\;,
    \\&W_1T_{h_1}+W_2T_{h_2}\mapsto (W_1+W_2)T_{a_1}+W_2T_{a_2}-W_2T_{a_3}
\end{split}
\end{align}
from the following identification
\begin{align}
    \CI_{M;\rm ref}[\mathbf{v}=(W_1+W_2)T_{a_1}+W_2T_{a_2}-W_2T_{a_3};\CD_0]=\CI_{M;\rm ref}[\mathbf{v}=W_1T_{h_1}+W_2T_{h_2};(\CD,\CT,\vec{\gamma})]\;.
\end{align}
We also checked the necessary condition for the maximal refinement \eqref{maximal refinement necessary} is violated as
\begin{align}
    {\rm coeff}_{q^1}\left[\mathcal{I}_{M;\rm ref}[\mathbf{v};(\CD,\CT,\vec{\gamma})]\right]=-3 < \underbrace{-\dim \mathfrak{F}_{\rm IR}[M;(\CD,\CT,\vec{\gamma})]-0}_{=-2}\;.
\end{align}
It indicates that there must be extra refinements. We can find at least one more from $\CI_{M;\rm ref}[\mathbf{v};\CD_0]$.

\paragraph{$\boldsymbol{M=S^2\bigr((5,2),(5,2),(5,-2)\bigr)}$} We consider the following Dehn surgery representation
\begin{align}
    \CD\;:\;N=v1845\;,\quad \{P_I\gamma_I+Q_I\delta_I\}_{I=1}^{d=1}=\{0\a_{\rm sp}+\b_{\rm sp}\}\;.
\end{align}
Here $(\a_{\rm sp},\b_{\rm sp})$ is a basis of $H_1(\partial(v1845);\mathbb{Z})$, as chosen by \texttt{SnapPy}. Using the ideal triangulation $\CT$ of $v1845$, as provided by \texttt{SnapPy}, we checked
\begin{align}
    \dim\mathfrak{F}^{\rm ref}_{\rm IR}[M;(\CD,\CT,\vec{\gamma}=(\a_{\rm sp}))]=\dim{\rm Span}\{T_{h_1},T_{h_2}\}=2<r_{\rm ref}[M]=3
\end{align}
using the refined 3D index obtained from the following choices
\begin{align}
    \begin{split}
        &\{H_1,H_2,E_1,E_2,E_3,E_4\}=\{C_4,C_6,C_1,C_3,C_5,C_7\}\;,
        \\&\{\gamma_1,\delta_1\}=\{\a_{\rm sp},\b_{\rm sp}\}\rightarrow \vec{\gamma}=(\a_{\rm sp})\text{ and }(P_1,Q_1)=(0,1)\;,
        \\&\mathbf{a}_{\rm d}=(1,2)\;,\quad \mathbf{b}_{\rm d}=(-\frac{1}{2},-\frac{1}{2})\;.
    \end{split}
\end{align}
Consequently, we checked the injective linear map in \eqref{main conjecture}
\begin{align}
\begin{split}
    &F\;:\;\mathfrak{F}^{\rm ref}_{\rm IR}[M;(\CD,\CT,\vec{\gamma})]\;\hookrightarrow\;\mathfrak{F}^{\rm ref}_{\rm IR}[M;\CD_0]\;,
    \\&W_1T_{h_1}+W_2T_{h_2}\mapsto W_1T_{a_1}-W_1T_{a_2}-W_2T_{a_3}
\end{split}
\end{align}
from the following identification
\begin{align}
    \CI_{M;\rm ref}[\mathbf{v}=W_1T_{a_1}-W_1T_{a_2}-W_2T_{a_3};\CD_0]=\CI_{M;\rm ref}[\mathbf{v}=W_1T_{h_1}+W_2T_{h_2};(\CD,\CT,\vec{\gamma})]\;.
\end{align}
We also checked the necessary condition for the maximal refinement \eqref{maximal refinement necessary} is violated as
\begin{align}
    {\rm coeff}_{q^1}\left[\mathcal{I}_{M;\rm ref}[\mathbf{v};(\CD,\CT,\vec{\gamma})]\right]=-3 < \underbrace{-\dim \mathfrak{F}_{\rm IR}[M;(\CD,\CT,\vec{\gamma})]-0}_{=-2}\;.
\end{align}
It indicates that there must be extra refinements. We can find at least one more from $\CI_{M;\rm ref}[\mathbf{v};\CD_0]$.

\subsubsection{Torus knot complements $M=S^3\backslash T_{(P,Q)}$}
In \cite{Gang:2025ykf}, the IR phases of $T[S^3\backslash T_{(P,Q)};\mu]$ were studied, and the authors found that
\begin{align}
T[S^3\backslash T_{(P,Q)};\mu] 
=
\begin{cases}
\textrm{unitary TQFT}\;,
& |P-Q|=1\;, \\[0.4em]
\textrm{rank-0 SCFT}\;,
&
\left(Q \equiv \pm 1 \pmod{P}\right)
\oplus
\left(P \equiv \pm 1 \pmod{Q}\right)\;, \\[0.4em]
\textrm{two copies of rank-0 SCFTs}\;,
&
\left(Q \not\equiv \pm 1 \pmod{P}\right)
\land
\left(P \not\equiv \pm 1 \pmod{Q}\right)\;.
\end{cases}
\end{align}
From this, we expect
\begin{align}
\label{Number of refinement, Torus knot complement}
r_{\rm ref}[S^3\backslash T_{(P,Q)}]
=
\begin{cases}
0\;,
& |P-Q|=1\;, \\[0.4em]
1\;,
&
\left(Q \equiv \pm 1 \pmod{P}\right)
\oplus
\left(P \equiv \pm 1 \pmod{Q}\right)\;, \\[0.4em]
2\;,
&
\left(Q \not\equiv \pm 1 \pmod{P}\right)
\land
\left(P \not\equiv \pm 1 \pmod{Q}\right)\;.
\end{cases}
\end{align}
The refinements arise from the $U(1)_A$ symmetries associated with each rank-0 SCFT. 
Here, $U(1)_A$ is a Cartan subgroup of the $SO(4)$ $R$-symmetry that commutes with the $\CN=2$ supercharges. The maximally refined 3D index can be computed using the following Dehn surgery representation:
\begin{align}
    \CD_0\;:\;N=\Sigma_{0,3}\times S^1\;,\quad \{P_I\gamma_I+Q_I\delta_I\}_{I=1}^{d=2}=\{P\gamma_1-R\delta_1,Q\gamma_2+S\delta_2\}
\end{align}
where the integers $(R,S)$ are chosen so that $PS-QR=1$. The boundary one-cycles are related to
$(\gamma_3,\delta_3)$ by
\begin{align}
(\mu, \lambda) = (\gamma_3, \delta_3 + PQ\,\gamma_3)\;.
\end{align}
It translates
\begin{align}
    A\mu+B\lambda=A\gamma_3+B(\delta_3+PQ\gamma_3)=(A+PQB)\gamma_3+B\delta_3
\end{align}
so using \eqref{refined Dehn filling kernel}, we get
\begin{align}
\label{Refined index : Torus knot complement}
\begin{split}
    &\CI_{M=S^3\backslash T_{(P,Q)};\rm ref}^{A\mu+B\lambda}[\mathbf{v}=V_1T_{a_1}+V_2T_{a_2};\CD_0]
    \\&=\sum_{m_1,m_2\in\mathbb{Z}}\sum_{e_1,e_2\in\mathbb{Z}/2}\CK_{\rm ref}(P,-R;m_1,e_1;\eta^{V_1})\CK_{\rm ref}(Q,S;m_2,e_2;\eta^{V_2})
    \\&\qquad\qquad\qquad\qquad\times\CI_{\Sigma_{0,3}\times S^1}(m_1,m_2,m,e_1,e_2,e)\bigr|_{m=-2B,e=A+PQB}
    \\&=\sum_{\substack{m_1,m_2 \in \mathbb{Z} \\ m_1+m_2-2B \in 2\mathbb{Z}}}\CK_{\rm ref}(P,-R;m_1,0;\eta^{V_1})\CK_{\rm ref}(Q,S;m_2,0;\eta^{V_2})\delta_{A,-PQB}\;.
\end{split}
\end{align}
See \eqref{Magnetic sum of refined kernel} for examples of non-trivial $\sum_{m_2}\CK_{\rm ref}$ part. Thanks to the property \eqref{Refined kernel summation property}, the resulting refined 3D index is well-defined under the ambiguity
\begin{align}
    (R,S)\rightarrow (R,S)+(P,Q)\mathbb{Z}
\end{align}
and equipped with the vector space of expected dimension \eqref{Number of refinement, Torus knot complement}. Assuming $Q>P\ge2$ since $T_{(P,Q)}\sim T_{(Q,P)}$ and $\mathrm{gcd}(P,Q)=1$, we can prove the following.

\paragraph{When $\boldsymbol{|P-Q|=1}$} Let $Q=P+ 1$. Then we can choose $(R,S)=(-1,-1)$ that satisfies $PS-QR=1$. Thus, there is no refinement and we get the 3D index
\begin{align}
\begin{split}
    &\CI^{A\mu+B\lambda}_{M;\rm ref}[\CD_0]
    \\&=\sum_{\substack{m_1,m_2 \in \mathbb{Z} \\ m_1+m_2-2B \in 2\mathbb{Z}}}\CK(P,1;m_1,0)\CK(P+ 1,-1;m_2,0)\delta_{A,-P(P+ 1)B}
    \\&=\begin{cases}
        \delta_{A,-6B}&(P=2)\;,
        \\\begin{cases}
            \delta_{A,-PQB}&B\in \mathbb{Z}
            \\0&B\in\mathbb{Z}+1/2
        \end{cases}&(P\ge 3)\;.
    \end{cases}
\end{split}
\end{align}
For example, the 3D index for $M=S^3\backslash T_{(2,3)}=S^3\backslash\mathbf{3}_1$ is evaluated as
\begin{align}
    \CI_{M}(m,e)= \CI^{A \mu+
    B\lambda}_{M;\rm ref}[\CD_0]\bigr|_{A=e,B=-m/2}=\delta_{e,3m}
\end{align}
and the 3D index for $S^3\backslash T_{(3,4)}=S^3\backslash\mathbf{8}_{19}$ is evaluated as
\begin{align}
    \CI_{M}(m,e)= \CI_{M;\rm ref}[\CD_0]\bigr|_{A=e,B=-m/2}=\begin{cases}
        \delta_{e,6m}&m\in 2\mathbb{Z}\;,\\
        0&m\in 2\mathbb{Z}+1\;.
    \end{cases}
\end{align}
We checked results agree with the 3D index obtained from the gluing equations, provided by \texttt{SnapPy}, up to an orientation reversal.

\paragraph{When $\boldsymbol{\left(Q \equiv \pm 1 \pmod{P}\right)
\oplus
\left(P \equiv \pm 1 \pmod{Q}\right)}$} Let $Q=NP\pm 1$ where $N\ge 2$. Then we can choose $(R,S)=\mp (1,N)$ that satisfies $PS-QR=1$ and we get the refined 3D index as
\begin{align}
\begin{split}
    &\CI^{A\mu+B\lambda}_{M;\rm ref}[\mathbf{v}= V_2T_{a_2};\CD_0]
    \\&=\sum_{\substack{m_1,m_2 \in \mathbb{Z} \\ m_1+m_2-2B \in 2\mathbb{Z}}}\CK(P,\pm1;m_1,0)\CK_{\rm ref}(NP\pm 1,\mp N;m_2,0;\eta^{V_2})\delta_{A,-P(NP\pm 1)B}
    \\&=\begin{cases}
        \delta_{A,-PQB}\sum_{m_2\in\mathbb{Z}}\CK_{\rm ref}(Q,\mp N;m_2,0;\eta^{V_2})&(P=2)\;,
        \\\begin{cases}
            \delta_{A,-PQB}\sum_{m_2\in 2\mathbb{Z}}\CK_{\rm ref}(Q,\mp N;m_2,0;\eta^{V_2})&B\in \mathbb{Z}
            \\\delta_{A,-PQB}\sum_{m_2\in2\mathbb{Z}+1}\CK_{\rm ref}(Q,\mp N;m_2,0;\eta^{V_2})&B\in\mathbb{Z}+1/2
        \end{cases}&(P\ge 3)\;.
    \end{cases}
\end{split}
\end{align}
For example, the refined 3D index for $M=S^3\backslash T_{(2,5)}=S^3\backslash\mathbf{5}_1$ is evaluated as
\begin{align}
\begin{split}
    &\CI_{M;{\rm ref}}(\mathbf{v}=V_2T_{a_2};m,e)= \CI^{A \mu+
    B\lambda}_{M;{\rm ref}}[\mathbf{v}=V_2T_{a_2};\CD_0]\bigr|_{A=e,B=-m/2}
    \\&=\delta_{e,5m}\sum_{m_2\in\mathbb{Z}}\CK_{\rm ref}(5,-2;m_2,0;\eta^{V_2})
\end{split}
\end{align}
and the refined 3D index for $M=S^3\backslash T_{(3,5)}=S^3\backslash \mathbf{10}_{124}$ is evaluated as
\begin{align}
\begin{split}
    &\CI_{M;{\rm ref}}(\mathbf{v}=V_2T_{a_2};m,e)= \CI^{A \mu+
    B\lambda}_{M;{\rm ref}}[\mathbf{v}=V_2T_{a_2};\CD_0]\bigr|_{A=e,B=-m/2}
    \\&=\begin{cases}
        \delta_{e,\frac{15}{2}m}\sum_{m_2\in2\mathbb{Z}}\CK_{\rm ref}(5,2;m_2,0;\eta^{V_2}) & m\in 2\mathbb{Z}\;,
        \\\delta_{e,\frac{15}{2}m}\sum_{m_2\in2\mathbb{Z}+1}\CK_{\rm ref}(5,2;m_2,0;\eta^{V_2})=0&m\in2\mathbb{Z}+1\;.
    \end{cases}
\end{split}
\end{align}
We checked results agree with the 3D index obtained from the gluing equations, provided by \texttt{SnapPy}, up to an orientation reversal and background shifts.

\paragraph{When $\boldsymbol{\left(Q \not\equiv \pm 1 \pmod{P}\right)
\land
\left(P \not\equiv \pm 1 \pmod{Q}\right)}$} From conditions for $P$ and $Q$, we can choose $(R,S)$ that satisfies $R,S\ge 2$ and $PS-QR=1$. We can also prove\footnote{It can be proven using proof by contradiction.} that such $R$ and $S$ satisfy
\begin{align}
    -R\neq  \pm 1\mod P\;,\quad S\neq \pm1\mod Q\;.
\end{align}
So we get the refined 3D index with the two dimensional vector space as
\begin{align}
\begin{split}
    &\CI^{A\mu+B\lambda}_{M=S^3\backslash T_{(P,Q)};\rm ref}[\mathbf{v}=V_1 T_{a_1}+V_2 T_{a_2};\CD_0]
    \\&=\sum_{\substack{m_1,m_2 \in \mathbb{Z} \\ m_1+m_2-2B \in 2\mathbb{Z}}}\CK_{\rm ref}(P,-R;m_1,0;\eta^{V_1})\CK_{\rm ref}(Q,S;m_2,0;\eta^{V_2})\delta_{A,-PQB}\;.
\end{split}
\end{align}

\subsubsection{SOL manifolds}

A SOL manifold is a torus bundle $M=\mathbb{T}^2 \times_{\varphi} S^1$ with
$|\mathrm{Tr}(\varphi)|>2$,
\begin{align}
 \mathbb{T}^2 \times_\varphi S^1
 := \mathbb{T}^2 \times I /\sim\; ,
 \qquad (x,0)\sim (\varphi(x),1)\;,
\end{align}
where $\varphi \in SL(2,\mathbb{Z})$ is the monodromy acting on
$\mathbb{T}^2$ as an element of its mapping class group.
In \cite{Choi:2022dju}, the IR phases of the corresponding 3D gauge theories
were analyzed, and it was found that $T[M]$ flows to a unitary TQFT in the IR. Thus we expect
\begin{align}
\label{Number of refinement, Sol}
r_{\rm ref}[M]=0\;,
\qquad
\CI_M =1\;.
\end{align}
We observed that, many representations where a hyperbolic manifold $N$ with a single torus boundary is Dehn filled, there is no refinement in the UV.\footnote{No hard edge and the Dehn filling slope is an integer.} The following is an example we found, which includes a case where a refinement exists in the UV but decouples non-trivially in the IR, so that the conjecture \eqref{main conjecture} holds.

\paragraph{$\boldsymbol\varphi=\begin{pmatrix}
    \mathbf{-2}&\mathbf{-1}\\
    \mathbf{-1}&\mathbf{-1}
\end{pmatrix}$} We consider the following two Dehn surgery representations of the SOL manifold we are interested in:
\begin{align}
\begin{split}
    \CD_0\;&:\;N=m003\;,\quad \{P_I\g_I+Q_I\d_I\}^{d=1}_{I=1}=\{-\a_{\rm sp}+2\b_{\rm sp}\}\;,
    \\ \CD\;&:\;N=m293\;,\quad \{P_I\g_I+Q_I\d_I\}^{d=1}_{I=1}=\{-\a_{\rm sp}+\b_{\rm sp}\}\;.
\end{split}
\end{align}
Here, $(\a_{\rm sp},\b_{\rm sp})$ is a basis of $H_1(\partial(m003);\mathbb{Z})$ and $H_1(\partial(m293);\mathbb{Z})$, as chosen by \texttt{SnapPy}. We use the ideal triangulations $\CT_0$, $\CT$ of $m003$ and $m293$, respectively, as provided by \texttt{SnapPy}. For $\CD_0$, from the following choices
\begin{align}
    \begin{split}
        &H_1=C_1\;,
        \\&\{\gamma_1,\delta_1\}=\{\a_{\rm sp},\beta_{\rm sp}\}\rightarrow \vec{\gamma}_0=(\a_{\rm sp})\text{ and }(P_1,Q_1)=(-1,2)\;,
        \\&\mathbf{a}_{\rm d}=1\;,\quad \mathbf{b}_{\rm d}=-\frac{1}{2}
    \end{split}
\end{align}
where $\a_{\rm sp}$ is non-closable boundary 1-cycle, we get the refined 3D index
\begin{align}
\begin{split}
    &\CI_{M;\rm ref}[\mathbf{v}=W_1 T_{h_1}+VT_a;(\CD_0,\CT_0,\vec{\gamma}_0)]
    \\&=\sum_{m\in\mathbb{Z},e\in\mathbb{Z}/2}\CK_{\rm ref}(-1,2;m,e;\eta^V)\CI_{N=m003,\rm ref}(\mathbf{v}=W_1T_{h_1};m,e;\mathbf{a}_{\rm d},\mathbf{b}_{\rm d})
    \\&=1
\end{split}
\end{align}
with the two dimensional vector space
\begin{align}
\begin{split}
    &\mathfrak{F}^{\rm hard}_{\rm d}[M;(\CD_0,\CT_0)]=\mathrm{Span}\{T_{h_1}\}\;,\quad \mathfrak{F}^{\rm dehn}[M;(\CD_0,\vec{\gamma}_0)]=\mathrm{Span}\{T_a\}
    \\&\rightarrow \mathfrak{F}^{\rm ref}[M;(\CD_0,\CT_0,\vec{\gamma}_0)]=\mathfrak{F}^{\rm dehn}\oplus\mathfrak{F}^{\rm hard}_{\rm d}=\mathrm{Span}\{T_a,T_{h_1}\}\;.
\end{split}
\end{align}
We also computed the refined 3D index with alternative choices of $\{\gamma_1,\delta_1\}$ as shown in \eqref{NC cycles in m003}. In every case, all elements in the vector space $\mathfrak{F}^{\rm ref}$ are decoupled and the resulting index becomes simply $1$. For $\CD$, from the following choices
\begin{align}
    \begin{split}
        &\{E_1,E_2,E_3,E_4\}=\{C_1,C_2,C_3,C_4\}\;,
        \\&\{\gamma_1,\delta_1\}=\{\a_{\rm sp},\b_{\rm sp}\}\rightarrow \vec{\gamma}=(\a_{\rm sp})\text{ and }(P_1,Q_1)=(-1,1)
    \end{split}
\end{align}
where $\a_{\rm sp}$ is non-closable boundary 1-cycle, the vector space becomes zero dimensional
\begin{align}
    &\mathfrak{F}^{\rm hard}_{\rm d}[M;(\CD,\CT)]=\mathfrak{F}^{\rm dehn}[M;(\CD,\vec{\gamma})]=\{\mathbf{0}\}\rightarrow\mathfrak{F}^{\rm ref}[M;(\CD,\CT,\vec{\gamma})]=\mathfrak{F}^{\rm dehn}\oplus \mathfrak{F}^{\rm hard}_{\rm d}=\{\mathbf{0}\}
\end{align}
so there is no refinement. The resulting (un)refined 3D index is
\begin{align}
    \begin{split}
        &\CI_{M;\rm ref}[(\CD,\CT,\vec{\gamma})]=\CI_M
        \\&=\sum_{m\in\mathbb{Z},e\in\mathbb{Z}/2}\CK(-1,1;m,e)\CI_{N=m293}(m,e)=1\;.
    \end{split}
\end{align}
Results agree with the expectation \eqref{Number of refinement, Sol}.

\subsubsection{Graph manifolds}
We follow the definition of graph manifolds in \cite{Choi:2022dju}.
A graph manifold is a prime 3-manifold whose JSJ decomposition \cite{JACO197991,Johannson9999} contains no hyperbolic pieces, and which is neither Seifert fibered nor SOL. In this case, the manifold necessarily contains incompressible tori arising from its nontrivial JSJ decomposition. 
We adopt the notation of \cite{martelli2005dehnfillingmagic3manifold} for graph manifolds.

\paragraph{$\boldsymbol{D^2\bigr((2,1),(2,1)\bigr)\bigcup_{\begin{pmatrix}
    2&1\\
    1&1
\end{pmatrix}}D^2\bigr((2,1),(3,1)\bigr)}$}We consider the following two Dehn surgery representations of the graph manifold we are interested in:
\begin{align}
    \begin{split}
        \CD_0\;&:\;N=m036\;,\quad \{P_I\gamma_I+Q_I\delta_I\}_{I=1}^{d=1}=\{2\a_{\rm sp}+\b_{\rm sp}\}\;,
        \\\CD\;&:\;N=m052\;,\quad \{P_I\gamma_I+Q_I\delta_I\}_{I=1}^{d=1}=\{2\a_{\rm sp}+\b_{\rm sp}\}\;.
    \end{split}
\end{align}
Here $(\alpha_{\rm sp},\beta_{\rm sp})$ is a basis of $H_1(\partial(m036);\mathbb{Z})$ and $H_1(\partial(m052);\mathbb{Z})$, as chosen by \texttt{SnapPy}. We use the ideal triangulations $\CT_0$, $\CT$ of $m036$ and $m052$, respectively, as provided by \texttt{SnapPy}. See \eqref{Gluing equation : m036} for the gluing equations associated with the ideal triangulation $\CT_0$ of $m036$. The gluing equations associated with the ideal triangulation $\CT$ of $m052$ are
\begin{align}
\label{Gluing equation : m052}
\begin{split}
m052 \;:\;&C_1=Z_1+Z_2+2Z_3+Z_4\;,\quad C_2=2Z_1'+Z_1''+2Z_2'+Z_2''+2Z_3''+2Z_4'\;,
\\&C_3=Z_1+Z_1''+Z_2+Z_2''+Z_3'\;,\quad C_4=Z_3'+Z_4+2Z_4''\;,
\\&hol(\alpha_\mathrm{sp})=-Z_1-Z_1''+Z_2+Z_2''\;,\quad hol(\beta_\mathrm{sp})=Z_1'-Z_1''+Z_2'+Z_2''-Z_3'\;.
\end{split}
\end{align}
For $\CD_0$, from the following choices
\begin{align}
    \begin{split}
        &\{H_1,H_2,E_1\}=\{C_3,C_4,C_2\}\;,
        \\&\{\gamma_1,\delta_1\}=\{\a_{\rm sp},\b_{\rm sp}\}\rightarrow \vec{\gamma}_0=(\a_{\rm sp})\text{ and }(P_1,Q_1)=(2,1)\;,
        \\&\mathbf{a}_{\rm d}=(0,-1)\;,\quad \mathbf{b}_{\rm d}=(-\frac{1}{2},\frac{1}{2})
    \end{split}
\end{align}
where $\a_{\rm sp}$ is non-closable boundary 1-cycle, we get the refined 3D index
\begin{align} \label{Refined index : m036}
\begin{split}
    &\CI_{M;\rm ref}[\mathbf{v}=W_1T_{h_1}+W_2T_{h_2};(\CD_0,\CT_0,\vec{\gamma}_0)]
    \\&=\sum_{m\in\mathbb{Z},e\in\mathbb{Z}/2}\CK(2,1;m,e)\CI_{N=m036;\rm ref}(\mathbf{v}=W_1T_{h_1}+W_2T_{h_2};m,e;\mathbf{a}_{\rm d},\mathbf{b}_{\rm d})
    \\&=\sum_{l=0}^\infty \eta^{-(W_1+W_2)l}+(-2-\eta^{W_1-W_2}-\eta^{W_2-W_1})q+(-2-2\eta^{-2W_1}-2\eta^{-2W_2}+\eta^{2W_1}+\eta^{2W_2}
        \\&\qquad-\eta^{W_1-3W_2}-\eta^{W_2-3W_1}-3\eta^{-W_1-W_2}+\eta^{W_1+W_2})q^2+(4-\eta^{-4W_1}
        \\&\qquad -3\eta^{-2W_1}-\eta^{-4W_2}-3\eta^{-2W_2}+4\eta^{W_1-W_2}+4\eta^{W_2-W_1}-2\eta^{-W_1-3W_2}-2\eta^{-3W_1-W_2}
        \\&\qquad -\eta^{W_1-3W_2}-\eta^{W_2-3W_1}+2\eta^{2W_1-2W_2}+2\eta^{2W_2-2W_1}-5\eta^{-W_1-W_2}-3\eta^{-2W_1-2W_2})q^3+\dots
\end{split}
\end{align}
with the two dimensional vector space
\begin{align}
\begin{split}
    &\mathfrak{F}^{\rm hard}_{\rm d}[M;(\CD_0,\CT_0)]=\mathrm{Span}\{T_{h_1},T_{h_2}\}\;,\quad \mathfrak{F}^{\rm dehn}[M;(\CD_0,\vec{\gamma}_0)]=\{\mathbf{0}\}
    \\&\rightarrow \mathfrak{F}^{\rm ref}[M;(\CD_0,\CT_0,\vec{\gamma}_0)]=\mathfrak{F}^{\rm dehn}\oplus \mathfrak{F}^{\rm hard}_{\rm d}=\mathrm{Span}\{T_{h_1},T_{h_2}\}\;.
\end{split}
\end{align}
From the index expression, we can confirm that $\mathfrak{F}^{\rm ref}_{\rm dec}[M;(\CD_0,\CT_0,\vec{\gamma}_0)]=\{\mathbf{0}\}$ and thus
\begin{align}
    \mathfrak{F}^{\rm ref}_{\rm IR}[M;(\CD_0,\CT_0,\vec{\gamma}_0)]=\mathfrak{F}^{\rm ref}\;,\quad \dim \mathfrak{F}^{\rm ref}_{\rm IR}=2\;.
\end{align}
We expect $r_{\rm ref}[M]=2$, so the refined 3D index above is maximally refined.

For $\CD$, from the following choices
\begin{align}
    \begin{split}
        &\{H_1,H_2,E_1\}=\{C_3,C_4,C_1\}\;,
        \\&\{\gamma_1,\delta_1\}=\{\a_{\rm sp},\b_{\rm sp}\}\rightarrow \vec{\gamma}=(\a_{\rm sp})\text{ and }(P_1,Q_1)=(2,1)\;,
        \\&\mathbf{a}_{\rm d}=\mathbf{0}\;,\quad \mathbf{b}_{\rm d}=(\frac{1}{2},0)
    \end{split}
\end{align}
where $\a_{\rm sp}$ is non-closable boundary 1-cycle, we can find the two dimensional vector space
\begin{align}
\begin{split}
    &\mathfrak{F}^{\rm hard}_{\rm d}[M;(\CD,\CT)]=\mathrm{Span}\{T_{h_1},T_{h_2}\}\;,\quad \mathfrak{F}^{\rm dehn}[M;(\CD,\vec{\gamma})]=\{\mathbf{0}\}
    \\&\rightarrow \mathfrak{F}^{\rm ref}[M;(\CD,\CT,\vec{\gamma})]=\mathfrak{F}^{\rm dehn}\oplus\mathfrak{F}^{\rm hard}_{\rm d}=\mathrm{Span}\{T_{h_1},T_{h_2}\}
\end{split}
\end{align}
and compute the refined 3D index
\begin{align}
\begin{split}
    &\CI_{M;\rm ref}[\mathbf{v}=W_1T_{h_1}+W_2T_{h_2};(\CD,\CT,\vec{\gamma})]
    \\&=\sum_{m\in\mathbb{Z},e\in\mathbb{Z}/2}\CK(2,1;m,e)\CI_{N=m052;\rm ref}(\mathbf{v}=W_1T_{h_1}+W_2T_{h_2};m,e;\mathbf{a}_{\rm d},\mathbf{b}_{\rm d})\;.
\end{split}
\end{align}
We checked $\mathfrak{F}^{\rm ref}_{\rm dec}[M;(\CD,\CT,\vec{\gamma})]=\{\mathbf{0}\}$ from the index expression and thus
\begin{align}
    \mathfrak{F}^{\rm ref}_{\rm IR}[M;(\CD,\CT,\vec{\gamma})]=\mathfrak{F}^{\rm ref}\;,\quad \dim \mathfrak{F}^{\rm ref}_{\rm IR}=2\;.
\end{align}
Since
\begin{align}
    \mathfrak{F}^{\rm ref}_{\rm IR}[M;(\CD,\CT,\vec{\gamma})]=\mathfrak{F}^{\rm ref}_{\rm IR}[M;(\CD_0,\cT_0,\vec{\gamma}_0)]\;,
\end{align}
the injective linear map in \eqref{main conjecture} will be a bijective for this case. We checked the following identification
\begin{align}
    \CI_{M;\textrm{ref}}[\mathbf{v}=W_1T_{h_1}+W_2T_{h_2};(\CD_0,\CT_0,\vec{\gamma}_0)]=\CI_{M;\textrm{ref}}[\mathbf{v}=W_1T_{h_1}+W_2T_{h_2};(\CD,\CT,\vec{\gamma})]\;.
\end{align}
Note that the resulting refined 3D index is invariant under the exchange of $W_1$ and $W_2$. For the first representation $\CD_0$, we can check that the refined 3D index for $N=m036$ with background shifts also has such invariance. So we can understand it as an inheritance. However, it is not the case for the second representation $\CD$. The refined 3D index for $N=m052$ with background shifts is NOT invariant under the exchange of $W_1$ and $W_2$. Such invariance is realized non-trivially during the Dehn filling procedure. This gives supporting evidence that the refined 3D index is a good invariant.

Another notable and interesting point is that this provides an example in which the refined index regularizes the divergence of the unrefined 3D index. Upon turning off the refinements $W_1$ and $W_2$, the index begins as $\infty - 4q - 8q^2 + \ldots$.  With refinement, the divergent term $\infty$ is replaced by the sum
\begin{align}
\sum_{l=0}^{\infty} \eta^{-(W_1+W_2)l}\,,
\end{align}
whose coefficients are all finite (equal to $1$), counting Hilbert space dimensions at fixed charges of the additional symmetries. We studied some other graph manifolds and found that the same thing happens. For example, from the following Dehn surgery representation
\begin{align}
    M=D^2\bigr((2,1),(2,1)\bigr){\bigcup}_{\left(\begin{array}{cc}
        -1&1\\
        0&1
    \end{array}\right)}D^2\bigr((3,1),(3,1)\bigr)=(m154)_{[2\alpha_{\rm sp}+\beta_{\rm sp}]}\;,
\end{align}
we computed the refined 3D index $\CI_{M;\rm ref}$ starting from $\CI_{N=m154;\rm ref}$ which has the two dimensional $\mathfrak{F}^{\rm ref}=\mathfrak{F}^{\rm hard}_{\rm d}=\textrm{Span}\{T_{h_1},T_{h_2}\}$. We found
\begin{align}
    \begin{split}
        &\CI_{M;\rm ref}[\mathbf{v}=W_1T_{h_1}+W_2T_{h_2}]
        \\&=\sum_{l=0}^\infty \eta^{2W_2l}-2q+(-3+\eta^{2W_1}+\eta^{-2W_2}+\eta^{-2W_1-2W_2}-2\eta^{2W_2})q^2+(\eta^{2W_1}-\eta^{-2W_1-4W_2}
        \\&\qquad +2\eta^{-2W_2}+\eta^{-2W_1-2W_2}-\eta^{2W_1-2W_2}-4\eta^{2W_2}-2\eta^{4W_2})q^3+(4+\eta^{-2W_1}+\eta^{2W_1}
        \\&\qquad -\eta^{-2W_1-4W_2}+\eta^{-2W_2}+\eta^{-2W_1-2W_2}-\eta^{2W_1-2W_2}-5\eta^{2W_2}+\eta^{2W_1+2W_2}
        \\&\qquad -4\eta^{4W_2}-2\eta^{6W_2})q^4+\dots
    \end{split}
\end{align}
From a physics perspective, we expect that the divergence of 3D index arises  from chiral primaries with a ``bad'' R-charge assignment. In the 3D--3D correspondence, such operators can be understood as contributions from M2-branes wrapping incompressible surfaces, with geometric R-charge \(R_{\rm geo}=2g-2\) \cite{Gadde:2013wq}. For graph manifolds, incompressible tori give rise to chiral primaries with \(R_{\rm geo}=0\), which can lead to divergences. It is then natural to expect that such divergences are regulated by a \(U(1)\) flavor symmetry acting on these operators, as seen above; however, in some cases, no such \(U(1)\) symmetry is available.

\paragraph{$\boldsymbol{D^2\bigr((2,1),(2,1)\bigr)\bigcup_{\begin{pmatrix}
    0&1\\
    1&0
\end{pmatrix}}D^2\bigr((2,1),(3,1)\bigr)}$} We consider various Dehn surgery representations of the graph manifold we are interested in:
\begin{align}
    \begin{split}
        \CD_0\;&:\;N=S^3\backslash\mathbf{4}_1\;,\quad \{P_I\gamma_I+Q_I\delta_I\}_{I=1}^{d=1}=\{4\mu+\lambda\}\;,
        \\\CD\;&:\;N=m039\;,\quad \{P_I\gamma_I+Q_I\delta_I\}_{I=1}^{d=1}=\{3\a_{\rm sp}+\b_{\rm sp}\}\;,
        \\\CD'\;&:\;N=m196\;,\quad \{P_I\gamma_I+Q_I\delta_I\}_{I=1}^{d=1}=\{-\a_{\rm sp}+\b_{\rm sp}\}\;,
        \\\CD''\;&:\;N=m364\;,\quad \{P_I\gamma_I+Q_I\delta_I\}_{I=1}^{d=1}=\{\a_{\rm sp}+\b_{\rm sp}\}\;,
        \\\CD'''\;&:\;N=v2408\;,\quad \{P_I\gamma_I+Q_I\delta_I\}_{I=1}^{d=1}=\{\a_{\rm sp}+\b_{\rm sp}\}\;,
        \\\vdots
    \end{split}
\end{align}
In every case, we get the unrefined 3D index
\begin{align}
    \infty-2q^2-2q^3-4q^4-4q^5-6q^6-4q^7-8q^8-6q^9-8q^{10}-6q^{11}-10q^{12}-6q^{13}-\dots
\end{align}
Even though \(N=m364\) and \(N=v2408\) admit a one-dimensional vector space \(\mathfrak{F}^{\rm hard}_{\rm d}\) compatible with Dehn filling, it decouples during the Dehn filling procedure, resulting in the unrefined 3D index. As mentioned above, it is natural to expect that the divergent part is refined by a certain \(U(1)\) flavor symmetry; however, this symmetry does not appear to be accessible in the UV.

\subsubsection{Connected sum}
Consider a 3-manifold given by a connected sum of two 3-manifolds, neither of which is the 3-sphere. Such a manifold contains incompressible spheres, which correspond to chiral primaries with \(R_{\rm geo} = -2\) and contribute to the 3D index as \(\sum_{n=0}^\infty q^{-n}\). Hence, the unrefined index is expected to diverge severely \cite{Choi:2022dju}. We have checked that, in many representations of such manifolds where a hyperbolic manifold \(N\) with a single torus boundary is Dehn filled, there is no refinement in the UV\footnote{No hard edge and the Dehn filling slope is an integer.}, and the resulting unrefined 3D index exhibits an unregularized divergence of the form $\sum_{n=0}^{\infty }q^{-n}$.

\section{Discussion and Future directions}
Based on our work, there are many directions to pursue in future work.

\paragraph{Relevance and independence of the set $\{\CO_{E_I}\}$ under F-term relations.}
When we speak of superpotential deformations by the chiral primary operators
$\{ \CO_I=\CO_{E_I}\}_{I=1}^{\sharp_E}$ in
\[
\CT_0 := g_{\rm NZ}\cdot (T_\Delta)^{\otimes r},
\]
we are assuming that there exists an ordering
$\sigma \in S_{\sharp_E}$ such that
\[
\CO_{\sigma(1)},\CO_{\sigma(2)},\ldots,\CO_{\sigma(\sharp_E)}
\]
can be turned on successively, generating an RG flow
\[
\CT_0 := g_{\rm NZ}\cdot (T_\Delta)^{\otimes r}
\xrightarrow{\delta W=\CO_{\sigma(1)}} \CT_1
\rightarrow \cdots \rightarrow
\CT_{\sharp_E-1}
\xrightarrow{\delta W=\CO_{\sigma(\sharp_E)}} \CT_{\sharp_E},
\]
where $\CT_{\sharp_E}$ is identified with $T[N;\vec{\alpha},\vec{\gamma}]$ in
\eqref{T[N] theory}. For this sequence of RG flows to be well defined, we require that at each step
\[
\begin{split}
&\text{i)}\quad \CO_{\sigma(i)} \text{ remains a nontrivial chiral primary operator in } \CT_{i-1},\\
&\text{ii)}\quad \CO_{\sigma(i)} \text{ is relevant in } \CT_{i-1}.
\end{split}
\]
Condition i) is nontrivial because, due to the F-term relations generated by
\[
\{\CO_{\sigma(1)},\ldots,\CO_{\sigma(i-1)}\},
\]
the operator $\CO_{\sigma(i)}$ may become trivial in the chiral ring of $\CT_{i-1}$.
Condition ii) requires that
\[
R_{\rm IR}\big(\CO_{\sigma(i)}\big)<2,
\]
where $R_{\rm IR}$ denotes the infrared superconformal $R$-charge in $\CT_{i-1}$,
which can in principle be determined by $F$-maximization \cite{Jafferis:2010un}. Verifying these two conditions in general is highly nontrivial, and throughout this paper we assume that there exists an appropriate ordering of $\{\CO_I\}$ for which both conditions are satisfied.
For some exotic choices of ideal triangulation, the conditions may not be satisfied.\footnote{A potential example where this issue arises is the case 
$N = S^3 \backslash \mathbf{4}_1$ (the figure-eight knot complement) 
with an ideal triangulation consisting of six tetrahedra. 
The internal edges are
\begin{align}
\begin{split}
&C_1 = X+W + 2(R'+S'+Z''), \quad 
C_2 = R+Y+2(Z'+W'+S''), \quad C_3 = S+W+2(R''+X''+Y'), 
\\
&C_4 = R+Z+2(Y''+W''+X'), \quad  C_5 = X+Y, \quad 
C_6 = S+Z . \nonumber
\end{split}
\end{align}
All of them are easy edges. 
However, if we first deform the theory by $\CO_{C_5}$ and $\CO_{C_6}$, 
the remaining operators $\CO_{C_i}$ with $i=1,\ldots,4$ become trivial 
due to the F-term relations.}
In such cases, one should instead choose an appropriate subset of $\{\CO_I\}$ for which the above sequence is well defined and which still breaks the flavor symmetry as much as possible through superpotential deformations. Even when these conditions are not fully satisfied, the refined index in \eqref{refined index in charge basis},
computed under the assumption that they hold, still probes a reduced
flavor vector space $\mathfrak{F}^{\rm ref}_{\rm IR}[M]$, and the main
conjecture in \eqref{main conjecture} is expected to remain valid.

\paragraph{Divergent refined index for some non-hyperbolic manifolds} As we saw in the examples of Section \ref{sec : examples}, for some non-hyperbolic 3-manifolds, graph manifolds and connected sums, the 3D index diverges for every choice of $(\CD,\CT,\vec{\gamma})$. In some graph-manifold examples, the refined index cures this divergence, whereas for other manifolds it remains divergent. The pattern of divergence is not random. For instance, in certain graph-manifold cases, only the $q^0$ term diverges, while all higher-order terms are finite and independent of the choice of $(\CD,\CT,\vec{\gamma})$. We expect that there should exist a systematic way to handle these divergences and extract meaningful invariants even in such cases.
Recently, there has been progress in understanding supersymmetric partition functions of such “bad theories” with $\CN=4$ supersymmetry \cite{Giacomelli:2023zkk,Giacomelli:2024laq,Comi:2025zwu}.

\paragraph{Refinements of perturbative Chern--Simons invariants}
We can extend our analysis to the squashed 3-sphere partition function \cite{Hama:2011ea}. Via the 3D--3D correspondence, the supersymmetric partition function is related to the $SL(2,\mathbb{C})$ Chern--Simons partition function at quantized level $k=1$ \cite{Dimofte:2011ju,Cordova:2013cea,Dimofte:2014zga}. Expanding the partition function in the limit $b^2 \to 0$ (where $b$ denotes the squashing parameter), one recovers the perturbative invariants of complex Chern--Simons theory \cite{Dimofte:2012qj}, such as the Chern--Simons action and the adjoint Reidemeister torsion associated with flat connections.
By turning on additional parameters associated with extra flavor symmetries in the supersymmetric partition function, one can obtain refined perturbative Chern--Simons invariants. At present, refinement is understood only via the 3D field-theoretic construction of $T[M;\vec{\alpha}]$. Although we have provided a prescription for obtaining the refined index from refined normal surface counting, as discussed in Section \ref{sec : index from normal surface counting}, its precise physical or topological meaning remains obscure. The study of refined perturbative invariants may instead reveal a purely topological interpretation.
    %%%%%%%%%%%%%%%%%%%%%%%%%%%%%%%%%%%%%%%%%%%%%%%%%%%
	\section*{Acknowledgements}
	We are grateful to  Sebastiano Garavaglia, Sungjoon Kim, Sara Pasquetti and Kazuya Yonekura for valuable discussions related to this work.  This work was supported in part by the National Research Foundation of Korea (NRF) grant NRF-2022R1C1C1011979. We also acknowledge support from the National Research Foundation of Korea (NRF) grant RS-2024-00405629.

\newpage 

	\appendix
	
		\section{$T[SU(2)]$ theory and refined Dehn filling Kernel} \label{appendix : I-TSU(2)}
        The \(T[SU(2)]\) theory \cite{Gaiotto:2008ak} plays several important roles in the construction of the 3D gauge theory \(T[M;\vec{\alpha}]\), and consequently in the refined 3D index. First, in the 6D construction of \(T[M;\vec{\alpha}]\) described in \eqref{T[M] from 6D}, a regular codimension-two defect is realized, upon reduction on \(S^1\), by coupling to the \(T[SU(2)]\) theory. Second, the Dehn filling operation can be implemented by coupling \(T[SU(2)]\) theories, as explained in \eqref{T[M] theory}. Using this Dehn filling description, we construct the refined Dehn filling kernel \(\CK_{\rm ref}(P,Q;m,e;\eta)\), which plays a crucial role in the refined 3D index \eqref{refined index in charge basis}.

The \(T[SU(2)]\) theory is a 3D \(\CN=4\) SQED with \(N_f=2\), namely a \(U(1)\) gauge theory coupled to two fundamental hypermultiplets. In terms of \(\CN=2\) multiplets, the theory can be described as a \(U(1)\) vector multiplet at zero Chern--Simons level, coupled to five chiral multiplets, \(\Phi_0\) and \(\Phi_{1\leq i \leq 4}\), with superpotential
\[
\CW \propto \Phi_0 (\Phi_1 \Phi_2 + \Phi_3 \Phi_4)\,.
\]
Here \(\Phi_0\) arises from the adjoint chiral multiplet in the \(\CN=4\) vector multiplet, while \((\Phi_1,\Phi_2)\) and \((\Phi_3,\Phi_4)\) form two hypermultiplets.
The theory has a \(U(1)_{\rm top}\times SU(2)_{H}\) flavor symmetry that commutes with the \(\CN=4\) supersymmetry, as well as an \(SO(4)_{R}\) R-symmetry. The \(U(1)_{\rm top}\) arises as the topological symmetry associated with the \(U(1)\) gauge field, while \(SU(2)_H\) rotates the two hypermultiplets. In the IR, the \(U(1)_{\rm top}\) symmetry is enhanced to an \(SU(2)\) symmetry, which we denote by \(SU(2)_{C}\). The \(SO(4)_R\) symmetry contains a Cartan subgroup \(U(1)_R \times U(1)_A\), where \(U(1)_R\) corresponds to the superconformal R-symmetry of the \(\CN=2\) subalgebra, while \(U(1)_A\) is an axial flavor symmetry in the \(\CN=2\) description.
\begin{table}[h]
\centering
\begin{tabular}{|c|c|c|c|c|}
\hline
  & $Q$ & $T^3_{H}$ & $A$ & $R$ \\
  \hline
\(\Phi_0\) & $0$ & $0$ & $-2$ & $2$\\
\hline
  \(\Phi_1\) & $-1$& $-1$ & $+1$ & $0$  \\
\hline
  \(\Phi_2\) & $+1$ & $-1$ & $+1$ & $0$ \\
 \hline
 \(\Phi_3\) & $-1$ & $+1$  & $+1$ & $0$ \\
 \hline
  \(\Phi_4\) & $+1$ & $+1$ & $+1$  & $0$ \\
\hline
\end{tabular}
\caption{Charge assignments of the \(T[SU(2)]\) theory in terms of \(\CN=2\) chiral multiplets. \(Q\) denotes the \(U(1)\) gauge charge, and \(T^3_H\) is the Cartan generator of the \(SU(2)_H\) flavor symmetry. \(A\) denotes the charge under \(U(1)_A\), while \(R\) is a R-charge. The superconformal R-charge is given by \(R+\frac{1}{2}A\). All the chirals are neutral under the $SU(2)_C$. }
\label{table : T[SU(2)]}
\end{table}
\\
Using the charge assignments, the generalized superconformal index for the \(T[SU(2)]\) theory is given by (\(m_i \in \tfrac{1}{2}\mathbb{Z}\)):
        \begin{align}
        \begin{split}
        &\CI_{T[SU(2)]} (m_1, u_1, m_2, u_2;\eta) 
        \\
        &= \sum_{m\in m_1 + \mathbb{Z}} \oint_{|u|=1} \frac{du}{2\pi \mathbf{i} u} u^{2(m+m_2)}u_1^{2m_1} u_2^{2m} \eta^{2(m+m_2)}  \CI^{\rm fug}_\Delta \left(0,\frac{q}{\eta^2}\right) \CI^{\rm fug}_\Delta \left(-m-m_1, \frac{1}{u u_1}\right)
        \\
        &\qquad  \times \CI^{\rm fug}_\Delta \left(m-m_1, \frac{u \eta^2}{u_1}\right) \CI^{\rm fug}_\Delta \left(-m+m_1, \frac{u_1}u\right) \CI^{\rm fug}_\Delta \left(m+m_1, u u_1 \eta^2 \right),
        \end{split}
        \end{align}
where $\CI^{\rm fug}_\Delta$ is the tetrahedron index in fugacity basis \eqref{tetrahedron index -fugacity}.
The index depends on the following variables:
\begin{align}
\begin{split}
& (m_1, u_1) \;:\; \text{monopole flux and fugacity for } SU(2)_H\,, \\
& (m_2, u_2) \;:\; \text{monopole flux and fugacity for } SU(2)_C\,, \\
& \eta \;:\; \text{fugacity for } U(1)_A\,.
\end{split}
\end{align}
In the formula, we assign the \(U(1)_A\) charge as \(A+Q\) rather than \(A\). This choice does not affect the result, since the index counts only gauge-invariant operators satisfying \(Q=0\). We take into account the background Chern--Simons level \(-\tfrac{1}{2}\) for the \(T_\Delta\) theory, and include the factor \(u^{2m}\) to compensate for this offset. We use a non-superconformal R-charge in the expression. The index at the superconformal R-charge can be obtained by shifting \(\eta \to \eta\,(-q^{1/2})^{1/2}\). As consistency checks, one can verify that the index enjoys the following properties:
        \begin{align}
        \begin{split}
        &\mathbb{Z}_2 \textrm{ Weyl symmetries} 
        \\
        &: \CI_{T[SU(2)]} (m_1, u_1, m_2, u_2;\eta)  = \CI_{T[SU(2)]} (\epsilon_1 m_1, u_1^{\epsilon_1}, \epsilon_2 m_2, u_2^{\epsilon_2};\eta)\;, \quad \epsilon_{1,2}  \in \{ \pm 1\}
        \\
        &\textrm{Self-mirror property}
        \\
        &: \CI_{T[SU(2)]} \left(m_1, u_1, m_2, u_2;\eta(-q^{1/2})^{1/2}\right)  = \CI_{T[SU(2)]} \left(m_2, u_2, m_1, u_1 ;  \eta^{-1} (-q^{1/2})^{1/2}\right)\;.
        \end{split}
        \end{align}
Using the generalized index, as a first step toward constructing the refined Dehn filling kernel, we define the \(S\)-transformation kernel in the charge basis, \(\CI_S(m_1, e_1, m_2, e_2;\eta)\), as follows:
\begin{align}
\begin{split}
&\Delta(m_1, u_1)\,\CI_{T[SU(2)]}(m_1, u_1, m_2, u_2;\eta) := \sum_{e_1, e_2 \in \tfrac{1}{2}\mathbb{Z}} 
\CI_S(2m_1, e_1, 2m_2, e_2;\eta)\, u_1^{2e_1} u_2^{2e_2}\,,
\end{split}
\end{align}
where
\begin{align}
\Delta(m,u) := (-1)^{2m}\bigl(q^m + q^{-m} - u^{-2} - u^2\bigr)\,,
\end{align}
is the contribution of the \(SU(2)\) vector multiplet. More explicitly, the kernel is given by (\(m_i \in \mathbb{Z}, \; e_i \in \tfrac{1}{2}\mathbb{Z}\)):
\begin{align}
\begin{split}
\CI_S(m_1, e_1, m_2, e_2;\eta) &:= \frac{1}{2} (-1)^{m_1} \Bigl[ (q^{\frac{m_1}{2}} + q^{-\frac{m_1}{2}})\, \widetilde{\CI}_S(m_1, e_1, m_2, e_2;\eta) 
\\
&\quad - \widetilde{\CI}_S(m_1, e_1 - 1, m_2, e_2;\eta) - \widetilde{\CI}_S(m_1, e_1 + 1, m_2, e_2;\eta) \Bigr]\,,
\end{split}
\end{align}
where
\begin{align}
    \begin{split}
        &\widetilde{\CI}_S(m_1,e_1,m_2,e_2;\eta)
        \\&=(-q^{1/2})^{e_1+e_2+\frac{m_1-m_2}{2}}
        \\&\quad\times\sum_{e_\eta\in\mathbb{Z}}\bigr(\eta(-q^{1/2})^{-1}\bigr)^{e_\eta}\CI_\Delta(-e_1-\frac{m_2}{2},e_1+\frac{m_1-e_\eta}{2})\CI_\Delta(e_1+\frac{m_2}{2},e_2-\frac{m_2+e_\eta}{2})
        \\&\quad\times\sum_{t\in\mathbb{Z}}\eta^{2t}\CI_\Delta(-e_2-\frac{m_1}{2},e_2+\frac{m_1+2t}{2})\CI_\Delta(e_2+\frac{m_1}{2},e_1-\frac{m_2-2t}{2})\;.
    \end{split}
\end{align}
Here $\CI_\Delta$ is the tetrahedron index in charge basis given in \eqref{tetrahedron index}. We used the pentagon identity of tetrahedron index \eqref{Tetra pentagon} to get the simplified expression. One can check that the kernel drastically simplifies in the unrefined limit, \(\eta \to 1\), as
\begin{align}
\CI_S(m_1, e_1, m_2, e_2;\eta)\big|_{\eta \to 1}
= \frac{1}{2} \left( \delta_{m_1 - 2e_2}\, \delta_{m_2 - 2e_1}
+ \delta_{m_1 + 2e_2}\, \delta_{m_2 + 2e_1} \right)\,.
\label{I_S in unrefined limit}
\end{align}
Finally, the refined kernel is given by
\begin{align} 
\begin{split} 
&\CK_{\rm ref} (P,Q;m,e;\eta) \\ &=\sum_{m_1, \ldots, m_{\ell-1} \in \mathbb{Z}}\sum_{e_1, \ldots, e_{\ell-1} \in \mathbb{Z}/2} \CI_S (m, -e-\frac{k_1}2 m, m_1, e_1) \CI_S (m_1, -e_1-\frac{k_2}2 m_1, m_2, e_2) \\ &\qquad \qquad \qquad \qquad \qquad \;\; \ldots \CI_S (m_{\ell-2}, -e_{\ell-2}-\frac{k_{\ell-1}}2 m_{\ell-2}, m_{\ell-1}, e_{\ell-1}) \CK(k_\ell,1 ; m_{\ell-1},e_{\ell-1}) \;.\label{refined Dehn filling kernel} \end{split} 
\end{align} 
Here the integer-valued vector $\vec{k} = (k_1, \ldots, k_\ell)$ is determined from $(P,Q)$ via the relation \eqref{vec-k from P/Q}, and $\CK$ denotes the unrefined Dehn filling kernel defined in \eqref{unrefined Dehn filling kernel}. 
The vector $\vec{k}$ is not uniquely determined; however, different choices ultimately lead to the same kernel. In the fugacity basis, the Dehn filling kernel can be written as
\begin{align}
\begin{split}
&\sum_{m \in \mathbb{Z},e \in \frac{\mathbb{Z}}2}\CK_{\rm ref}(P,Q;m,e;\eta) \CI(m,e) 
\\
& = \left( \sum_{m\in \frac{\mathbb{Z}}2} \oint_{|u|=1}\frac{\Delta (m,u)du}{2\pi \mathbf{i} u} \right) \prod_{i=1}^{\ell-1 }\left( \sum_{m_i \in \frac{\mathbb{Z}}2} \oint_{|u_i|=1}\frac{\Delta (m_i, u_i)du_i}{2\pi \mathbf{i} u_i} \right) \CI^{\rm fug}(2m, u^2) 
\\
& \quad \times u^{2k_1 m}    \CI_{\rm T[SU(2)]} (m, u, m_1, u_1)  \left(  \prod_{i=1}^{\ell-2} u_i^{2k_{i+1} m_{i}}  \CI_{\rm T[SU(2)]} (m_i, u_i,  m_{i+1}, u_{i+1}) \right)  u_{\ell-1}^{2k_\ell m_{\ell-1}}\;, \label{refined Dehn filling Kernel in fugacity basis}
\end{split}
\end{align}
where we define the index in the fugacity basis as
\begin{align}
\CI^{\rm fug}(m,u) \;:=\; \sum_{e\in \frac{\mathbb{Z}}{2}} \CI(m,e)\,u^{e}\, .
\end{align}
Using \eqref{I_S in unrefined limit} and \eqref{refined Dehn filling kernel}, one can show that
\begin{align}
\begin{split}
&\CK_{\rm ref}(P,Q; m,e;\eta)\big|_{\eta\rightarrow 1} = \CK (P,Q;m,e)\;,
\\
&\CK_{\rm ref}(P, Q=\pm 1; m,e;\eta) = \CK (P, Q=\pm 1)\;.
\end{split}
\end{align}
Here $\CK (P, Q;m,e)$ denotes the (unrefined) Dehn filling kernel defined in \eqref{unrefined Dehn filling kernel}. 
The first identity can be proven by induction, while the second is immediate, since one can choose $\ell=1$ and $k_1 = P/Q$.

Except for $(P,Q)=(\pm 1,0)$, the refined kernel converges as a formal $q$-series.  For examples ($\bar{\eta} := \eta^{-1}$),
\begin{align}
\begin{split}
&\CK_{\rm ref}(1,2;0,0;\eta) = 1+(\eta^2-1)\,q
+(\bar\eta^2-2+\eta^4)\,q^2
+(\bar\eta^2-1-\eta^2+\eta^6)\,q^3 \\
&+(\bar\eta^2-2\eta^2+\eta^8)\,q^4
+(-\bar\eta^2+4-3\eta^2-\eta^4+\eta^{10})\,q^5 \\
&+(\bar\eta^4-3\bar\eta^2+5-2\eta^2-2\eta^4+\eta^{12})\,q^6 +(\bar\eta^4-6\bar\eta^2+9-\eta^2-3\eta^4-\eta^6+\eta^{14})\,q^7 \\
&+(3\bar\eta^4-9\bar\eta^2+7+4\eta^2-4\eta^4-2\eta^6+\eta^{16})\,q^8
+O(q^9)\,,
\\
&2\CK_{\rm ref}(1,2;2,0;\eta) =-1+(\bar\eta^2-\eta^2)\,q
+(1-\eta^4)\,q^2
+(1-\eta^6)\,q^3  \\
&+(-2\bar\eta^2+3-\eta^8)\,q^4+(\bar\eta^4-3\bar\eta^2+2+\eta^2-\eta^{10})\,q^5 +(\bar\eta^4-5\bar\eta^2+4+\eta^2-\eta^{12})\,q^6 \\
&+(3\bar\eta^4-6\bar\eta^2+4\eta^2-\eta^{14})\,q^7 +(2\bar\eta^4-4\bar\eta^2-2+4\eta^2+\eta^4-\eta^{16})\,q^8
+O(q^9)\,
\end{split}
\end{align}          
The refined kernel expression does not converge when $(P,Q)=(1,0)$. 
For example, for the choice $\vec{k}=(0,0)$, we find
\begin{align}
\begin{split}
&\CK_{\rm ref}(1,0;2,-2;\eta)
= \eta^{\infty}\Big[\,1
+(\eta^2-\bar\eta^2)\,q
+(-1-\bar\eta^2+\eta^2+\eta^4)\,q^2 \\
&\qquad+(\bar\eta^4-\bar\eta^2-2+\eta^4+\eta^6)\,q^3
+(\bar\eta^4-3-\eta^2+\eta^4+\eta^6+\eta^8)\,q^4
+O(q^5)\Big]\,.
\end{split}
\end{align}
As we enlarge the summation range, the overall prefactor diverges, behaving as $\eta^{N}$ with $N\to\infty$ (schematically denoted by $\eta^{\infty}$ above).
Nevertheless, because we choose $\vec{\gamma}$ to be a non-closable cycle \eqref{def:non-closable}, after performing the $(P,Q)=(1,0)$ surgery the refined index of the Dehn-filled manifold is expected to vanish.  
In other words, the would-be contribution is of the form $\eta^{\infty}\times 0=0$, and the refined index remains well-defined.

\section{Refined Index Calculator} \label{sec : Refined Index Calculator}
The computations presented in this paper can be performed easily using \texttt{Refined Index Calculator}\cite{soochang_lee_2026_19626481}, an open-source desktop application that implements the full refined 3d index pipeline. Triangulation data are loaded directly from the \texttt{SnapPy} census\cite{SnapPy}. The app is freely available at:
\begin{center}
    \url{https://physicsmasterpig.github.io/refined-index/}
\end{center}
The app consists of the following four main stages: Load, Refined Index, Dehn Filling, and Export. The detailed workflow across the whole pipeline is illustrated in Figure~\ref{fig:app workflow}.

\paragraph{Load.} The user selects a 3-manifold \(N\) together with a truncation \(N_{\max}\) on the magnetic charges \(m_i\). The application then constructs the Neumann--Zagier data \(g_{\rm NZ}\) in \eqref{NZ matrices}, searches for easy edges, and assembles the phase-space variables entering the refined index \eqref{refined index for N in charge basis}. A probe step inspects the on-disk cache and reports, for the current \(N\), whether a refined Dehn filling kernel \(\mathcal{K}_{\rm ref}\) of \eqref{refined Dehn filling kernel}, a precomputed refined index \(\mathcal{I}_{N;{\rm ref}}\), or a table of non-closable (NC) cycles \eqref{def:non-closable} is already available; cached artefacts are consumed transparently by the later stages.

\paragraph{Refined index.} The evaluation proceeds in one of three modes: a single-point \emph{Query} at specified external charges \((\vec m,\vec e)\), a \emph{Grid} sweep over a charge range, or a \emph{From cache} replay of a previously stored \(\mathcal{I}_{N;{\rm ref}}\) pack. Each component of the refinement vector \(\mathbf{v}\in\Gamma\bigl(\mathfrak{F}^{\rm hard}[N]\bigr)\) is set independently to \texttt{full}, \texttt{unrefined}, or a user-supplied \texttt{custom} weight \(W_{i}\); selecting \texttt{unrefined} on the \(i\)-th component projects that direction out of the fully refined result. The core routine then evaluates \(\mathcal{I}_{N;{\rm ref}}(\mathbf{v};\vec m,\vec e)\) as defined in \eqref{refined index for N in charge basis}.

\paragraph{Dehn filling.} Clicking \emph{Find NC cycles} launches a scan over every cusp of \(N\): for each cusp the application enumerates primitive pairs \((P,Q)\) and retains those for which \(\gamma=P_i\alpha_i+Q_i\beta_i\) is non-closable in the sense of \eqref{def:non-closable}. Each surviving cycle is immediately tested against the unrefined projection of the \(q^1\) coefficient onto the adjoint of \(SU(2)_{I}\) and flagged as \emph{marginal} whenever the criterion \eqref{Def : marginal NC} fails, which records in advance that the refined kernel \(\mathcal{K}_{\rm ref}\) must later be replaced by the unrefined kernel \(\mathcal{K}\) of \eqref{unrefined Dehn filling kernel} for that cycle. After this test, if the refined index is non-zero only for \(e_{n+I}=0\), then the marginal flag is cleared, consistent with the second clause of \eqref{Def : marginal NC}. The user then picks a cusp, an NC cycle from the flagged list, and a filling slope \((P_{i},Q_{i})\) in the original \((\alpha_i,\beta_i)\) basis; for multi-cusp manifolds each unfilled cusp is supplied either a single \emph{point} charge \((m,e)\) or a \emph{grid} range over which the result is tabulated, and per-cusp NC-cycle/slope pairs can be stacked for simultaneous multi-cusp filling. Upon submission, the application carries out the \((\alpha_{n+I},\beta_{n+I})\!\to\!(\gamma_{n+I},\delta_{n+I})\) basis change via the \(\mathrm{SL}(2,\mathbb{Z})\) transformation, transports the user slope to \((p,q)\) in the new basis, updates \(g_{\rm NZ}\) and the affine shifts \((\mathbf{a},\mathbf{b})\) of \eqref{mixing : a}--\eqref{mixing : b}, and extracts \((\mathbf{a},\mathbf{b})\) from the basis-changed \(\mathcal{I}_{N;{\rm ref}}\) on a probe grid. Each hard-edge direction \(j\) is then audited against the Dehn filling compatibility conditions \eqref{Dehn filling compatibility}: the refinement \(W_{j}\) is kept active only when \(a^{(j)}_{I}\in\mathbb{Z}\) and \(b^{(j)}_{I}\in\mathbb{Z}/2\), and the remaining directions are projected to \(W_{j}=0\) with the matching entries of \((\mathbf{a},\mathbf{b})\) zeroed so that no shift re-introduces them. The refined projection of the \(q^1\) coefficient onto the adjoint of \(SU(2)_{I}\) is then tested against \(-1\); as in the unrefined case, a non-zero index only for \(e_{n+I}=0\) flips the flag to pass. The Hirzebruch--Jung continued fraction \eqref{vec-k from P/Q} then yields \(\vec{k}=(k_{1},\ldots,k_{\ell})\), and the refined kernel \(\mathcal{K}_{\rm ref}(P,Q;m,e;\eta)\) is assembled as the chain of \(\mathcal{I}_{S}\) symplectic kernels in \eqref{refined Dehn filling kernel}.

\paragraph{Export.} Once any previous stage has populated the session, the final card collects the available artefacts---the Neumann--Zagier data \(g_{\rm NZ}\), the refined index \(\mathcal{I}_{N;{\rm ref}}(\mathbf{v};\vec m,\vec e)\) of \eqref{refined index for N in charge basis}, the NC-cycle table, and the Dehn filled refined index \(\mathcal{I}_{M;\textrm{ref}}(\mathbf{v};\vec{m},\vec{e})\) of \eqref{refined index in charge basis}---and writes them to disk in a user-selected subset of formats (LaTeX, Mathematica, JSON). 
\newpage

\begin{figure}[t]
    \centering
    \includegraphics{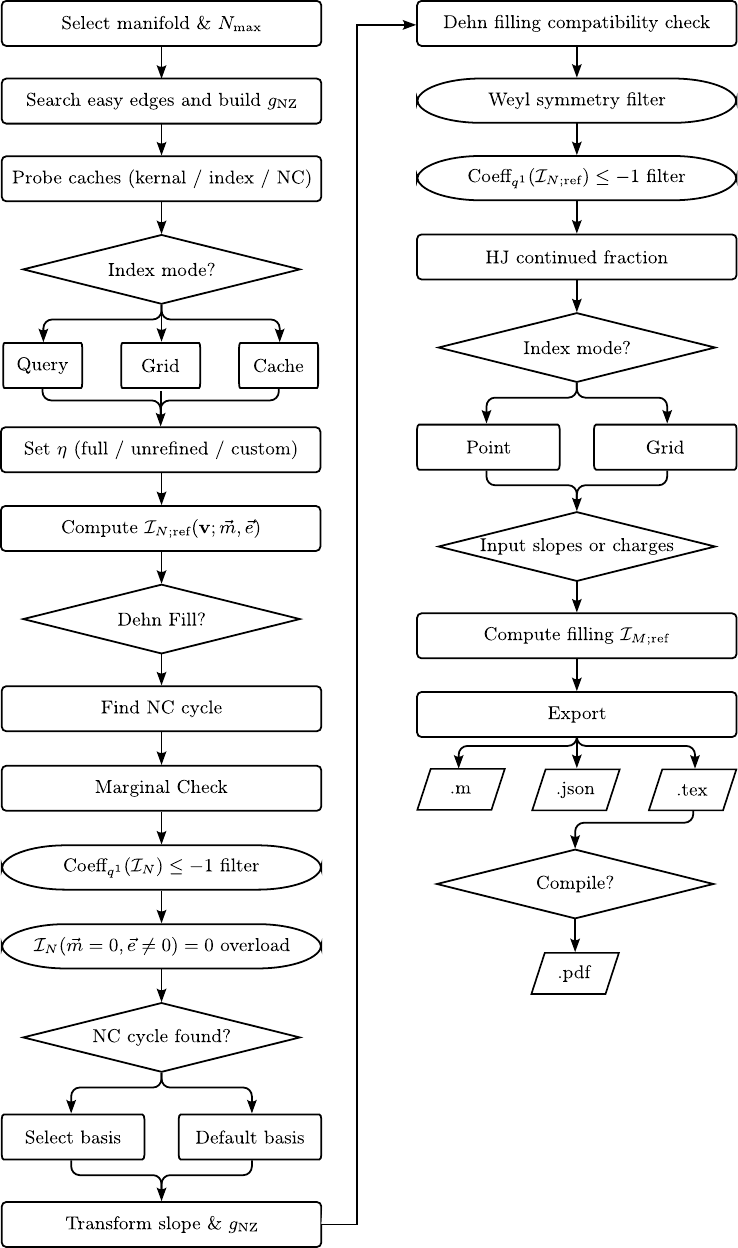}
    \caption{Refined index calculator workflow diagram. }
    \label{fig:app workflow}
\end{figure}

\clearpage

\bibliographystyle{ytphys}
\bibliography{ref}
\end{document}